\documentclass{ut-thesis}
\usepackage[backend=bibtex,style=numeric,maxbibnames=99]{biblatex} 
\usepackage[colorlinks]{hyperref} 
\usepackage{graphicx} 
\usepackage{booktabs} 
\usepackage{amsthm}
\usepackage{amsmath}
\usepackage{amsfonts}
\usepackage{mathtools}
\usepackage{cleveref} 
\usepackage{algorithmic}
\usepackage[ruled,lined,boxed]{algorithm2e}
\usepackage{thmtools}
\usepackage{wasysym}
\usepackage{multirow}
\usepackage{amssymb}

\usepackage{tikz}
\usetikzlibrary{graphs}

\Crefname{algocf}{Algorithm}{Algorithms}

\PassOptionsToPackage{frozencache,cachedir=.}{minted}
\usepackage[most]{tcolorbox}

\newtcblisting{myminted}{%
    listing engine=minted,
    minted language=python,
    listing only,
    breakable,
    enhanced,
    minted options = {
            linenos,
            breaklines=true,
            breakbefore=.,
            fontsize=\footnotesize,
            numbersep=2mm
        },
    overlay={%
            \begin{tcbclipinterior}
                \fill[gray!25] (frame.south west) rectangle ([xshift=4mm]frame.north west);
            \end{tcbclipinterior}
        }
}

\BeforeBeginEnvironment{minted}
{\begin{tcolorbox}[breakable, enhanced]}

        \AfterEndEnvironment{minted}
        {\end{tcolorbox}}

\hypersetup{
    colorlinks=true,
    citecolor=blue
}

\allowdisplaybreaks
\newcommand{\Sym}{\mathcal{S}}
\newcommand{\Y}{\mathcal{Y}}
\newcommand{\V}{\mathcal{V}}
\newcommand{\X}{\mathcal{X}}
\newcommand{\J}{\mathcal{J}}
\newcommand{\Q}{\mathcal{Q}}

\newcommand{\Z}{\mathcal{Z}}
\newcommand{\binstrs}{\{0, 1\}^\star}

\newcommand{\M}{\mathcal{M}}
\newcommand{\1}{\mathbf{1}}
\newcommand{\N}{\mathcal{N}}

\DeclareMathOperator*{\suchthat}{s.t.}
\DeclareMathOperator*{\argmax}{argmax}

\newcommand{\emptystr}{\varepsilon}

\theoremstyle{plain}
\newtheorem{theorem}{Theorem}[section]

\theoremstyle{definition}
\newtheorem{definition}[theorem]{Definition}
\newtheorem{lemma}[theorem]{Lemma}
\newtheorem{remark}[theorem]{Remark}
\newtheorem{crv}[theorem]{CRV}

\theoremstyle{remark}
\newtheorem{example}[theorem]{Example}

\DeclarePairedDelimiter\abs{\lvert}{\rvert}  
\DeclarePairedDelimiter{\ceil}{\lceil}{\rceil}
\newcommand\g{\,\vert\,}  

\newcommand{\encode}{\mathtt{encode}}
\newcommand{\decode}{\mathtt{decode}}
\newcommand{\flookup}{\mathtt{forward\_lookup}}
\newcommand{\rlookup}{\mathtt{reverse\_lookup}}
\newcommand{\bstinsert}{\mathtt{insert}}
\newcommand{\bstremove}{\mathtt{remove}}

\newcommand{\Naturals}{\mathbb{N}}

\newcommand{\Integers}{\mathbb{Z}}

\newcommand{\multiset}[1]{\{\!\!\{#1\}\!\!\}}

\DeclarePairedDelimiterX{\Ex}[1]{[}{]}{#1}
\DeclareMathOperator*{\ExOp}{\mathbb{E}}
\newcommand{\E}{\ExOp\Ex}

\DeclarePairedDelimiterX{\KLx}[2]{(}{)}{%
    #1\;\delimsize\|\;#2%
}
\DeclareMathOperator*{\KLOp}{KL}
\newcommand{\KL}{\KLOp\KLx}

\DeclarePairedDelimiterX{\ELBOx}[1]{(}{)}{#1}
\DeclareMathOperator*{\ELBOOp}{ELBO}
\newcommand{\ELBO}{\ELBOOp\ELBOx}

\DeclareMathOperator{\A}{\mathcal{A}}

\DeclareMathOperator{\bin}{\mathtt{bin}}
\newcommand{\B}{\mathcal{B}}

\newcommand{\eclass}[1]{[#1]}

\newcommand{\quotient}[2]{{#1}/{#2}}

\newcommand{\defeq}{\coloneqq}
\newcommand{\SetOrders}{\mathcal{S}_n(\mathcal{\tilde{X}})}
\newcommand{\SetDataSet}{\mathcal{\tilde{X}}}
\newcommand{\changed}[1]{#1}
\author{Daniel Severo}
\title{Random Permutation Codes: \\ Lossless Source Coding of Non-Sequential Data}
\degree{Doctor of Philosophy}
\department[]{Graduate Department of Electrical and Computer Engineering}
\gradyear{2025}
\begin{document}
\frontmatter
\maketitle
\begin{abstract}
  This thesis deals with the problem of communicating and storing non-sequential data.
We investigate this problem through the lens of lossless source coding, also sometimes referred to as lossless compression, from both an algorithmic and information-theoretic perspective.

Lossless compression algorithms typically preserve the ordering in which data points are compressed.
However, there are data types where order is not meaningful, such as collections of
files, rows in a database, nodes in a graph, and, notably, datasets in machine learning applications.

Compressing with traditional algorithms is possible if we pick an order for the elements
and communicate the corresponding ordered sequence.
However, unless the order information is somehow removed during the encoding process, this procedure will be sub-optimal, because the order contains information and therefore more bits are used to represent the source than are truly necessary.

In this work we give a formal definition for non-sequential objects as random sets of equivalent sequences, which we refer to as \emph{Combinatorial Random Variables} (CRVs).
The definition of equivalence, formalized as an equivalence relation, establishes the non-sequential data type represented by the CRV.
The achievable rates of CRVs is fully characterized as a function of the equivalence relation as well as the data distribution.

The optimal rates of CRVs are achieved within the family of \emph{Random Permutation Codes} (RPCs) developed in later chapters.
RPCs randomly select one-of-many possible sequences that can represent the instance of the CRV.
The selection is done through sampling with bits-back coding \cite{townsend2019practical, frey1996free} and asymmetric numeral systems \cite{duda2009asymmetric}, and guarantees the achievability of the optimal rate.

Specialized RPCs are given for the case of multisets, graphs, and partitions/clusterings, providing new algorithms for compression of databases, social networks, and web data in the JSON file format.
The computational and memory complexity of RPCs is discussed and shown to be attractive for the applications considered.
\end{abstract}
\begin{acknowledgements}
  To my mother, father, and sister, who, despite having nothing, gave me everything.
  To my brother, who makes me proud of the person he is becoming.
  To my wife, for all her sacrifice, patience, and partnership, during these times of change.
  Finally, to my advisors, friends, and collaborators, for pointing me in the right direction.
\end{acknowledgements}
\tableofcontents
\listoftables
\listoffigures
\listoftheorems[ignoreall, show={theorem,lemma,remark}, swapnumber, title={List of Theorems, Lemmas, and Remarks}]
\listoftheorems[ignoreall, show={definition}, swapnumber, title=List of Definitions]
\listoftheorems[ignoreall, show={example}, swapnumber, title=List of Examples]
\mainmatter
\chapter*{Notation}

\paragraph{Random Variables, Distributions, and Entropy}
\begin{enumerate}
    \item In most cases, random variables are represented by capital letters such as $X,Y,Z$, while their instances are lower case $x,y,z$. Some exceptions are made and will be clearly specified in the text. For example, random permutations and their instances will both be represented by either $\sigma$ or $\pi$.
    \item The alphabet of a random variable $X$ is denoted with $\mathcal{X}$, the calligraphic version of the same symbol. Exceptions are made clear in the text with the most important being for multiset-valued random variables.
    \item A sequence $X_1, \dots, X_n$ of size $n$ is abbreviated as $X^n$.
    \item All distributions $P_{(\cdot)}$ are discrete and will be sub-indexed by their respective random variables, such as $P_X$ and $P_{X\g Y}$.
    \item Entropy of a random variable $X \sim P_X$ is denoted by both $H(X)$ and $H(P_X)$.
    \item Given two distributions $P$ and $Q$, over $x \in \X$ and $y \in \Y$, the product distribution $P \cdot Q$ over $\X\times\Y$ assigns probabilities $(PQ)(x, y) = P(x)\cdot Q(y)$.
    \item \changed{$\E{\log P_X}$ is a shorthand for $\E{\log P_X(X)} = \sum_{x \in \X} P_X(x)\cdot \log P_X(x)$. Similarly, $\E{\log P_{Z \g X}} = \sum_{(x, z) \in \X \times \Z} P_{X, Z}(x, z) \cdot \log P_{X, Z}(x, z)$.}
    \item $\KL{P_Z}{P_{Z \g X}(\cdot \g x)}$ is the KL divergence with second argument the conditional distribution $Q_Z(z) = P_{Z \g X}(z \g x)$, which is a function of $x$.
    \item When dealing with collections of i.i.d.\ random variables $X_i \sim P$ we will sometimes drop the subscript and write $X \sim P$.
\end{enumerate}

\paragraph{Multisets}
\begin{enumerate}
    \setcounter{enumi}{9}
    \item The number of elements contained in a multiset $\M$, including repetitions, is sometimes referred to as the ``size of $\M$" and is denoted by $\abs{\M}$.
    \item When the elements of a multiset $\M$ are the same as those in a sequence $X^n$, we write $\M = \{X_1, \dots, X_n\} = \multiset{X^n}$.
    \item The number of elements in $\M = \{x_1, \dots, x_n\}$ that are equal to some $x \in \X$ is denoted as $\M(x) = \sum_{i=1}^n \1\{x = x_i\}$.
    \item Given a multiset $\M$ and an element $x$, the multiset $\M \setminus \{x\}$ has the same elements as $\M$ but with the occurrence count of element $x$ decreased by $1$ (unless the occurrence count was $1$, then $x$ is removed from $\M$).
    \item Given two multisets $\M$ and $\M'$, we define $\M \setminus \M'$ to be the resulting multiset after decreasing the occurrences of elements in $\M$ by their occurrence counts in $\M'$.
    \item Occurrence counts cannot be negative, hence $\{x\} \setminus \{x, x\}$ results in the empty multiset $\emptyset$.
    \item Similarly, $\M \cup \M'$ denotes the additive union of two multisets.
\end{enumerate}

\paragraph{Graphs}
\begin{enumerate}
    \setcounter{enumi}{16}
    \item All graphs in this work are labeled, have a fixed number of nodes ($n$), a variable number of edges ($m$), and are in general non-simple (i.e., allow loops and repeated edges), unless mentioned otherwise.
    \item A sequence of graphs on $n$ nodes, where an edge is added at each step, can be represented as a sequence of vertex elements $v_i$, taking on values in the vertex set $\V$, with the $i$-th edge defined as $e_i = (v_{2i-1}, v_{2i})$.
          The $i$-th graph is taken to be $G_i = \{\{v_1, v_2\}, \dots, \{v_{2i-1}, v_{2i}\}\}$ and the $k$-th vertex of an edge is indicated by $e[k]$.
\end{enumerate}
\paragraph{Coding with Asymmetric Numeral Systems (ANS)}
Given a quantized probability distribution (\Cref{def:quantized-pmf}) $Q_X(x) = \frac{q_x}{N}$,
\begin{enumerate}
    \setcounter{enumi}{18}
    \item Encoding with ANS is denoted by both $\encode(s, x, Q_X)$; and ``$\encode(s, x)$ with $Q_X$".
    \item Decoding with ANS is denoted by both $\decode(s, Q_X)$; and ``$\decode(s)$ with $Q_X$".
\end{enumerate}

\paragraph{Miscellaneous}
\begin{enumerate}
    \setcounter{enumi}{20}
    \item The ascending factorial function $a: \mathbb{R}\times\mathbb{N} \mapsto \mathbb{R}$ is defined by $a(x, k) = x(x+1)(x+2)\dots(x+k-1)$, for $k > 0$, with $a(x, 0) = 1$ for all $x$, and is abbreviated as $x^{\uparrow k}$.
    \item $\defeq$ denotes equality by definition.
    \item \changed{$[n] \defeq \{1, \dots, n\}$ and $[n) \defeq \{0, \dots, n-1\}$ denote intervals of integers.}
    \item The symbol $\equiv$ denotes equivalence between mathematical objects under the given context.
    \item Rate (\Cref{def:rate}) and asymptotic rate (\Cref{def:asymptotic-rate}) will sometimes be multiplied by the sequence length, but will still be referred to by their original names, when clear from context.
    \item \changed{The set of all finite-length binary strings is $\binstrs \defeq \{\varepsilon, 0, 1, 00, 01, \dots\}$, where $\varepsilon$ is the empty string.}
    \item Logarithms are always base $2$.
\end{enumerate}
\chapter*{Preface}
\changed{
    This thesis attempts to study, and propose algorithms for, the communication and storage of \emph{non-sequential data} using as few bits of information as possible.
    To illustrate, consider the problem of representing a set of $n$ distinct elements, from a total of $m$ possible candidates, on a modern digital computer.
    A common solution is to represent the set as a sequence of its elements, such as in an array or list, in decreasing or increasing order (i.e., sorted).
    However, a set is a mathematical object void of any ordering between elements: any of the $n!$ possible orderings of the sequence would be valid representatives.
    One can therefore hypothesize a communication protocol where only the first $n-1$ elements are stored, while each of the $m - (n-1)$ possible values for the last element are mapped to one of the $(n-1)!$ orderings of the first $n-1$ elements.
    The value of the last element need not be stored, and can be deduced directly from the ordering of the first $n-1$ elements, assuming the map is known.
    This example illustrate the key observation underlying this thesis: \textbf{a non-sequential object can be stored using less bits than what is required to represent an arbitrary sequence of its elements; by encoding information in the order between elements in the sequence.}

    Throughout this manuscript we will identify different types of non-sequential objects including a variety of graphs, generalizations of sets (i.e., multisets), partitions, clusters, and certain families of permutations of integers.
    For each we will establish the minimal number of bits required to store such objects through the lens of information theory \cite{shannon1948mathematical,cover1999elements}.
    We will show all these objects can be unified under a common framework (\Cref{chapter:crv-rpc}) allowing us to develop computationally efficient compression algorithms for these data types.

}

Encoding information in the order between elements of a sequence is second-nature to human beings.
In scholarly writings, the expected contribution of an individual is communicated to the academic community via their position in the author-list.
When viewing the results of a competition, we expect the ordering to encode the ranking of players.
Changing the order between words in a sentence can drastically alter the overall meaning (at least in some languages).

There are situations where we would like to avoid having to explicitly choose a specific ordering.
For example, an author-list where all persons contributed equally, or the listing of players in a team.
Similarly, politicians of bilingual countries, e.g., Canada, often go through great lengths during public speeches to avoid giving preferential treatment to a single language, by switching back and forth between both languages.

To remove information from the ordering it is common to decide the order through a random coin toss.
The hope of the transmitter is that the receiver, knowing the order was decided randomly, will be less likely to extract meaning from the order between elements.
The decision of the sequence order is outsourced to a source of randomness not controlled by the transmitter.

The solution proposed in the previous paragraph, as well as the family of algorithms developed in this thesis, \emph{Random Permutation Codes}, makes use of the following fact defining the initial seed of our work: \textbf{from an information-theoretic point of view, a randomly ordered sequence carries the same information as a non-sequential collection of the same elements.}

\vspace{2em}

\changed{
    \paragraph*{Outline} The following is a summary of each paragraph intended to inform the reader of what to expect in each chapter.
    \begin{itemize}
        \item[\Cref{chapter:lossless-source-coding}]
              begins by defining what is a source code and shows a greedy algorithm is the optimal code minimizing the number of bits needed to represent any data source, but is infeasible to use in practice.
              An alternative, and computationally less demanding, family of codes, known as Prefix-Free Codes (\Cref{def:extended-codes}), are discussed and shown to achieve the same rate as the optimal code asymptotically with the sequence length.
              As a highlight, we give an example of how the greedy code can achieve lossless compression rates below the entropy of the source.
              The chapter concludes by briefly discussing the optimal code within the Prefix-Free family, Huffman Codes, which are known to be sub-optimal up to $1$ bit for each symbol in the sequence.

        \item[\Cref{chapter:entropy-coding-with-ans}]
              develops the theory of coding with probability models, referred to as ``entropy coding'', via asymmetric numeral systems (ANS) \cite{duda2009asymmetric}: a last-in-first-out compression algorithm.
              The achievable rates for compression with ANS are characterized in the regime where the state of ANS is large.
              We show ANS can be used as an ``invertible sampler'', by performing a decode operation with the desired probability distribution.
              The chapter concludes with a discussion on bits-back or free-energy coding \cite{frey1996free,townsend2019practical}, the mechanism that allows us to store information in the order between elements in further chapters.

        \item[\Cref{chapter:combinatorial-objects}]
              discusses a myriad of non-sequential objects and establishes notation for further chapters.
              This chapter also serves as a review of basic mathematical concepts such as equivalence relations (\Cref{def:equivalence-relation}), total orderings (\Cref{def:total-order}), permutations on multisets (\Cref{def:permutations-on-multisets}), and others.

        \item[\Cref{chapter:roc}]
              studies compression of sets and multisets and proposes an optimal compression algorithm, \emph{Random Order Coding} (ROC), that is quasi-linear in the number of elements. ROC exploits ANS as an invertible sampler to perform sampling without replacement, implicitly encoding information in the ordering between elements (as discussed previously at the start of this preface).
              The chapter discusses in detail the data structure developed which allows the algorithm to execute in quasi-linear time (a modified binary search tree).
              A discussion is provided on the connection between multisets, method of types \cite{cover1999elements}, and universal source coding of independent and identically distributed symbols.
              We conclude with a variety of experiments on synthetic multisets, sets of images, as well as sets of sets of text (i.e., JSON maps).

        \item[\Cref{chapter:rcc}]
              deals with compression of partitions of arbitrary sets or, equivalently, clustering of data points.
              The algorithm introduced, \emph{Random Cycle Coding} (RCC), encodes information in the ordering between elements, similar to ROC.
              However, RCC exploits the cycle structure of permutations to define a partition of the elements in the set.
              We show this procedure optimally compresses partitions which have probability proportional to the product of the number of elements in each subset of the partition.
              The chapter concludes with a focus on experiments for the particular application of similarity search with vector databases such as FAISS \cite{johnson2019billion}, as well as synthetic data.

        \item[\Cref{chapter:rec}]
              discusses the fundamental limits on graph compression for the undirected and direct case, as well as the generalization of graphs known as hyper-graphs.
              It is shown that the non-sequential nature of graphs comes from the freedom of permuting edges, as well as vertices within an edge when edges are undirected.
              An algorithm is given, named \emph{Random Edge Coding} (REC), achieving the optimal rate.
              The computational complexity of REC is quasi-linear in the number of edges present in the graph, making it efficient for sparse graphs, where the number of edges is significantly smaller than the total number of possible edges.
              The chapter shows how to leverage Pólya's Urn model \cite{mahmoud2008polya} as a probabilistic model over graphs to achieve competitive compression performance on graphs with millions of vertices and billions of edges.

        \item[\Cref{chapter:crv-rpc}]
              concludes the thesis by giving a general definition for non-sequential objects as a random variable with alphabet equal to equivalence classes over sequences, under some given equivalence relation.
              Changing the definition of the equivalence relation allows us to model different non-sequential objects using the same framework.
              We refer to random variables created via this mechanism as \emph{Combinatorial Random Variables} (CRVs, \Cref{def:combinatorial-random-variables}), and establish the fundamental limits on compression of these data types.
              A broad family of algorithms is introduced which achieves the optimal compression rate of CRVs, and generalizes ROC, RCC, and REC, which we call \emph{Random Permutation Codes} (RPCs).
              The chapter concludes with a discussion on suggestions for future work.
    \end{itemize}
}

\paragraph*{Contributions}
The contributions of this thesis are detailed in the following paragraph.
Claims of novelty are supported by the literature review presented in each chapter, but are inevitably to the best of our knowledge.

\Cref{chapter:lossless-source-coding} reviews lossless source coding.
All results are known, but some parts of the exposition is novel.
In particular, the characterization of the gap between the optimal and prefix-free codes, shown in  \Cref{example:optimal-rate-for-a-uniform-source}, as well as the centering of the discussion around optimal, and greedy, non-prefix-free codes.

\Cref{chapter:entropy-coding-with-ans} closely follows the structure of Chapter 5 in \cite{cover1999elements}, but differs in that it establishes source coding theorems with the use of asymmetric numeral systems \cite{duda2009asymmetric}.
The characterization of the achievable ANS rates (\Cref{theorem:source-coding-ans}), through the concept of the large state regime (\Cref{theorem:optimality-ans}), are known in practice \cite{duda2009asymmetric}, but a formal characterization has not been given in the literature.
The same applies to the results on correctness of sampling via ANS decoding (\Cref{lemma:ans-sampling}).
The characterization of the bits-back coding rates (\Cref{lemma:bb-ans-rate}) is known \cite{townsend2019practical}, but the formalization is new.

\Cref{chapter:roc},  \Cref{chapter:rcc}, and \Cref{chapter:rec} are entirely novel and first appeared in the following published manuscripts:
\changed{
    \begin{enumerate}
        \item[\cite{severo2021your}] Daniel Severo, James Townsend, Ashish J Khisti, Alireza Makhzani, and Karen Ullrich. ``Your dataset is a multiset and you should compress it like one”. In: \emph{NeurIPS 2021 Workshop on Deep Generative Models and Downstream Applications}. 2021
        \item[\cite{severo2023compressing}] Daniel Severo, James Townsend, Ashish J Khisti, Alireza Makhzani, and Karen Ullrich. "Compressing multisets with large alphabets." \emph{IEEE Journal on Selected Areas in Information Theory} 3, no. 4 (2022): 605-615.
        \item[\cite{severo2023random}] Daniel Severo, James Townsend, Ashish J Khisti, and Alireza Makhzani. ``One-Shot Compression of Large Edge-Exchangeable Graphs using Bits-Back Coding”. In: \emph{International Conference on Machine Learning}. PMLR. 2023, pp. 30633–-30645
    \end{enumerate}
}
as well as the following work currently under review:
\changed{
    \begin{itemize}
        \item Daniel Severo, Ashish J Khisti, and Alireza Makhzani. ``Random Cycle Coding: Lossless Compression of Cluster Assignments via Bits-Back Coding". Submitted to \emph{Advances in Neural Information Processing Systems}, 2024.
    \end{itemize}
}
The results in the final chapter, \Cref{chapter:crv-rpc}, are also novel, and debut in this thesis.
\chapter{Lossless Source Coding}\label{chapter:lossless-source-coding}
This chapter reviews the central topic of this thesis: the design of algorithms for representing data, with perfect fidelity, in digital media.
This problem is formulated following the developments of Shannon \cite{shannon1948mathematical} into what is now known as \emph{Information Theory}, and, more specifically in our case, \emph{Source Coding}.

In this classic setting, data and information are viewed through a statistical and probabilistic lens, formalized as a random variable with a known discrete probability law.
The amount of information carried by this mathematical object is completely defined by the probability law, being completely agnostic to the semantics of the data itself.
Under this theory, the resources required to store and transmit text, numerical quantities, high-resolution images and videos, or any other arbitrary data, are equal, as long as the probability law governing their appearance in our observations is the same.

A common misconception amongst practitioners is that the average number of bits required to store a single observation is lower bounded by a quantity known as the \emph{entropy}.
This is true only on average and for an infinite number of samples.
It is possible to communicate a sample from a data source using less bits than the entropy, on average, when the number of samples is finite, as shown in \Cref{subsec:rates-below-entropy-are-achievable}.

\section{Codes}
Let $X^n \sim P_{X^n}$ be a sequence of discrete random variables, with alphabet $\X^n$, representing a random sample of the data source.
The objective of lossless source coding is to find a binary representation for the elements of $\X^n$, under a set of given constraints, such that any $x^n \in \X^n$ can be recovered from its representation unambiguously and without loss of information.
This is achieved through the use of a \emph{lossless source code}.

\begin{definition}[Lossless Source Code]\label{def:code}
    A \emph{lossless source code}, or simply \emph{code}, is a bijection $C:\X^n \mapsto \binstrs$.
    The \emph{codeword} of $x^n \in \X^n$ is the image $C(x^n) \in \binstrs$.
    The \emph{codebook} is the set of codewords $C(\X^n) \subset \binstrs$.
    The number of bits in the binary string $C(x^n)$ is called the \emph{code length} of $x^n$ and is abbreviated as $\ell_{x^n}$ when the code is clear from context.
\end{definition}

\emph{Lossless compression} is concerned with finding codes that assign codewords of small lengths for the purpose of storing and transmitting data.
The restriction on length can take on many forms with the most common being minimizing the average- or worst-case length.
In this thesis we will restrict our discussion to the average case, defined by the following quantity which we refer to as the \emph{rate}.
\begin{definition}[Rate]\label{def:rate}
    Given a code $C$, and data source $(X^n, \X^n, P_{X^n})$, the average number of bits required to communicate an instance of $x^n \in \X^n$, drawn with probability $P_{X^n}(x^n)$, is the expected code length of $C$ under $P_{X^n}$, normalized by sequence length,
    \begin{align}\label{eq:bitrate}
        R(C, P_{X^n}) \defeq \frac{1}{n}\cdot\E{\ell_X} = \sum_{x^n \in \X^n}P_{X^n}(x^n)\cdot\ell_{x^n}.
    \end{align}
\end{definition}

\begin{definition}[Asymptotic Rate]\label{def:asymptotic-rate}
    For the same setup as \Cref{def:rate}, \changed{assuming $P_{X^n}$ is defined for every $n \in \Naturals$}, the asymptotic rate is the limit of the rate for increasing $n$,
    \begin{align}
        \lim_{n \rightarrow \infty} R(C, P_{X^n}).
    \end{align}
\end{definition}

\subsection{Optimal Codes}
A lossless code assigns a binary string to each element of $\X^n$.
For any code $C$, if there exists an element $s \in \binstrs$ which is not present in the codebook, $s \notin C(\X^n)$, and has length smaller than some code-word, $\abs{s} < \ell_{x^n}$, then redefining $C(x^n) \defeq s$ will decrease the rate.
The recursive application of this argument implies that a code achieving the smallest rate must assign code-words sequentially based on length, with ties broken arbitrarily.
Codes achieving the smallest rate for a given data source are referred to as \emph{optimal codes}.
\begin{definition}[Optimal Code]\label{def:optimal-code}
    Given a data source $(X^n, \X^n, P_{X^n})$, an \emph{optimal code} is the solution to the following optimization problem
    \begin{align}\label{opt:optimal-code}
        \begin{split}
            \min_C    & \quad R(C, P_{X^n})                  \\
            \suchthat & \quad C \text{ is a code for } \X^n.
        \end{split}
    \end{align}
    The following construction is a solution to \eqref{opt:optimal-code} for \emph{any} data source.
    Sort the elements of $\X^n$ by their probabilities under $P_{X^n}$, in descending order, breaking ties arbitrarily.
    Assign code words sequentially, ordered by their lengths, starting from the most probable sequence, $\argmax_{x^n \in \X^n} P_{X^n}(x^n)$.
\end{definition}
\begin{example}[Optimal Codes]\label{example:optimal-codes}
    For $\X = \{a, b, c, d\}$ and $n=1$; $C_1$, $C_2$, and $C_3$, are optimal codes.
    \changed{
        \begin{align}
            P_X(a) & = 0.5 & C_1(a) & = \varepsilon & C_2(a) & = \varepsilon & C_3(a) & = \varepsilon \\
            P_X(b) & = 0.3 & C_1(b) & = 0           & C_2(b) & = 0           & C_3(b) & = 1           \\
            P_X(c) & = 0.3 & C_1(c) & = 1           & C_2(c) & = 1           & C_3(c) & = 0           \\
            P_X(d) & = 0.1 & C_1(d) & = 00          & C_2(d) & = 10          & C_3(d) & = 11
        \end{align}
    }
\end{example}

\subsection{Sequential Codes}
\begin{definition}[Sequential Code]
    A \emph{sequential code} is a collection of functions,
    \begin{align}
        C_i \colon \X^{i-1} \times \X \mapsto \binstrs,
    \end{align}
    where $C_i^\prime (x) \defeq C_i(x^{i-1}, x)$ is a lossless code over $\X$ for any sequence $x^{i-1}$.
\end{definition}

\begin{example}[Optimal Sequential Codes]\label{example:optimal-code-seq}
    $C_n$ in \Cref{table:example-optimal-code-seq} is an optimal code over $\X = \{a, b\}$ for any $n \in \Naturals$, when symbols are independent and identically distributed, and $P_X(a) \geq P_X(b)$,
    \begin{table}[h!]
        \centering
        \changed{
            \begin{tabular}{llllllllll}
                \toprule
                $C_n(x)$ & $\varepsilon$ & $0$   & $1$   & $00$  & $01$  & $10$  & $11$  & $000$ & $\dots$ \\
                \midrule
                $n=1$    & $a$           & $b$   &       &       &       &       &       &       &         \\
                $n=2$    & $aa$          & $ab$  & $ba$  & $bb$  &       &       &       &       &         \\
                $n=3$    & $aaa$         & $aab$ & $aba$ & $baa$ & $abb$ & $bab$ & $bba$ & $bbb$ &         \\
                $\dots $ &               &       &       &       &       &       &       &       &         \\
                \bottomrule
            \end{tabular}
        }
        \caption{Optimal code from \Cref{example:optimal-code-seq}}
        \label{table:example-optimal-code-seq}
    \end{table}
\end{example}

An optimal sequential code can be implemented via a pair of encoding and decoding functions that specify the next codeword as a function of the current codeword, the incoming symbol, and the number of symbols seen so far $i \in [n]$.
The signatures of the pair of functions are
\begin{align}\label{eq:codes-encode-decode}
    \encode & \colon \binstrs \times \X \times [n]  \mapsto \binstrs,           \\
    \decode & \colon \binstrs \times [n]            \mapsto \X \times \binstrs.
\end{align}
\begin{example}[Encode and Decode functions]
    For the table of \Cref{example:optimal-code-seq} the steps taken to encode $bab$, with the step argument omitted for clarity, are
    \changed{
        \begin{align*}
            \encode(\emptystr, b) & = 0  & \decode(10) & = (b, 1)          \\
            \encode(0, a)         & = 1  & \decode(1)  & = (a, 0)          \\
            \encode(1, b)         & = 10 & \decode(0)  & = (b, \emptystr),
        \end{align*}
    }
    where $\emptystr$ is the empty string.
\end{example}
As it stands, encoding and decoding must be done via table or dictionary lookups.
For each $i \in [n]$ the dictionary contains $\abs{\X}^i$ key-value pairs totalling,
\begin{align}
    \sum_{i=1}^n \abs{\X}^i
    = \frac{\abs{\X}\cdot\left(\abs{\X}^n -1 \right)}{\abs{\X} - 1}
    = \Omega\left(\abs{\X}^n \right),
\end{align}
which quickly becomes infeasible for real-world applications.

\subsection{Extended Codes}
Optimal codes are the solution to \eqref{opt:optimal-code} and provide the shortest description, on average, for any discrete random variable, but with complexity that scales exponentially in the sequence length $n$.
We therefore search for a solution with lower complexity by restricting the family of possible codes.
This leads to an upper bound in the objective of $\eqref{opt:optimal-code}$.
Under certain conditions, the following family can define codes over sequences $\X^n$ through a base code over $\X$, referred to as a \emph{symbol} code, which can be implemented with dictionaries of size $\Omega(\abs{\X})$.
\begin{definition}[Extended Codes]\label{def:extended-codes}
    Given a symbol code $C\colon \X \mapsto \binstrs$, its \emph{extension}, $C \colon \X^n \mapsto \binstrs$, is defined as the concatenation of codewords,
    \begin{align}
        C(x^n) \defeq C(x_1)C(x_2) \dots C(x_n).
    \end{align}
    The extension is called an \emph{extended code} if it is a valid code (\Cref{def:code}) for $\X^n$.
    The symbol code is said to be \emph{uniquely decodable} if its extension is a code.
    Extended codes are sequential codes.
\end{definition}
\begin{example}[Fixed-length Codes]
    The following code over $\X = [8]$ is uniquely decodable.
    \begin{align}
        C(0) & = 000 & C(4) & = 100 \\
        C(1) & = 001 & C(5) & = 101 \\
        C(2) & = 010 & C(6) & = 110 \\
        C(3) & = 011 & C(7) & = 111
    \end{align}
    The codewords can be recovered by partitioning the extended codeword into $n$ contiguous substrings of length $3$. For example,
    \begin{align}
        \text{Encoding: } & 263 \rightarrow C(263) = C(2)C(6)C(3) = 010110011,                          \\
        \text{Decoding: } & 010110011 \rightarrow 010,110,011 \rightarrow C(2)C(6)C(3) \rightarrow 263.
    \end{align}
    Equality between codeword lengths of the symbol code is sufficient for unique decodability.
    These are known as \emph{fixed-length} codes.
\end{example}
\begin{example}[Non-code]\label{example:extension-not-a-code}
    For $C_2$ of \Cref{example:optimal-code-seq}, and $n=3$, both $acd$ and $cad$ map to $11110$ implying the extension is not a bijection and therefore not a code.
\end{example}

\subsection{Prefix-free Codes}
We can transform any extension to a valid code by augmenting $\X^n$ to introduce a separation symbol.
Interleaving the separation symbol with the concatenated codewords from the symbol code would indicate where one codeword ends and another begins.
The same can be achieved by requiring that no codeword be a prefix of another in the symbol code.
\begin{definition}[Prefix-free Codes]\label{def:prefix-free-codes}
    A symbol code over $\X$ is \emph{prefix-free} if no codeword is a prefix of another.
    The extension of a prefix-free code can be decoded by traversing the code-word $C(x^n)$, from left-to-right, until a codeword is matched.
    The extension of a prefix-free code is \emph{always} a valid code over sequences.
\end{definition}
\changed{
    \begin{example}[Prefix-free Codes]\label{example:prefix-free-codes}
        For $\X = \{a, b, c, d\}$, $C_1$, $C_2$, defined below, are prefix-free codes, but $C^\star$ is not due to $C^\star(b)$ being a prefix of $C^\star(d)$.
        However, $C^\star$ is the optimal code for $n=1$ if $P_X(a) \geq P_X(b) \geq P_X(c) \geq P_X(d)$.
        \begin{align}
            C_1(a) & = 0   & C_2(a) & = 1   & C^\star(a) & = \varepsilon \\
            C_1(b) & = 10  & C_2(b) & = 01  & C^\star(b) & = 0           \\
            C_1(c) & = 110 & C_2(c) & = 001 & C^\star(c) & = 1           \\
            C_1(d) & = 111 & C_2(d) & = 000 & C^\star(d) & = 00
        \end{align}
    \end{example}
}
The optimal code always assigns the smallest codewords in $\{0, 1\}^\star$ to the most likely symbols.
This precludes it from being prefix-free for any alphabet of size larger than $2$.

Assigning $1$ as a codeword in a prefix-free code eliminates the possibility of adding any other codeword starting with $1$ to the codebook.
This observation leads to the following condition on prefix-free codes.
\changed{
    \begin{theorem}[Kraft's Inequality \cite{cover1999elements}]\label{theorem:kraft}
        The set of code-lengths $\{\ell_x \in \Naturals \colon x \in \X\}$ of a prefix-free symbol code obey the following inequality, known as Kraft's Inequality,
        \begin{align}
            \sum_{x \in \X} 2^{-\ell_x} \leq 1.
        \end{align}
        Conversely, for any set of code-lengths satisfying Kraft's Inequality, there exists a prefix-free symbol code with this set of code-lengths.
    \end{theorem}
}
\begin{proof}
    This proof is adapted from \cite[Theorem 5.2.1]{cover1999elements}.

    Consider the complete binary tree shown on the left in \Cref{fig:binary-tree-kraft}.
    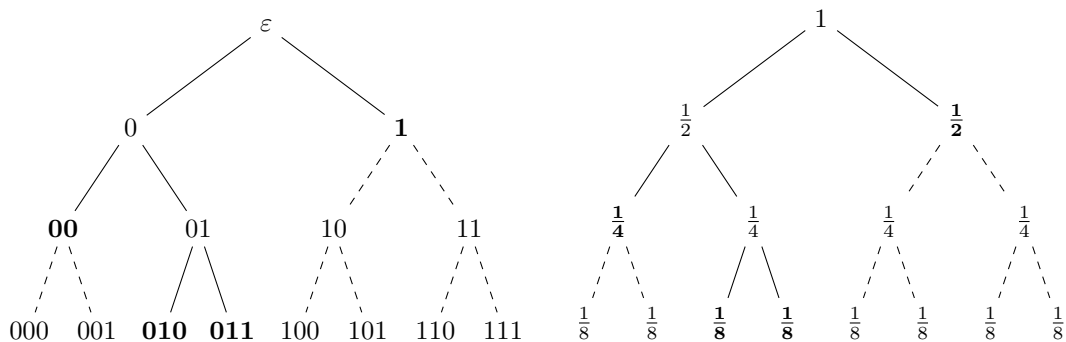
\begin{figure}[h!]
        \centering
        \begin{tikzpicture}[scale=0.9,level distance=1.5cm,
                level 1/.style={sibling distance=4cm},
                level 2/.style={sibling distance=2cm},
                level 3/.style={sibling distance=1cm}]
            \node {$\emptystr$}
            child {node {$0$}
                    child {node {$\mathbf{00}$}
                            child[dashed] {node {$000$}}
                            child[dashed] {node {$001$}}
                        }
                    child {node {$01$}
                            child {node {$\mathbf{010}$}}
                            child {node {$\mathbf{011}$}}
                        }
                }
            child {node {$\mathbf{1}$}
                    child[dashed] {node {$10$}
                            child {node {$100$}}
                            child {node {$101$}}
                        }
                    child[dashed] {node {$11$}
                            child {node {$110$}}
                            child {node {$111$}}
                        }
                };
        \end{tikzpicture}
        \quad
        \begin{tikzpicture}[scale=0.9, level distance=1.5cm,
                level 1/.style={sibling distance=4cm},
                level 2/.style={sibling distance=2cm},
                level 3/.style={sibling distance=1cm}]
            \node {$1$}
            child {node {$\frac{1}{2}$}
                    child {node {$\mathbf{\frac{1}{4}}$}
                            child[dashed] {node {$\frac{1}{8}$}}
                            child[dashed] {node {$\frac{1}{8}$}}
                        }
                    child {node {$\frac{1}{4}$}
                            child {node {$\mathbf{\frac{1}{8}}$}}
                            child {node {$\mathbf{\frac{1}{8}}$}}
                        }
                }
            child {node {$\mathbf{\frac{1}{2}}$}
                    child[dashed] {node {$\frac{1}{4}$}
                            child {node {$\frac{1}{8}$}}
                            child {node {$\frac{1}{8}$}}
                        }
                    child[dashed] {node {$\frac{1}{4}$}
                            child {node {$\frac{1}{8}$}}
                            child {node {$\frac{1}{8}$}}
                        }
                };
        \end{tikzpicture}
        \caption{Binary trees with binary strings (left) and $2^{-\ell}$ (right) as nodes. Codewords in bold form a prefix-free code for an alphabet of size $4$.}
        \label{fig:binary-tree-kraft}
    \end{figure}
    The vertex set is equal to all possible codewords in $\binstrs$ up to length $3$.
    The left and right child nodes are constructed by appending $0$ and $1$, respectively, to the parent node.
    The right tree in \Cref{fig:binary-tree-kraft} was created by replacing each potential codeword of length $\ell$ for $2^{-\ell}$.
    The sum of right and left children always equal the parent node.
    Any node is a prefix of its descendants in the left tree.
    Adding any node to the codebook precludes the entire subtree of its descendants from being valid codewords.
    For example, adding $1$ and $00$ to the codebook eliminates the subtrees with dashed edges.
    Consider the uniform codebook composed of all leaves in the left tree of \Cref{fig:binary-tree-kraft},
    \begin{align}
        \{000, 001, 010, 011, 100, 101, 110, 111\},
    \end{align}
    where it is clear that $\sum_{x \in \X} 2^{-\ell_x} = 1$.
    This is a prefix-free code for an alphabet of size $k = 8$.
    A prefix-free code for $k-1$ can be constructed from either removing one codeword from the codebook, $(\sum_{x \in \X} 2^{-\ell_x} < 1)$, or replacing two codewords, that are siblings in the tree, for their parent $(\sum_{x \in \X} 2^{-\ell} = 1)$.
    Similarly, replacing any codeword with its children creates a code for $k+1$.
    Repeating this procedure allows us to construct any possible prefix-free code, which will always satisfy Kraft's Inequality.

    \changed{
        Conversely, given a set of code-lengths, we can construct a codebook using the same procedure outlined before.
        First, select a length from the set and associate it to one of the nodes in the left tree with the same codeword length.
        Eliminate the entire subtree of descendants of the selected node as possible candidates.
        Repeat this procedure until every code-length has an associated code-word.
        By construction, the code composed of the resulting codebook will be prefix-free.
    }
\end{proof}

\changed{
    Kraft's Inequality is a necessary condition on the set of code-word lengths for the existence of a prefix-free code.}
From an optimization perspective, this inequality can used as a constraint to parameterize sets of codeword lengths within the prefix-free family.
The best solutions, i.e., codes with shortest average length, are those that achieve equality in Kraft's Inequality.
Kraft's Inequality directly limits the rate achievable by prefix-free codes.
The codebook must increase in size to accommodate new symbols as the alphabet increases; if we wish to guarantee unique decodability.
This is characterized by the following theorem.

\begin{theorem}[Source Coding Theorem: Prefix-Free Codes \cite{cover1999elements}]\label{theorem:source-coding-theorem-prefix-free-codes}
    The rate of any prefix-free code is lower bounded by the entropy of the mixture of marginals,
    \begin{align}
        \bar{P}_X(x) = \frac{1}{n}\cdot\sum_{i=1}^n P_{X_i}(x).
    \end{align}
\end{theorem}
\begin{proof}
    The family of prefix-free codes is composed of all codes that satisfy Kraft's Inequality (\Cref{theorem:kraft}).
    The best performing code in this family is the solution to the following optimization problem,
    \begin{align}\label{opt:prefix-free-code}
        \begin{split}
            \min_C    & \quad R(C, P_{X^n}) = \frac{1}{n} \cdot \E{\ell_{X^n}}         \\
            \suchthat & \quad \ell_x \in \Naturals, \sum_{x \in \X} 2^{-\ell_x} \leq 1
        \end{split}
    \end{align}
    The code-word length of the extended code equals the sum of symbol code-word lengths due to concatenation,
    \begin{align}
        \ell_{x^n} = \sum_{i=1}^n \ell_{x_i}.
    \end{align}
    The objective function can be rewritten to expose the mixture of marginals,
    \begin{align}
        \frac{1}{n} \cdot \E{\ell_{X^n}} & = \frac{1}{n} \cdot \sum_{i=1}^n \E{\ell_{X_i}}                                    \\
                                         & = \frac{1}{n} \cdot \sum_{i=1}^n \sum_{x \in \X}P_{X_i}(x)\cdot\ell_x              \\
                                         & =\sum_{x \in \X}\ell_x \cdot \left(\frac{1}{n}\cdot\sum_{i=1}^n P_{X_i}(x)\right).
    \end{align}
    The rest of the proof mirrors that of \cite[5.3 Optimal Codes]{cover1999elements}.
    Dropping the integer constraint in \eqref{opt:prefix-free-code} gives a lower bound on the solution.
    If $\tilde{x}_i \in \X$ are the symbols in the alphabet, then applying the Lagrange-Multiplier method \cite{boyd2004convex} yields,
    \begin{align}
        \mathcal{L}(\ell_{\tilde{x}_1}, \dots, \ell_{\tilde{x}_{\abs{\X}}}) & = \sum_{x \in \X}\ell_x \cdot \bar{P}_X(x) + \lambda \cdot \left(\sum_{x \in \X} 2^{-\ell_x} - 1 \right), \\
        \frac{\partial\mathcal{L}}{\partial\ell_x}(\ell_x^\star)            & = \bar{P}_X(x) - \lambda \cdot 2^{-\ell_x^\star} \cdot \log_e(2) = 0.
    \end{align}
    Using the constraint given by Kraft's Inequality we can bound the value of $\lambda$,
    \begin{align}
        \sum_{x \in \X} 2^{-\ell_x^\star} = \frac{1}{\lambda\log_e(2)} \leq 1.
    \end{align}
    Substituting $\lambda$ back into $\mathcal{L}^\prime(\ell_x^\star)$ gives the final solution,
    \begin{align}
        \ell_x^\star & \geq -\log\bar{P}_X(x),
    \end{align}
    with rate,
    \begin{align}
        R(C^\star, P_{X^n}) & =\sum_{x \in \X}\ell_x^\star \cdot \bar{P}_X(x) \geq H(\bar{P}_X).
    \end{align}
\end{proof}
\Cref{theorem:source-coding-theorem-prefix-free-codes} shows the rate of a prefix-free code is lower bounded by the entropy of the mixture of marginals, but does not address its achievability.
The lower bound appears due to the integer constraint on code-lengths.
The integer constraint is automatically satisfied if the probability values of the data distribution are powers of $2$,
\begin{align}
    P_{X}(x) = 2^{-k_x},
\end{align}
where $k_x \in \Naturals$ for all $x \in \X$.
A distribution satisfying this constraint is said to be \emph{dyadic}.
A prefix code with code-word lengths $\ell_x = -\log P_X(x)$ will always exist for any dyadic source, with rate equal to the entropy $H(P_X)$.
\begin{example}[Prefix-Free Code for a Dyadic Source]
    The following distribution is dyadic and has an optimal prefix-free code $C$.
    \begin{align}
        P_X(a) & = 2^{-1} & C(a) & = 0   \\
        P_X(b) & = 2^{-2} & C(b) & = 10  \\
        P_X(c) & = 2^{-3} & C(c) & = 110 \\
        P_X(d) & = 2^{-3} & C(d) & = 111
    \end{align}
    The code-lengths are guaranteed to be integer if $P_X$ is dyadic, implying
    \begin{align}
        \ell_x = -\log P_X(x) &  & R(C, P_{X^n}) = H(P_X).
    \end{align}
\end{example}
A general algorithm was given in \cite{huffman1952method} that constructs the optimal prefix-free code, i.e., a prefix-free code that has rate less than or equal to any other prefix-free code, for any source distribution.
These codes are famously referred to as Huffman Codes and are heavily used in practice as sub-components of larger compression systems (see \cite{steinruecken2014b} for a review).
The gap between the entropy and the rate of a Huffman Code is at most $1$ bit per symbol for any i.i.d.\ data source \cite{huffman1952method,cover1999elements}.
An extra bit per symbols can be significant for practical applications where the sequence length is large.
Huffman Codes are beyond the scope of this manuscript.
A full description is given in \cite{cover1999elements}.

\section{Lossless Source Coding Theorems}
Extended and prefix-free codes provide upper bounds to the objective of \eqref{opt:optimal-code} in exchange for a lower implementation complexity.
In this section we discuss the gap in bits between the solution given by optimal, extended, and prefix-free codes as a function of $n$.
For ease of exposition, we limit the discussion to identically distributed data sources, $P_{X_i} = P_X$ for all $i \in [n]$.

\subsection{Rates Below Entropy are Achievable for Finite-Length Sequences}\label{subsec:rates-below-entropy-are-achievable}
The rate of any prefix-free code is lower bounded by the entropy of the mixture of marginals, which equals
the entropy of the data source, $H(P_X)$, when symbols are identically distributed.
Optimal codes assign codewords sequentially based on length and have the lowest rates amongst all codes, including prefix-free codes.
For dyadic sources, optimal codes achieve rate below entropy for any finite sequence length.

\begin{example}[The Optimal Rate of a Dyadic Source is Below Entropy]\label{example:optimal-rate-dyadic-below-entropy}
    For any dyadic probability distribution over an alphabet $\X$, the rate achieved by the optimal code (\Cref{def:optimal-code}) is less than or equal to the entropy of the source, with equality when $\abs{\X} = 2$.

    Codewords higher up in the binary tree of \Cref{fig:binary-tree-dyadic-vs-optimal} have smaller code lengths.
    The optimal code assigns codewords in bread-first fashion, filling shallower depths first.
    Given a prefix-free code, we can construct the optimal code by swapping out existing codewords in the codebook for codewords higher up in the binary tree, as shown in \Cref{fig:binary-tree-dyadic-vs-optimal}.
    For $\abs{\X} = 2$ the only optimal codebook, which is also prefix-free, is $\{0, 1\}$
    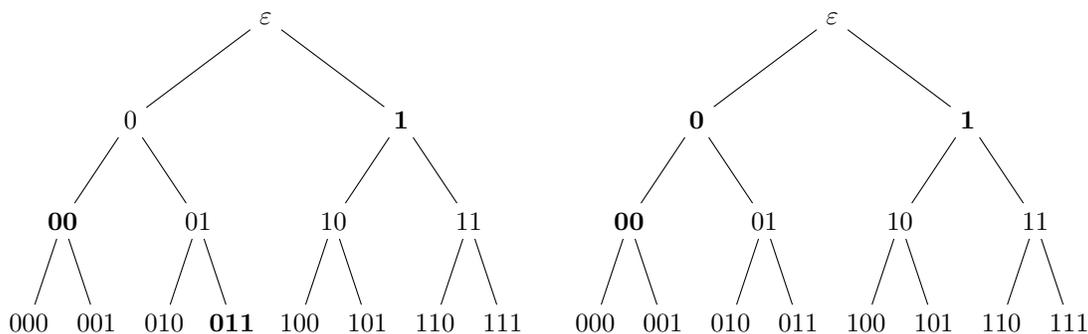
\begin{figure}[h!]
        \centering
        \begin{tikzpicture}[scale=0.9,level distance=1.5cm,
                level 1/.style={sibling distance=4cm},
                level 2/.style={sibling distance=2cm},
                level 3/.style={sibling distance=1cm}]
            \node {$\emptystr$}
            child {node {$0$}
                    child {node {$\mathbf{00}$}
                            child {node {$000$}}
                            child {node {$001$}}
                        }
                    child {node {$01$}
                            child {node {$010$}}
                            child {node {$\mathbf{011}$}}
                        }
                }
            child {node {$\mathbf{1}$}
                    child {node {$10$}
                            child {node {$100$}}
                            child {node {$101$}}
                        }
                    child {node {$11$}
                            child {node {$110$}}
                            child {node {$111$}}
                        }
                };
        \end{tikzpicture}
        \quad
        \begin{tikzpicture}[scale=0.9,level distance=1.5cm,
                level 1/.style={sibling distance=4cm},
                level 2/.style={sibling distance=2cm},
                level 3/.style={sibling distance=1cm}]
            \node {$\emptystr$}
            child {node {$\mathbf{0}$}
                    child {node {$\mathbf{00}$}
                            child {node {$000$}}
                            child {node {$001$}}
                        }
                    child {node {$01$}
                            child {node {$010$}}
                            child {node {$011$}}
                        }
                }
            child {node {$\mathbf{1}$}
                    child {node {$10$}
                            child {node {$100$}}
                            child {node {$101$}}
                        }
                    child {node {$11$}
                            child {node {$110$}}
                            child {node {$111$}}
                        }
                };
        \end{tikzpicture}
        \caption{Binary tree for a prefix-free codebook $\{1, 00, 011\}$ (left). Optimal codebook for $\abs{\X}=3$ (right).}
        \label{fig:binary-tree-dyadic-vs-optimal}
    \end{figure}
\end{example}

\begin{figure}[t]
    \centering
    \includegraphics[width=\textwidth]{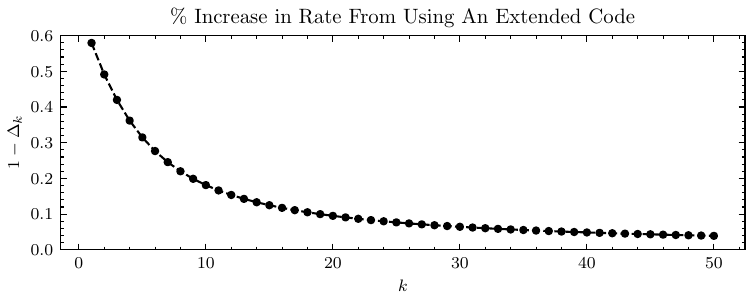}
    \caption{Percentage increase, with respect to the optimal, from using an extended uniform code. See \Cref{example:optimal-rate-for-a-uniform-source}.}
    \label{fig:fig-optimal-code-uniform-source}
\end{figure}

\changed{
    \begin{example}[Optimal Rate for a Uniform Source]\label{example:optimal-rate-for-a-uniform-source}
        The optimal rate for a data source with uniform distribution over $\X$ can be characterized exactly.
        The entropy of a uniform distribution is equal to
        \begin{align}
            H(P_X) = \E{-\log P_X(X)} = \log\abs{\X}.
        \end{align}
        The optimal code assigns codewords sequentially, in increasing order of length, starting from the most likely symbol in $\X$.
        There are $2^\ell$ codewords of length $\ell$.
        Without loss of generality, we can reparameterize the alphabet size $\abs{\X}^n$ to be a sum of polynomial of powers of $2$,
        \begin{align}\label{eq:alphabet-as-sum-of-k-powers}
            \abs{\X}^n & \defeq \sum_{\ell=0}^k 2^\ell \\
                       & = 2^{k+1} - 1
        \end{align}
        The sequence length can be written as a function of $k$,
        \begin{align}\label{eq:n-as-function-of-k}
            n = \frac{\log \left( 2^{k+1} - 1 \right)}{\log\abs{\X}}.
        \end{align}
        We can now show the rate is strictly less than the entropy of the uniform source for all $ k \geq 2$,
        \begin{align}
            R(C^\star, P_{X^n})
             & = \frac{1}{n} \cdot \sum_{x^n \in \X^n} P_{X^n}(x^n) \cdot \ell_{x^n}                                    \\
             & = \frac{1}{n\abs{\X}^n} \cdot \sum_{x^n \in \X^n} \ell_{x^n}                                             \\
             & = \frac{1}{n\abs{\X}^n} \cdot \sum_{\ell=0}^k \abs{\{x^n \in \X^n \colon \ell_{x^n} = \ell\}} \cdot \ell \\
             & = \frac{1}{n\abs{\X}^n} \cdot \sum_{\ell=0}^k 2^\ell \cdot \ell                                          \\
             & = \frac{1}{n\abs{\X}^n} \cdot \left( 2^{k+1} \cdot k - 2^{k+1} + 2 \right).
        \end{align}
        Together with \Cref{eq:n-as-function-of-k}, the rate can be shown to equal,
        \begin{align}
            R(C^\star, P_{X^n}) = \Delta_k \cdot \log\abs{\X},
        \end{align}
        where,
        \begin{align}
            \Delta_k \defeq \frac{2^{k+1} \cdot k - 2^{k+1} + 2}{(2^{k+1} - 1) \cdot \log\left(2^{k+1} - 1\right)} \leq 1,
        \end{align}
        approaching equality as $k \rightarrow \infty$, and with equality when $k=1$.
        The rate of the best performing extended code is worse (larger) than the optimal by a multiple of $1 - \Delta_k$.
        Extended codes approach optimality for the uniform source as the sequence length increases.

        \Cref{fig:fig-optimal-code-uniform-source} shows $1 - \Delta_k$ as a function of $k$.
        The loss in performance drops off quickly and is irrelevant for the applications considered in this thesis.
        For example, when the symbols are pixels of an image, or bytes of a file, where the alphabet size is $\abs{\X} = 256$,
        a sequence length of $n > 25$ is sufficient for $1 - \Delta_k < 1\%$.
    \end{example}
}

\subsection{Prefix-Free Codes are Optimal Asymptotically}
The following theorem guarantees that the set of rates achievable by extended and prefix-free codes are equal when $n \rightarrow \infty$.
The rate of extended codes are therefore also lower bounded by the mixture of marginals presented in \Cref{theorem:source-coding-theorem-prefix-free-codes}.
\begin{theorem}[Sufficiency of Prefix-free Codes \cite{cover1999elements}]\label{theorem:mcmillan}
    For any extended code there exists a prefix-free code with the same set of codeword lengths and rate when $n \rightarrow \infty$.
\end{theorem}
\begin{proof}
    This proof was adapted from \cite[Theorem 5.5.1]{cover1999elements}.

    We will prove that the codeword lengths of any uniquely decodable symbol code satisfy
    \begin{align}\label{eq:mcmillan}
        \sum_{x \in \X} 2^{-\ell_x} \leq (n \ell_{\max})^{1/n},
    \end{align}
    where $\ell_{\max} = \max_{x \in \X} \ell_x$ is the largest codeword in the symbol code.
    \Cref{eq:mcmillan} converges to Kraft's Inequality (\Cref{theorem:kraft}) as $n \rightarrow \infty$, which concludes the proof.

    Let $a(\ell)$ be the number of codewords of length $\ell$ in the extension.
    There are only $2^\ell$ binary strings of length $\ell$ implying,
    \begin{align}
        a(\ell) \defeq \abs{\{x^n \in \X^n \colon \ell_{x^n} = \ell \}} \leq 2^\ell
    \end{align}
    The largest codewords in the extension will have length $\argmax_{x^n \in \X^n} \ell_{x^n} = n\ell_{\max}$.
    Therefore,
    \begin{align}
        \left(\sum_{x \in \X} 2^{-\ell_x} \right)^n
         & = \sum_{(x_1, \dots, x_n) \in \X^n} 2^{-(\ell_{x_1} + \dots + \ell_{x_n})} \\
         & = \sum_{x^n \in \X^n} 2^{-\ell_{x^n}}                                      \\
         & = \sum_{\ell=1}^{n\ell_{\max}} a(\ell) \cdot 2^{-\ell}                     \\
         & \leq \sum_{\ell=1}^{n\ell_{\max}} 2^{\ell} \cdot 2^{-\ell}                 \\
         & = n \cdot \ell_{\max}.
    \end{align}
\end{proof}

\section{Conclusion and Discussion}
For any data source $(X^n, \X^n, P_{X^n})$, the optimal rate is achieved by sorting the instances of $\X^n$, according to $P_{X^n}$, and assigning codewords sequentially ordered by their length (see \Cref{def:optimal-code}).
This scheme requires large dictionaries for encoding and decoding, which scale exponentially with the sequence length $n$, and are therefore infeasible for the applications considered in this manuscript.
For finite $n$ the optimal code can achieve rates \emph{lower than the entropy} of the source.
This is shown in \Cref{example:optimal-rate-dyadic-below-entropy} where the sequence is composed of independent and identically distributed (i.i.d.) elements and the distribution is dyadic.

The family of Extended Codes (\Cref{def:extended-codes}), which include Prefix-Free Codes (\Cref{opt:prefix-free-code}), are codes defined over the symbol alphabet $\X$.
These are introduced as a solution with dictionary-based implementations independent of sequence length.
The rate of extended codes can be larger than the optimal by a significant amount (\Cref{fig:fig-optimal-code-uniform-source}) when the sequence length is small.
For large sequence lengths the achievable rates of optimal, extended, and prefix-free codes are equal, in the i.i.d.\ setting, and are lower bounded by the entropy of the source when the data distribution is uniform.

The smallest rate of any prefix-free code is achieved by Huffman Codes \cite{huffman1952method}, and is guaranteed to be within $1$ of the entropy of the source.
It is possible to achieve the entropy by constructing a Huffman Code over $\X^n$, by consider the sequence itself as the symbol.
However, the alphabet $\X^n$ grows exponentially with $n$, defeating the purpose of introducing an extended code to lower the complexity of implementation.
The next chapter discusses an alternative family that make use of the conditional distributions to design a code and do not suffer from this $1$ bit penalty.

\chapter{Entropy Coding with Asymmetric Numeral Systems}\label{chapter:entropy-coding-with-ans}
In \Cref{chapter:lossless-source-coding} we showed that the rates achievable by prefix-free and extended codes approach the entropy of the source for i.i.d.\ random variables.
Asymmetric Numeral Systems (ANS) \cite{duda2009asymmetric} is a family of codes capable of achieving rates very close to the entropy for large sequence lengths.
At a high level, ANS encoding works by mapping sequences of symbols to a single, large, natural number, called the \emph{state}, in a reversible way so that the data and state before encoding can be easily recovered.

ANS requires access to a quantized probability distribution (\Cref{def:quantized-pmf}) for both encoding and decoding, which we will often refer to a the \emph{model}.

Symbols are decoded in the opposite order in which they were encoded, making ANS act as a \emph{stack}, a last-in-first-out data structure, in contrast to the coding schemes discussed in \Cref{chapter:lossless-source-coding} which are \emph{queue}-like or first-in-first-out.

The construction of ANS guarantees that any \emph{non-negative integer} is a valid state.
This allows ANS to be used as a random variate generator.
The state is randomly initialized to a large non-negative integer and successive decode operations are applied, generating a sequence from the chosen probability distribution.
The state need not have been randomly initialized, but may instead have been constructed from a sequence of encode steps from a previous encoding task, as the origin of the state is irrelevant as long as it is integer.
This allows the sampling distribution to differ from the distribution used for the previous task.
In this use case, the ANS state is thus used as a \emph{random seed}, and because decoding/sampling removes approximately bits from the state, the random seed is slowly consumed as symbols are sampled.
The operation is invertible in the sense that the random seed can be recovered by encoding the generated symbols back into the ANS state.

In real-world applications the true data distribution is rarely available.
Instead, an estimate is computed from empirical observations.
We delay this discussion until later in the section and assume, initially, that $P_{X^n}$ is known.

\section{Asymmetric Numeral Systems (ANS)}
ANS forms a central component of the compression algorithms developed in \Cref{chapter:roc}, \Cref{chapter:rcc}, and \Cref{chapter:rec} of this thesis.
In what follows, we discuss an idealized, mathematical, version of ANS.
We frame the analysis of the rate in terms of the increase in the number of bits needed to represent the ANS state as a function of the encoded symbols.
For large state values, the change is equal to the negative log-probability of the encoded symbol.
Implementation details are delayed to the experimental sections of the chapter mentioned above.

\subsection{ANS Encoding}\label{sec:ans-encoding}
Given a data source $(X, \X, P_X)$, ANS requires the distribution be \emph{quantized} to some \emph{precision} parameter $N \in \Naturals$.
The probability of any symbol under a quantized probability distribution can be computed by a ratio between some integer and the precision parameter.
\begin{definition}[Quantized Probability Distribution]\label{def:quantized-pmf}
    A probability distribution $P_X$ is \emph{quantized} if it can be described by an element of the \emph{discrete simplex},
    \begin{equation}
        \left\{ (p_1, \dots, p_{\abs{\X}}) \in \Naturals^{\abs{\X}} \colon \sum_{i=1}^{\abs{\X}} p_i = N \text{ and } p_i \geq 1 \right\},
    \end{equation}
    where $\left(p_1, \dots, p_{\abs{\X}}\right)$ maps to
    \begin{equation}\label{eq:quantized-distribution}
        P_X(x) \defeq \frac{p_x}{N}.
    \end{equation}
\end{definition}
\changed{In \Cref{def:quantized-pmf}, we require that $p_i \geq 1$ to avoid zero-division errors in the encoding function, as defined later in \Cref{def:ans-encode}.}
In practice, a function specifying the quantized probability, $p_x \in \Naturals$, as well as the \emph{quantized} cumulative probability, $c_x = \sum_{y\colon y < x}p_y$, is required for ANS encoding; \changed{where it is assumed that the set of symbols is equipped with an ordering such that $y < x$ in the summation is well defined.}
This function is known as a \emph{forward lookup},
\begin{align}
     & \flookup\colon \X \mapsto [N)^2,          \\
     & \flookup(x) \defeq \left(p_x, c_x\right).
\end{align}
The interval $[c_x, c_x + p_x)$ is referred to as the \emph{range} of $x$ in the probability model.
\begin{example}[Quantized PMF]\label{example:initial-bits-symbol-order-matters}
    Let $\X = \{\square, \triangle, \lozenge\}$, with $p_\square = 1, p_\triangle=2, p_\lozenge = 4$, and $N = 7$.
    The values of $c_x$ are listed below, depending on the assumed ordering between symbols,
    \begin{align}
        \square \triangle \lozenge &  & \square \lozenge \triangle &  & \triangle \square \lozenge &  & \triangle \lozenge \square &  & \lozenge \square \triangle &  & \lozenge \triangle \square \\
        0, 1, 3                    &  & 0, 1, 5                    &  & 0, 2, 3,                   &  & 0, 2, 6                    &  & 0, 4, 5                    &  & 0, 4, 6
    \end{align}
    where the number underneath each symbol is the corresponding quantized cumulative probability.
\end{example}
Encoding and decoding can be specified by defining $\encode$ and $\decode$ as in \eqref{eq:codes-encode-decode}.
The codebook is a subset of the set of binary strings of length $\ceil{\log s}$ used to represent the integer state $s \in \Naturals$,
\begin{align}
    C(\X^n) \subset \left\{b \in \binstrs \colon b \text{ is the binary representation of } s \in \Naturals \text{ of size } \ceil{\log s} \right\}.
\end{align}
Representing the integer state as a binary string is deterministic and can be done efficiently on modern digital computers.
To facilitate the exposition we define $\encode$ and $\decode$ as manipulating the integer state, without loss of generality.
\begin{definition}[ANS Encode]\label{def:ans-encode}
    Given a symbol $x \in \X$ with range defined by $[c_x, c_x + p_x)$, ANS encodes $x$ into an existing state $s$ through the following function,
    \begin{align}
         & \encode\colon \Naturals \times [N)^2 \rightarrow \Naturals, \\
         & \encode(s, p, c) \defeq N\cdot(s \div p) + c + s \bmod p,
    \end{align}
    with $\div$ denoting \emph{integer} division, discarding any remainder.
    The new state $s^\prime = \encode(s, p_x, c_x)$ is said to ``contain'' $x$, which can be recovered via decoding (\Cref{sec:ans-decode}).
    When clear from context, ``$\encode(s, x) \text{ with $P_X$}$'' will be used as a shorthand for encoding $x$ into the ans state $s$ via $\flookup$ with the quantized probability distribution $P_X$.
\end{definition}
\changed{
    \begin{example}[ANS Encoding]\label{example:ans-encoding}
        Let $s=5$, $\X = \{a, b, c\}$, with ordering $a < b < c$, $P_X(a) = \frac{1}{2}, P_X(b)=P_X(c) = \frac{1}{4}$, and precision $N=4$.
        Let $\bin\colon \Naturals \mapsto \{0, 1\}^\star$ be a function that returns the binary representation of an integer.
        The encoding of each symbol results in a different ANS state $s^\prime = \encode(5, p_x, c_x) = 4 \cdot (5 \div p_x) + c_x + 5 \bmod p_x$,
        \begin{center}
            \begin{tabular}{c c c c}
                $x  $                & $a$    & $b$     & $c$     \\
                $p_x$                & $2$    & $1$     & $1$     \\
                $c_x$                & $0$    & $2$     & $3$     \\
                $s^\prime$           & $9$    & $22$    & $23$    \\
                $\bin(s^\prime)$     & $1001$ & $10110$ & $10111$ \\
                $\frac{s}{s^\prime}$ & $0.56$ & $ 0.23$ & $0.22$
            \end{tabular}
        \end{center}
        where the last row is rounded to the second decimal point.
        As the state $s$ increases, the ratio $\frac{s}{s^\prime}$ approaches $P_X(x)$, as can be seen in the encoding table below where $s=20$, $bin(s) = 10100$.
        \begin{center}
            \begin{tabular}{c c c c}
                $x  $                & $a$      & $b$       & $c$       \\
                $s^\prime$           & $40$     & $82$      & $83$      \\
                $\bin(s^\prime)$     & $101000$ & $1010010$ & $1010011$ \\
                $\frac{s}{s^\prime}$ & $0.5$    & $ 0.24$   & $0.24$
            \end{tabular}
        \end{center}
        Note the code-words in this example are in the family of Huffman Codes \cite{huffman1952method}.
        In fact, this will be true if the values of the probability distribution are inverse powers of two, i.e., the source is \emph{dyadic} \cite{duda2009asymmetric}.
    \end{example}

}

ANS defines a code over sequences of length $n$ by encoding symbols sequentially into a common state.
The codeword of the sequence is the binary representation of the final state after all symbols have been encoded.
\begin{definition}[ANS Code]\label{def:ans-code}
    Let $(X^n, \X^n, P_{X^n})$ be a data source with quantized conditional probability distributions (\Cref{def:quantized-pmf}),
    \begin{align}
        P_{X_i \g X^{i-1}}(x \g x^{i-1}) \defeq \frac{p_x^{(i)}}{N}.
    \end{align}
    Then, ANS defines a code on $\X^n$ through the following recursive equations,
    \begin{align}
        s_i & \defeq \encode(s_{i-1}, p_{x_i}^{(i)}, c_{x_i}^{(i)}).
    \end{align}
    The initial state, $s_0 \in \Naturals$, can be initialized to an arbitrary large integer \changed{(a precise definition of what constitutes as large is discussed in \Cref{sec:bbans-initial-bits})}.
    The codeword for $x^n$ is the binary representation of the final state, $s_n$, and has length $\ceil{\log s_n}$.
\end{definition}

\changed{
    \begin{example}[ANS Encoding]
        For the same source in \Cref{example:ans-encoding}, with $s_0=20$, the following table shows the values of the states $s_i$ after encoding $x_i$, as a function of all possible permutations of the sequence $aabc$.
        \begin{center}
            \begin{tabular}{c c c c c}
                $ x^n$ & $s_1$ & $s_2$ & $s_3$ & $s_4$  \\
                \midrule
                $aabc$ & $40$  & $80$  & $322$ & $1291$ \\
                $aacb$ & $40$  & $80$  & $323$ & $1294$ \\
                $abac$ & $40$  & $162$ & $324$ & $1299$ \\
                $abca$ & $40$  & $162$ & $651$ & $1301$ \\
                $acab$ & $40$  & $163$ & $325$ & $1302$ \\
                $acba$ & $40$  & $163$ & $654$ & $1308$ \\
                $baac$ & $82$  & $164$ & $328$ & $1315$ \\
                $baca$ & $82$  & $164$ & $659$ & $1317$ \\
                $bcaa$ & $82$  & $331$ & $661$ & $1321$ \\
                $caab$ & $83$  & $165$ & $329$ & $1318$ \\
                $caba$ & $83$  & $165$ & $662$ & $1324$ \\
                $cbaa$ & $83$  & $334$ & $668$ & $1336$
            \end{tabular}
        \end{center}
        Note how for all permutations,
        \begin{align}
            -\log P_{X^n}(x^n) & = -\log \left( P_X(a)^2 \cdot P_X(b) \cdot P_X(c) \right) \\
                               & = 6 \text{ bits}                                          \\
                               & \approx \log \frac{s_4}{s_0}.
        \end{align}
        Each symbol increases the number of bits required to represent the ANS state by roughly its information content.
        This is the key factor to proving the optimality of ANS and is discussed further in \Cref{sec:ans-achievable-rates}.
    \end{example}
}

\subsection{Achievable Rates}\label{sec:ans-achievable-rates}
\Cref{chapter:lossless-source-coding} showed the rates achievable with extended and prefix-free codes are equal for large sequence lengths.
For large states, the increase in the ANS state, \changed{measured in bits}, is close to the negative log-probability, i.e., information content, of the encoded symbol under the quantized probability
The large state regime happens naturally due to the state increasing in value as more symbols are encoded.
We therefore center the discussion of rate and optimality around the change in size of the ANS state, as well as the asymptotic rate (\Cref{def:asymptotic-rate}).

\begin{definition}[State Change]
    Given an ANS state $s$, the increase in the ANS state from encoding a symbol $x \in \X$ (up to rounding errors) is,
    \begin{align}
        \Delta(s, x, P_X) \defeq \log\left( \frac{s^\prime}{s} \right),
    \end{align}
    where $s^\prime = \encode(s, p_x, c_x)$ is the new state and $P_X(x) = \frac{p_x}{N}$ is some quantized probability distribution with probabilities $p_x$ and precision $N$.
\end{definition}

\begin{theorem}[Optimality of ANS]\label{theorem:optimality-ans}
    For large state values $s$, encoding a symbol $x \in \X$ with range defined by $(p_x, c_x)$ increases the number of bits needed to represent the ANS state by the information content of the encoded symbol,
    \begin{align}
        \Delta(x, P_X) \defeq \lim_{s \rightarrow \infty} \Delta(s, x, P_X) = -\log P_X(x).
    \end{align}
\end{theorem}
\begin{proof}
    The next state can be upper bounded by,
    \begin{align}
        s^\prime & = N\cdot\left(s \div p_x\right) + c_x + s \bmod p_x &  &                                                         \\
                 & \leq N\cdot\frac{s}{p_x} + c_x + s \bmod p_x        &  & \left( a \div b \leq a/b \right)                        \\
                 & < N\cdot\frac{s}{p_x} + c_x + p_x                   &  & \left( a \bmod b < b \right)                            \\
                 & \leq N\cdot\frac{s}{p_x} +  N                       &  & \left( c_x + p_x \leq N \right)                         \\
                 & = \frac{s}{P_X(x)} + N                              &  & \left( \text{\Cref{eq:quantized-distribution}} \right).
    \end{align}
    Similarly, a lower bound can be given,
    \begin{align}
        s^\prime & \geq N\cdot\left(\frac{s}{p_x} - 1\right) + c_x + s \bmod p_x &  & \left( a \div b \geq a/b  - 1 \right)                   \\
                 & > N\cdot\left(\frac{s}{p_x} - 1\right)                        &  & \left( c_x + s \bmod p_x > 0 \right)                    \\
                 & = \frac{s}{P_X(x)} - N                                        &  & \left( \text{\Cref{eq:quantized-distribution}} \right).
    \end{align}
    The state change is therefore sandwiched by the logarithm of the upper and lower bounds, and approaches $-\log P_X(x)$ as $s \rightarrow \infty$,
    \begin{align}\label{eq:state-change-sandwich}
        \log\left( \frac{1}{P_X(x)} - \frac{N}{s} \right) \leq \Delta(s, x, P_X) \leq \log\left( \frac{1}{P_X(x)} + \frac{N}{s} \right).
    \end{align}
\end{proof}
Optimality is only guaranteed if the state is large enough to overwhelm the constant factors of \Cref{theorem:optimality-ans}.
This can be achieved by a technique known as \emph{renormalization}, where the state is forced to be above a minimal value at every step.
Renormalization ensures that the inaccuracy is bounded by \(2.2\times10^{-5}\) bits per operation, which is equivalent to one bit of redundancy for every 45,000 operations \cite{townsend2020tutorial}.
As well as this small per-symbol redundancy, in practical ANS implementations there are also one-time redundancies incurred when initializing and terminating encoding.
The one-time overhead is usually bounded by 16, 32 or 64 bits, depending on the implementation \cite{townsend2020tutorial,duda2009asymmetric}.

The asymptotic rate achieved by ANS will depend on the distribution used for encoding and decoding.
In the previous discussion we assumed $P_{X^n}$ was known.
For an arbitrary distribution $Q_{X^n}$ encoding in the large state regime yields an average state change of,
\begin{align}\label{eq:rate-entropy-model}
    \E*{\Delta(X^n, Q_{X^n})} = \E{- \log Q_{X^n}(X^n)},
\end{align}
where the code in \eqref{eq:rate-entropy-model} is implied from context to be constructed with ANS and the conditional distributions of $Q_{X_i \g X^{i-1}}$, as in \Cref{def:ans-code}.
This quantity is known as the \emph{cross-entropy} and is lower bounded by the entropy of the source distribution $P_{X^n}$, as shown next.
\begin{theorem}[Source Coding Theorem: ANS Codes]\label{theorem:source-coding-ans}
    Given a data source $(X^n, \X^n, P_{X^n})$ and model $Q_{X^n}$, the expected state change, in the large state regime, is lower bounded by the entropy of the data source,
    \begin{align}
        \E*{\Delta(X^n, Q_{X^n})}
         & = \E{- \log Q_{X^n}(X^n)}            \\
         & = H(P_{X^n}) + \KL{P_{X^n}}{Q_{X^n}} \\
         & \geq H(P_{X^n}),
    \end{align}
    with equality when $Q_{X^n} = P_{X^n}$.
\end{theorem}
\begin{proof}
    \begin{align}
        \E{-\log Q_{X^n}(X^n)}
         & = \E*{-\log \left( P_{X^n}(X^n) \cdot \frac{Q_{X^n}(X^n)}{P_{X^n}(X^n)}\right)}       \\
         & = \E*{-\log P_{X^n}(X^n)} + \E*{\log \left( \frac{P_{X^n}(X^n)}{Q_{X^n}(X^n)}\right)} \\
         & =  H(P_{X^n}) + \KL{P_{X^n}}{Q_{X^n}}                                                 \\
         & \geq  H(P_{X^n}),
    \end{align}
    where the last step follows from the non-negativity of the KL divergence \cite{cover1999elements}.
\end{proof}
\Cref{theorem:source-coding-ans} shows that the model achieving the lowest average state change is the data distribution itself, $Q_{X^n} = P_{X^n}$.
The increase from using any other distribution is equal to the KL divergence between the data distribution and the model,
and is known as the \emph{wrong-code penalty} \cite{cover1999elements}.

\begin{example}[Mixture of Marginals - ANS]
    \Cref{theorem:source-coding-theorem-prefix-free-codes} shows the rate of prefix-free codes is lower bounded by the entropy of the mixture of marginals.
    The lower bound is achieved by ANS within the family of i.i.d.\ models,
    \begin{align}
        \Q_{\text{i.i.d.}}
        = \left\{ Q_{X^n} \colon Q_{X^n}(x^n) = \prod_{i=1}^n Q_{X}(x_i) \text{ for some $Q_X$ over $\X$} \right\}.
    \end{align}
    The distribution $Q_{X^n}^\star \in \Q_{\text{i.i.d.}}$ minimizing the average state change and asymptotic rate is easily found to be the mixture of marginals, as can be seen from,
    \begin{align}
        \E*{\Delta(X^n, Q_{X^n})}
         & = \frac{1}{n}\cdot\E{-\log Q_{X^n}(X^n)}                                                               \\
         & = \frac{1}{n}\cdot\sum_{i=1}^n \E{-\log Q_X(X_i)}                                                      \\
         & = \frac{1}{n}\cdot\sum_{i=1}^n \sum_{x \in \X} P_{X_i}(x)\cdot (-\log Q_X(x))                          \\
         & = \sum_{x \in \X} \left(-\log Q_X(x) \cdot \left(\frac{1}{n}\cdot\sum_{i=1}^n P_{X_i}(x)\right)\right) \\
         & = H(\bar{P}_{X}) + \KL{\bar{P}_{X}}{Q_X}.
    \end{align}
\end{example}

The theorems and results presented so far discuss only the average change in the number of bits required to represent the ANS state, but not the final rate, or asymptotic rate, of the code.
The rate values observed in practice are close to the cross-entropy for the applications considered later in this thesis.
Given these results, for what follows we will focus the discussion around the cross-entropy between the model and source distribution.

\subsection{ANS Decoding and Sampling}\label{sec:ans-decode}
Decoding from ANS proceeds in the \emph{reverse} order of encoding.
The first symbol to be decoded is $x_n$, followed by $x_{n-1}$, and so on.
Doing so requires recovering the previous state, $s$, and the encoded symbol, $x$, from the current state $s_i$.
Decoding is possible as the state encodes a value in the range of $x$,
\begin{equation}
    j \defeq s_i \bmod N = c_x + s\bmod p_x \in [c_x, c_x + p_x),
\end{equation}
making it possible to recover the symbol by performing a binary search on all intervals $[c_x, c_x + p_x)$.
In the worst case, this search is $\Omega(\log\abs{\mathcal \X})$, although in some cases a search can be avoided, by mathematically computing the required interval and symbol (either analytically or via a lookup table).
Whether implemented using search or otherwise, we refer to the function which recovers $x$, $c_x$, and $p_x$, as the \emph{reverse lookup} function,
\begin{align}
     & \rlookup\colon  [n] \times [N) \mapsto \X \times [N)^2 \\
     & \rlookup(i, j) \defeq \left(x, p_x, c_x\right),
\end{align}
Knowing $j$, the decoded symbol $x$, and the current state $s_i$, we can recover the previous state via modular integer arithmetic.
We know $c_x + s \bmod p_x < c_x + p_x \leq N$ will never contribute to the increase in multiplicity, with respect to $N$, of the quantity in parenthesis below, resulting in the following equality,
\begin{align}
    s_i \div N
     & = \left( N\cdot(s \div p_x) + c_x + s \bmod p_x\right) \div N \\
     & = s \div p_x.
\end{align}
The decoder first recovers $x, p_x, c_x$ via $\rlookup$, from the current state $s_i$, and then restores the state to its value before the encoding of $x$,
\begin{align}
     & p_x \cdot (s_i \div N) + s_i \bmod N - c_x \\
     & =  p_x \cdot (s \div p_x) + s \bmod p_x    \\
     & = s.
\end{align}
This implies $\encode$ (\Cref{def:ans-encode}) has a well defined inverse,
\begin{align}\label{eq:ans-decode}
     & \decode\colon \Naturals \times [N)^{\abs{\X}} \rightarrow \Naturals \times\X \\
     & \decode(s_i, p_1, \dots, p_{\abs{\X}}) \defeq \left(s, x \right).
\end{align}
We sometimes replace the quantized probability parameters in \Cref{eq:ans-decode} with the symbol representing the distribution, when clear from context.

For every quantized distribution in the discrete simplex, ANS defines a partitioning of $\Naturals$ into disjoint sets $\Naturals_x$, for each $x \in \X$.
If the state is in $\Naturals_x$, then the most recently encoded symbol is $x$, which can be recovered by performing $\decode$.
The relative density of integers, for any sub-interval of $\Naturals$, is equal to the quantized probability $p_x$, as long as the interval length is a multiple of the precision $N$,
\begin{align}
    \frac{1}{m}\cdot\abs{\Naturals_x \cap [k, k + m \cdot N)} & = p_x, \text{ for all $k, m \in \Naturals \text{ and } m > 0$}.
\end{align}
\begin{example}[ANS Partitioning]
    For precision $N=6$, and sequence length $n=1$, the quantized probability distributions over $\X = \{a, b, c\}$,
    \begin{align}
        N \cdot P_X(a) & = 1 & N \cdot P^\prime_X(a) & = 2  \\
        N \cdot P_X(b) & = 3 & N \cdot P^\prime_X(b) & = 2  \\
        N \cdot P_X(c) & = 2 & N \cdot P^\prime_X(c) & = 2,
    \end{align}
    define partitionings,
    \begin{align*}
        \Naturals_a & = \{0 &  &   &  &   &  &   &  &   &  &   &  & 6 &  &   &  &   &  &   &  &    &  &    &  & 12 &  &    &  &    &  &    &  &    &  &    & \dots \\
        \Naturals_b & = \{  &  & 1 &  & 2 &  & 3 &  &   &  &   &  &   &  & 7 &  & 8 &  & 9 &  &    &  &    &  &    &  & 13 &  & 14 &  & 15 &  &    &  &    & \dots \\
        \Naturals_c & = \{  &  &   &  &   &  &   &  & 4 &  & 5 &  &   &  &   &  &   &  &   &  & 10 &  & 11 &  &    &  &    &  &    &  &    &  & 16 &  & 17 & \dots
    \end{align*}
    and,
    \begin{align*}
        \Naturals^\prime_a & = \{0 &  & 1 &  &   &  &   &  &   &  &   &  & 6 &  & 7 &  &   &  &   &  &    &  &    &  & 12 &  & 13 &  &    &  &    &  &    &  &    & \dots \\
        \Naturals^\prime_b & = \{  &  &   &  & 2 &  & 3 &  &   &  &   &  &   &  &   &  & 8 &  & 9 &  &    &  &    &  &    &  &    &  & 14 &  & 15 &  &    &  &    & \dots \\
        \Naturals^\prime_c & = \{  &  &   &  &   &  &   &  & 4 &  & 5 &  &   &  &   &  &   &  &   &  & 10 &  & 11 &  &    &  &    &  &    &  &    &  & 16 &  & 17 & \dots
    \end{align*}
\end{example}
This property allows ANS to be used as an \emph{invertible sampler}, where a sample is generated by decoding from a randomly initialized state.
Sampling can be inverted by encoding symbol back into the ANS state under the same quantized probability distribution.
\begin{lemma}[ANS Sampling]\label{lemma:ans-sampling}
    For any $k, m \in \Naturals$, if $U$ is a discrete, uniform, random variable in the interval $[k, k + m \cdot N)$, then decoding with ANS on the random state $U$,
    \begin{align}
        S, X \defeq \decode(U, P_X),
    \end{align}
    gives a random sample from the quantized probability distribution,
    \begin{align}
        X \sim P_X,
    \end{align}
    for any precision $N$.
    The value of $U$ can be recovered via encoding,
    \begin{align}
        U = \encode(S, X, p_x, c_x),
    \end{align}
    with probability one.
\end{lemma}
\Cref{lemma:ans-sampling} guarantees samples generated from ANS decoding will be distributed according to the quantized probability distribution $P_X$, but requires randomness in the form of the uniform random variable $U$.
The algorithms presented in this thesis make use of invertible sampling but with a shared ANS state.
Given an initial state $s_0$, we can generate a sequence of samples and states by performing,
\begin{align}
    s_i, x_i \defeq \decode(s_{i-1}, P_X).
\end{align}
Conditioned on the initial state $s_0$, the resulting sequences $s^n$ and $x^n$ are deterministic.
It therefore is not possible to make probabilistic statements regarding the distribution of the generated sample, as was done in \Cref{lemma:ans-sampling}.
Instead, we can show that most integer states map to sequences with empirical distributions close to $P_X$ in KL divergence.
\begin{definition}[Typical Sequence]
    An \emph{$\epsilon$-typical} sequence of some probability distribution $P_X$ is an instance $x^n \in \X^n$ with empirical distribution close to $P_X$ in KL divergence,
    \begin{align}
        \KL{P_X}{\hat{P}[x^n]} \leq \epsilon,
    \end{align}
    where $\hat{P}[x^n]$ is the empirical distribution,
    \begin{align}
        \hat{P}[x^n](x)
        = \frac{1}{n}\cdot\sum_{i=1}^n \1\{x = x_i\}
        = \frac{1}{n}\cdot\sum_{x^\prime \in \X} \1\{x = x^\prime\}.
    \end{align}
\end{definition}

\begin{theorem}[All Sequences are Typical Asymptotically]
    Let $X_i \sim P_X$ be i.i.d.\ random variables with common alphabet $\X$.
    Then, a random sequence is typical, for any $\epsilon$, as $n \rightarrow \infty$,
    \begin{equation}
        \lim_{n \rightarrow \infty} \KL{P_X}{\hat{P}[x^n]} = 0,
    \end{equation}
    with probability $1$.
\end{theorem}
\begin{proof}
    This proof is a direct consequence of the Law of Large Numbers \cite{durrett2019probability}.
    \begin{align}
        \hat{P}[X^n](x)
        = \frac{1}{n}\cdot\sum_{i=1}^n \1\{x = X_i\}
        \rightarrow
        = \E{\1\{x = X_i\}}
        = P_X(x)
    \end{align}
\end{proof}
In practice, the same renormalization technique discussed in \Cref{def:ans-code} is sufficient to guarantee typicality of the generated sequence.

ANS encoding increases the state inversely proportional to the probability of the encoded symbol.
Therefore, decoding must decrease the state by the same amount.
Intuitively, when used as a sampler, the ANS state is ``consumed'' to generate the random variate, but can be recovered by encoding the samples back into the state.

\begin{remark}[Decoding Reduces the ANS State]
    Decoding reduces the size of the ANS state by the same amount increased by encoding (up to rounding errors),
    \begin{align}
        -\Delta(s, x, P_X) = \log\left( \frac{s}{s^\prime} \right),
    \end{align}
    where $s = \decode(s^\prime, P_X)$ is the previous state and $P_X$ is the quantized probability distribution.
    For large values of the ANS state, this equals the information content of the symbol,
    \begin{align}
        \lim_{s^\prime \rightarrow \infty} \log \left(\frac{s}{s^\prime}\right) = \log P_X(x) = -\Delta(x, P_X).
    \end{align}
\end{remark}

\section{Bits-Back Coding}\label{chapter:bits-back}
ANS is a compression algorithm that requires access to a quantized probability distribution over the observations (i.e., data).
However, there are probabilistic models where the probability values are not readily available, precluding the direct application of ANS.
The model family considered in this thesis, subject to this constraint, is the family of \emph{latent variable models}.
A probability distribution in this class defines a model over observations indirectly through a conditional distribution, conditioned on a \emph{latent} variable, together with a prior.

\changed{To encode data without access to the marignal distribution over data, we will make use of the ANS state as a sampler; by performing a decode operation with the prior to select an instance for the latent variable.
    The probability distribution used to encode the data with ANS is the conditional probability distribution, conditioned on the decoded instance.
    Decoding to sample reduces the ANS state by approximately the log-probability of the decoded instance under the prior, which allows us to achieve a rate equal to the NELBO.

    Here, the latent variable can be seen as a ``degree of freedom'' of the encoding procedure, as the value of the latent itself is irrelevant; only the instance of the data is of interest.
    This technique, of using the ANS state as a sampler to exploit a degree of freedom for encoding, is known in the literature as \emph{bits-back} or \emph{free-energy} coding \cite{frey1996free, townsend2019practical}.
}

\subsection{Latent Variable Models (LVMs)}\label{chapter:bits-back-lvms}
\begin{definition}[Latent Variable Model]
    A \emph{latent variable model} is a probabilistic model over data defined by a joint distribution $P_{X, Z}$.
    The variable $X$, modeling the data source, and $Z$ are referred to as the \emph{observation} and \emph{latent}, respectively.
    An LVM can be specified indirectly via a \emph{prior}, $Z \sim P_Z$, and \emph{conditional probability}, $P_{X \g Z}$, where the implied model over data is defined via marginalization,
    \begin{align}
        P_X(x) \defeq \sum_{z \in \Z} P_{Z \g X}(z \g x) \cdot P_Z(z) .
    \end{align}
    In general, the probability values from the marginal are assumed to be unavailable as computing the marginalization requires a significant amount of computational resources in practice.
    The model can be extended to a sequence of observations through the i.i.d.\ assumption,
    \begin{align}
        P_{X^n, Z^n}(x^n, z^n) \defeq \prod_{i=1}^n P_{X, Z}(x_i, z_i),
    \end{align}
    which we will default to in this thesis unless specified otherwise.
\end{definition}

\begin{example}[Gaussian LVM]
    A \emph{Gaussian Mixture Model} is an LVM constructed from a weighted combination of $\abs{\Z}$ Gaussian distributions, $P_{X \g Z} = \N(\mu_Z, \sigma_Z)$.
    There is no known analytical expression allowing easy probability evaluations of the posterior, $P_{Z \g X}$, or marginal over data, $P_X$, without explicit marginalization.
\end{example}

Assuming the prior and conditional probability is quantized (for all $z \in \Z$), then it is possible to use ANS to encode the observations by selecting a value for $z$ using some procedure.

\begin{definition}[ANS Code for LVMs]
    Given an LVM specified by a prior $P_Z$ and conditional probability $P_{X \g Z}$, the following procedure defines an ANS code over $x^n \in \X^n$,
    \begin{align}
        s^\prime_i & \defeq \encode(s_{i-1}, x_i, P_{X \g Z}(\cdot \g z_i)) \\
        s_i        & \defeq \encode(s^\prime_i, z_i, P_{Z}),
    \end{align}
    where $z_i \in \Z$ are chosen arbitrarily.
\end{definition}

The latent variable is not an observed quantity and has no meaning outside the model.
If possible, we would encode the observations directly, but the latent variable is necessary as the probability values from the marginal $P_X$ are not available, forcing the use of the conditional probability $P_{X \g Z}$ during coding.
The rate will be a function of the latent sequence $z^n$ used for encoding, and must be known by the decoder to guarantee decodability.
By fixing a common random seed between the encoder and decoder, the latents can be chosen by sampling i.i.d.\ from the prior.
Unfortunately, we can show that the increase in rate will be at least the entropy $H(P_{Z^n})$.
\begin{lemma}[LVM Rate]\label{lemma:rate-for-lvm-with-ans}
    Let $P_{X, Z}$ be an LVM over a data source $X \sim D_X$.
    Construct an ANS code with latents selected by i.i.d.\ sampling from the prior, $Z \sim P_Z$.
    Then, the cross-entropy defining the rate is,
    \begin{align}\label{eq:lvm-gap}
        \E*{-\log P_{X, Z}}
        \geq \E*{-\log P_{X}} + H(P_Z),
    \end{align}
    with equality when the latent is independent of the observation,
    \begin{align}
        \sum_{x \in \X} D_X(x) \cdot {P_{Z \g X}(z \g x)} = P_Z(z).
    \end{align}
\end{lemma}
\begin{proof}
    The cross-entropy decomposes into the cost of encoding with the marginal (if it was available) plus the extra bits spent to encode the latent,
    \begin{align}
        \E{-\log P_{X, Z}}= \E{-\log P_{X}}  + \E{-\log P_{Z \g X}}.
    \end{align}
    The second term is lower bounded by the entropy of the prior,
    \begin{align}
        \E{-\log P_{Z \g X}}
         & = \E{\E{-\log P_{Z \g X} \g X}}                 \\
         & = \E{H(P_Z) + \KL{P_Z}{P_{Z \g X}(\cdot \g X)}} \\
         & = H(P_Z) + \E{\KL{P_Z}{P_{Z \g X}(\cdot \g X)}} \\
         & \geq H(P_Z),
    \end{align}
    with equality when $P_{Z \g X} = P_Z$ for all $x \in \X$.
\end{proof}

To improve the achievable rate we can change the mechanism for choosing the latents.
The lower bound in \Cref{lemma:rate-for-lvm-with-ans} appears due to the i.i.d.\ sampling of the latents.
A better rate can be achieved by sampling from the posterior $P_{Z \g X}$, if available.
This follows directly from applying the inequality \emph{information never hurts} \cite{cover1999elements} to \eqref{eq:lvm-gap} with $Z \sim P_{Z \g X}$,
\begin{align}\label{eq:lvm-rate-posterior}
    \E{-\log P_{Z \g X}} & = \E{\E{-\log P_{Z \g X}(\cdot \g X) \g X}} \\
                         & = \E{H(P_{Z \g X}(\cdot \g X))}             \\
                         & = H(P_{Z \g X})                             \\
                         & \leq H(P_Z)
\end{align}
Intuitively, interpreting the latent as a cluster index, the posterior gives the probability of an observation belonging to that cluster.
The average code length under $P_{X \g Z}$ with latents from the posterior should therefore be smaller on average compared to sampling latents from the prior.

\subsection{Bits-Back with ANS (BB-ANS)}
From \Cref{lemma:rate-for-lvm-with-ans} we can see the rate of an LVM depends on the how well the model $P_X$ can estimate the data distribution $D_X$, measured by the KL divergence,
\begin{align}
    \E{-\log P_{X}} = H(D_X) + \KL{D_X}{P_X}.
\end{align}
The quality of generative models has rapidly improved in recent years \cite{wavenet, salimans2017pixelcnn, razavi2019generating, vahdat2020nvae}.
Latent variable models are particularly attractive for compression applications, because they are typically easy to parallelize.
Some of the most successful learned compressors for large scale natural images are based on deep latent variable models \cite{yang2020improving,townsend2020a}.
Examples include diffusion models \cite{kingma2021variational}, and integer discrete flows \cite{hoogeboom2019integer, berg2020idf++}, leading to state-of-the-art compression performance on image, speech \cite{havtorn2022benchmarking} and smart meter time-series \cite{jeong2022lossless} data.
Most of these models are variations of the family of \emph{Variational Autoencoders} (VAEs) \cite{kingma2013auto}.
\begin{definition}[Variational Autoencoder (VAE) \cite{kingma2013auto}]
    A VAE is an LVM $(P_{X \g Z}, P_Z)$ together with a distribution $Q_{Z \g X}$, called the \emph{approximate posterior}, intended to closely approximate the true posterior $P_{Z \g X}$ in terms of KL-divergence.
\end{definition}
\begin{theorem}[Evidence Lower Bound \cite{jordan1999introduction}]\label{theorem:vae-elbo}
    Given a VAE $\left(P_{X \g Z}, P_Z, Q_{Z \g X}\right)$ over data $X \sim D_X$, the following quantity,
    \begin{align}
        \ELBO{x}
         & = \E*{\log\left(\frac{P_{X \g Z}(x \g Z)\cdot P_Z(Z)}{Q_{Z \g X}(Z \g x)}\right)}                                       \\
         & = \sum_{z \in \Z} Q_{Z \g X}(z \g x) \cdot \log\left(\frac{P_{X \g Z}(x \g z)\cdot P_Z(z)}{Q_{Z \g X}(z \g x)} \right),
    \end{align}
    is a lower bound on the evidence of the data source under the model's marginal,
    \begin{align}
        \ELBO{x} \leq \log D_X(x).
    \end{align}
\end{theorem}
\begin{proof}
    By the non-negativity of the KL-divergence, for $Z \sim Q_{Z \g X}(\cdot \g x)$,
    \begin{align}
        \ELBO{x} & = \E*{\log\left(\frac{P_{X \g Z}(x \g Z)\cdot P_Z(Z)}{Q_{Z \g X}(Z \g x)}\right)}    \\
                 & = \E*{\log\left(\frac{P_{Z \g X}(Z \g x)\cdot P_X(x)}{Q_{Z \g X}(Z \g x)}\right)}    \\
                 & = \E*{-\log\left(\frac{Q_{Z \g X}(Z \g x)}{P_{Z \g X}(Z \g x)}\right)} + \log P_X(x) \\
                 & = -\KL{Q_{Z \g X}(\cdot \g x)}{P_{Z \g X}(\cdot \g x)} + \log P_X(x)                 \\
                 & \leq \log P_X(x).
    \end{align}
\end{proof}
The negation of the ELBO (NELBO) is an upper bound on the state change  under the marginal $P_X$ defined by the VAE,
\begin{align}
    - \log P_X(x) = \Delta(x, P_X) \leq -\ELBO{x}.
\end{align}
The gap between the state change and the ELBO equals the mismatch, measured in KL divergence, between the approximate and true posteriors.
Constructing an ANS code by sampling latents from the posterior achieves a rate higher than \Cref{eq:lvm-rate-posterior} due to this mismatch.
Next, we discuss a construction with asymptotic rate equal to the NELBO.

\begin{definition}[Bits-back with ANS (BB-ANS) \cite{townsend2019practical}]\label{def:bb-ans}
    Given a VAE defined by the quantized probability distributions (\Cref{def:quantized-pmf}),
    \begin{equation}
        Q_{Z \g X}(z \g x) = \frac{q_{z \g x}}{N_Z}, \quad
        P_{X \g Z}(x \g z) = \frac{p_{x \g z}}{N_X}, \quad
        P_Z(z) = \frac{p_z}{N_Z},
    \end{equation}
    BB-ANS encodes a single observation $X=x$, into an existing ANS state $s$, by first decoding a latent from the posterior,
    \begin{align}
        s^\prime, z            & \defeq \decode(s, Q_{Z \g X}(\cdot \g x)),           \label{eq:bb-ans-sample} \\
        s^{\prime\prime}       & \defeq \encode(s^\prime, x, p_{x \g z}, c_{x \g z}),                          \\
        s^{\prime\prime\prime} & \defeq \encode(s^{\prime\prime}, z, p_{z}, c_{z}).
    \end{align}
    Decoding proceeds in a similar fashion but in reverse order,
    \begin{align}
        s^{\prime\prime}, z & = \decode(s^{\prime\prime\prime}, P_Z)               \\
        s^{\prime}, x       & = \decode(s^{\prime\prime}, P_{X \g Z}(\cdot \g z)), \\
        s                   & = \encode(s^\prime, Q_{Z \g X}(\cdot \g x)).
    \end{align}
\end{definition}
BB-ANS defines codes over sequences $x^n \in \X^n$ by encoding all elements $x_i$ to a common state, beginning with an initial state $s_0$.
In the large state regime, decoding as the first step reduces the size of the state by exactly the information content of the latent under the approximate posterior, while encoding increases it by the negative log-probability of the observation and latent under the model's joint distribution.
This is depicted visually in \Cref{fig:bb-ans}.
\begin{figure}[t]
    \centering
    \includegraphics[width=\textwidth]{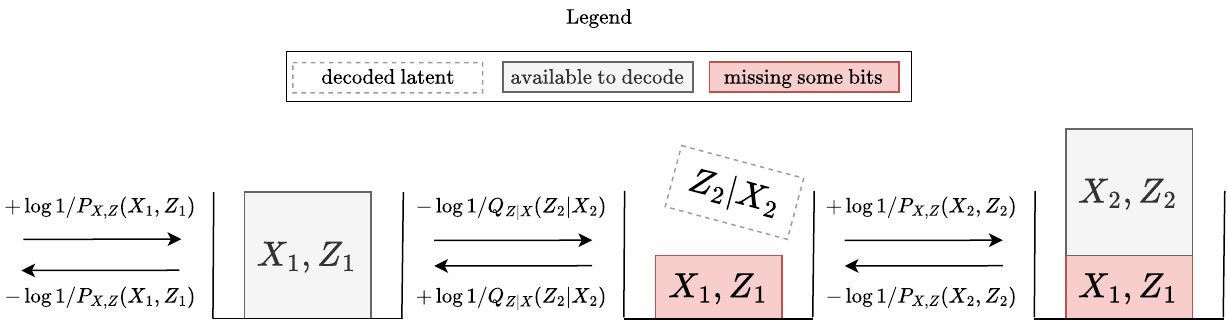}
    \caption{ANS state change under BB-ANS.}
    \label{fig:bb-ans}
\end{figure}

For an initial state $S_0 = s_0$, if $Z_i$ are the random variables produced during the encoding of a sequence $X^n$ using BB-ANS, the state change in the large state regime is equal to,
\begin{equation}
    -\sum_{i=1}^n \log\left(\frac{P_Z(Z_i) \cdot P_{X \g Z}(X_i \g Z_i)}{Q_{Z \g X}(Z_i \g X_i)}\right).
\end{equation}
The randomness of this scheme is due solely to the data sequence $X^n$ and initial state $S_0$.
The latent sequence $Z^n$ is deterministic conditioned on both these values.
If $Z_i$ is distributed according to the sampling distribution $Q_{Z \g X}$, then the change in state from encoding $X=x$ would equal $\E*{\ELBO{X}}$ in expectation.
In its current form presented, the determinism of the latent precludes us from making useful probabilistic statements regarding the correctness of sampling.
Correctness is measurable, to some extent, in terms of the KL-divergence between the empirical distribution of the latents and the approximate posterior, as a function of $x$,
\begin{align}
    \KL{Q_{Z \g X}(\cdot \g x)}{P[z^n]}.
\end{align}
This issue is often referred to as the \emph{dirty bits} issue in existing literature \cite{townsend2019practical,ruan2021improving}.
As of this writing, follow-up papers applying BB-ANS to real-world domains have found the value of the KL divergence to be small and the dirty bits are not a significant issue affecting the rate \cite{townsend2019practical,ruan2021improving, townsend2020a, townsend2020tutorial,townsend2021lossless,severo2023compressing,severo2023random}.
Nonetheless, modifying the sampling step \eqref{eq:bb-ans-sample} to add i.i.d.\ uniform noise to the state before sampling the latent is sufficient to guarantee correctness,
\begin{align}\label{eq:ans-add-noise-to-state}
    S^\prime, Z & \defeq \decode(s + U, Q_{Z \g X}(\cdot \g x)),
\end{align}
where $U$ is a uniform random variable in the interval $[-N_Z \div 2, N_Z \div 2]$ if $N_Z$ is odd, and $[-N_Z \div 2, (N_Z \div 2) - 1]$ if $N_Z$ is even.
From \Cref{lemma:ans-sampling} we know that $Z \sim Q_{Z \g X}(\cdot \g x)$ as $s + U$ is uniform in an interval of size $N_Z$.
Extending this to sequences, we sample $U_i \sim U$ i.i.d.\ from a common source of randomness between the encoder and decoder.
The encoder adds $U_i$ to state $S_i$ at each step to encode a symbol $X_i$, and the decoder, knowing $U_i$, removes it before sampling the latent.
The increase in the average state change due to adding noise will be close to zero.
The expected state change at each step, in the large state regime, will equal the ELBO.
The law of large numbers \cite{durrett2019probability} guarantees the state change per symbol converges, with probability $1$, to the ELBO as well.
\begin{lemma}[BB-ANS Rate]\label{lemma:bb-ans-rate}
    Let $\left(P_Z, P_{X \g Z}, Q_{Z \g X}\right)$ be a VAE over an i.i.d\ data source $X_i \sim D_X$.
    Then, if $Z_i \sim Q_{Z \g X}(\cdot \g X_i)$ are the random variables produced during the encoding procedure of BB-ANS in \Cref{def:bb-ans},
    \begin{equation}
        \lim_{n \rightarrow \infty} \frac{1}{n} \cdot \sum_{i=1}^n \log\left(\frac{P_Z(Z_i) \cdot P_{X \g Z}(X_i \g Z_i)}{Q_{Z \g X}(Z_i \g X_i)}\right)
        = \E{\ELBO{X}},
    \end{equation}
    with probability $1$, as well as,
    \begin{equation}
        \E*{\log\left(\frac{P_Z(Z_i) \cdot P_{X \g Z}(X_i \g Z_i)}{Q_{Z \g X}(Z_i \g X_i)}\right)}
        = \E{\ELBO{X}},
    \end{equation}
    for all $i$.
\end{lemma}
\begin{proof}
    For large state values $s \in \Naturals$ the change in $\ceil{\log(s)}$ is equal to the negative log-probability of the quantized probability distribution in use (\Cref{theorem:optimality-ans}).
    Tallying the change due to $\decode$ and both $\encode$ operations gives the ratio of probability distributions shown in \Cref{lemma:bb-ans-rate}.
\end{proof}

\subsection{Initial Bits and State Depletion}\label{sec:bbans-initial-bits}
The first step in encoding a symbol with BB-ANS is to sample/decode the latent, contributing bit savings to the rate due to the decrease in size of the ANS state.
For any sequence $X^n$, the savings can be seen from \Cref{lemma:bb-ans-rate} to be,
\begin{equation}
    -\sum_{i=1}^n \log Q_{Z \g X}(Z_i \g X_i).
\end{equation}
In expectation, and as $n \rightarrow \infty$, this quantity converges to the conditional entropy of the approximate posterior: $H(Q_{Z \g X})$, when normalized by $n$.
This assumes the distribution of the sampled latent will be equal to the approximate posteriors, $Z_i \sim Q_{Z \g X}(\cdot \g X_i)$.
For this to be achieved with the addition of uniform noise to the state, as shown in \Cref{eq:ans-add-noise-to-state}, the ANS state must be above a minimal value at all times.
More specifically, we must guarantee the state is at least \changed{$N_Z \div 2$} for all $i$ with probability 1.
When this condition is not met it is common to refer to the state as being ``depleted of randomness" \cite{townsend2019practical}, as it can not be used to sample without first increasing it to the minimal value artificially, increasing the rate as well.

As of the writing of this manuscript there is no known method that is guaranteed to avoid the initial bits problem.
Instead, the method adopted here is to decompose the latent and observation into sub-variables.
The latent is replaced for a sequence of latents of known length,
\begin{align}
    Z_i \defeq (Z_i^{(1)}, \dots, Z_i^{(k_Z)}),
\end{align}
as well as the observation
\begin{align}
    X_i \defeq (X_i^{(1)}, \dots, X_i^{(k_X)}).
\end{align}
To perform source coding with BB-ANS requires that the approximate posterior, conditional probability, and priors be defined for the sequence elements.
To encode $X_i$ with BB-ANS, first $Z_i^{(1)}$ is decoded with the approximate posterior, conditioned on $X_i$, from the ANS state.
Then, $X_i^{(1)}$ is encoded conditioned on $Z_i^{(1)}$.
Next, $Z_i^{(2)}$ is decoded, conditioned on $(Z_i^{(1)}, X_i)$, and $X_i^{(2)}$ is encoded conditioned on $(Z_i^{(1)}, Z_i^{(2)})$.
This algorithm continues until $X_i$ is encoded.

The negative log-probability of the sub-latent sequence is the same as that of the latent $Z_i$, implying the reduction in the ANS state will be the same in the large state regime.
This technique reduces the probability of the initial bits problem as it increases the ANS state, by encoding a sub-observation, before each sampling each sub-latent.
In \Cref{chapter:roc}, \Cref{chapter:rcc}, and \Cref{chapter:rec} we give examples of this technique and show how it can be simplified to reduce the number of required distributions and computational as well as memory resources.

\changed{
    \subsection{Discussion}
    The key aspect of ANS that enables bits-back coding is its ``stack-like" nature, in the sense that the first symbol decoded from the state is the last symbol to be encoded.
    The decoder implements the exact inverse of the steps as done during encoding, which is only possible due to this first-in-last-out mechanism of the stack.
    This precludes the use of more traditional entropy coders such as Arithmetic Coding \cite{witten1987arithmetic,cover1999elements}, which are ``queue-like''; the first element encoded is the first to be decoded.

    It is possible to use bits-back with prefix-free codes, such as Huffman Codes \cite{huffman1952method}, by constructing a code with the prior $P_Z$, as well as codes with the conditional probability $P_{X \g Z}(\cdot \g z)$, for every $z \in \Z$, and, similarly, with the approximate probability $Q_{Z \g X}(\cdot \g x)$, for every $x \in \X$.
    Similar to BB-ANS, the encoder initially samples random bits independently from a Bernoulli distribution with parameter $\frac{1}{2}$, and uses this to sample a latent value $z_1$ with the approximate posterior probability distribution $Q_{Z \g X}(\cdot \g x_1)$; where $x_1$ is the instance of the first element of the sequence to be encoded.
    This results in savings of approximately the negative log-probability under the approximate posterior.
    The element $x_1$ is then encoded with the code for $P_{X \g Z}(\cdot \g z_1)$, followed by the sampled latent $z_1$ with the code constructed from the prior $P_{Z}$.%
    \footnote{\changed{A full example of Bits-back with Huffman Coding, in Python, is available at \url{https://gist.github.com/dsevero/8e7c38b44953964d3b9873b6bd96d9b2}}}

    Empirically, in the experiments considered in this thesis, the depletion of the ANS state is rarely observed.
    In all experiments run in this manuscript, the depletion occurs only at the very first sampling step (i.e., to sample the first latent value), as expected.
    Nonetheless, since the encoder knows exactly what the decoder will observe, as the decoder implements the exact inverse steps, it is possible to switch to a backup compression scheme if needed.
    In the worst case, if the state is completely depleted at all steps, then the savings relative to decoding with the approximate posterior will not occur (i.e., the term in the denominator of \Cref{lemma:bb-ans-rate}).
    The achieved rate, in the large ANS state regime, will therefore be that of the cross-entropy between the product distributions $P_X \cdot Q_{Z \g X}$ and $P_{X \g Z} \cdot P_Z$.
}%

\chapter{Combinatorial Objects}\label{chapter:combinatorial-objects}
This manuscript makes extensive use of many notions of combinatorial objects such as sets, multisets, permutations, orderings, and others.
This chapter makes the basic definitions used throughout the manuscript and helps clarify the notation used.

\section{Equivalence Relations}
\begin{definition}[Equivalence Relation]\label{def:equivalence-relation}
    An equivalence relation on an arbitrary set $\X$ is any binary relationship between elements of $\X$ that is
    \begin{itemize}
        \item Reflexive: $x \sim x$,
        \item Symmetric: if $x \sim y$, then $y \sim x$,
        \item Transitive: if $x \sim y$ and $y \sim z$, then $x \sim z$,
    \end{itemize}
    for any $x,y,z \in \X$.
\end{definition}
Equivalence relations are used in the later chapters to define equivalence amongst sequences representing the same combinatorial object.
Any equivalence relation partitions its set into disjoint subsets, where two elements are in the same set if, and only if, they are equivalent.
\begin{definition}[Quotient Set]
    Given a set $\X$ paired with an equivalence relation $\sim$, the quotient set $\quotient{\X}{\sim}$ is a partition of $\X$ into disjoint sets $\X_i$ such that $x,y \in \X_i$ if, and only if, $x \sim y$.
    The subsets $\X_i$ are known as \emph{equivalence classes}.
\end{definition}
\begin{example}[Equivalence Relation]\label{example:equivalence-relation}
    For any $x, y \in \Naturals$, let $x \sim y$ if, and only if, $x$ and $y$ have equal parity (i.e., are both even or odd).
    This relationship is clearly an equivalence relation.
    The quotient set $\quotient{\Naturals}{\sim}$ is composed of two sets, one of all even integers, and the other of all odd integers.
\end{example}
\begin{definition}[Finer Equivalence Relation]\label{def:finer-equivalence-relation}
    Let $\sim$ and $\sim^\prime$ be equivalence relations on $\X$. Then, $\sim^\prime$ is said to be \emph{finer} than $\sim$ if any equivalence class in $\quotient{\X}{\sim^\prime}$ is a strict subset of some equivalence class in $\quotient{\X}{\sim}$; and \emph{as fine as} when equal.
\end{definition}
Finer equivalence relations can be constructed from an initial equivalence relation by including the condition of the latter into the definition of the former, as the next example shows.
\begin{example}[Finer Equivalence Relation]
    For any $x, y \in \Integers_{\neq 0} = \Integers\setminus\{0\}$, let $x \sim y$ if, and only if, $x$ and $y$ have the same sign  (i.e., both are positive or negative).
    This relationship is clearly an equivalence relation.
    Define $x \sim^\prime y$ if, and only if, $x \sim y$ and $x$ and $y$ have the same parity.
    These relations partition the set $\Integers$,
    \begin{align}
        \quotient{\Integers}{\sim}        & = \{\Integers_{-}, \Integers_{+}\}                                                                                       \\
        \quotient{\Integers}{\sim^\prime} & = \{\Integers_{-}^{\text{even}}, \Integers_{-}^{\text{odd}}, \Integers_{+}^{\text{even}},  \Integers_{+}^{\text{odd}}\}.
    \end{align}
    Since $\Integers_{-}^{\text{even}}, \Integers_{-}^{\text{odd}} \in \Integers_{-}$, and $\Integers_{+}^{\text{even}}, \Integers_{+}^{\text{odd}} \in \Integers_{+}$, then $\sim^\prime$ is finer than $\sim$.
\end{example}

\section{Permutations on Sets}
\begin{definition}[Permutations on Sets]\label{def:permutations-on-sets}
    A \emph{permutation} on $n$ elements is a bijective function $\sigma\colon [n] \mapsto [n]$ used to define arrangements of elements from arbitrary sets.
    Permutations are usually expressed in one-line notation,
    \begin{equation}
        \sigma \defeq [i_1, i_2, \dots, i_n], \text{ where } \sigma(j) = i_j.
    \end{equation}
    There are $n! = n \cdot (n-1) \cdot (n-2) \cdots 1$ unique permutations on $n$ elements.
\end{definition}
\begin{example}[Permutations on Sets]
    Three permutations on $2$, $3$, and $4$ elements are shown below together with their one-line notations.
    The second is known as the \emph{identity} permutation, which maps elements to themselves, $\sigma_2(i) = i$.
    \begin{align}
         & \sigma_1 \defeq [2, 1] &  & \sigma_2 \defeq [1, 2, 3] &  & \sigma_3 \defeq [3, 1, 2, 4] \\
         & \sigma_1(1) = 2        &  & \sigma_2(1) = 1           &  & \sigma_3(1) = 3              \\
         & \sigma_1(2) = 1        &  & \sigma_2(2) = 2           &  & \sigma_3(2) = 1              \\
         &                        &  & \sigma_2(3) = 3           &  & \sigma_3(3) = 2              \\
         &                        &  &                           &  & \sigma_3(4) = 4
    \end{align}
\end{example}

\begin{definition}[Cycle]
    A \emph{cycle} $(c_1, \dots, c_k)$, of a permutation $\sigma$, is the sequence constructed from the repeated application of $\sigma$, to some element $c_1 \in [n]$, until $c_1$ is recovered, i.e., $c_{k+1} = c_1$,
    \begin{equation}
        c_i \defeq \sigma(c_{i-1}), \text{ for } i \geq 2.
    \end{equation}
    It is common to drop the commas in the notation of cycles.
    The \emph{size} of a cycle is the number of elements it contains, i.e., $k$.
\end{definition}
\begin{example}[Cycles]
    The cycles of $\sigma \defeq [3, 1, 2, 5, 4]$, for each $c_1 \in [5]$, are shown below.
    \begin{table}[!h]
        \centering
        \begin{tabular}{cccccc}
            $c_1$ & $1$         & $2$         & $3$         & $4$      & $5$      \\
            cycle & $(1\ 3\ 2)$ & $(2\ 1\ 3)$ & $(3\ 2\ 1)$ & $(4\ 5)$ & $(5\ 4)$
        \end{tabular}
    \end{table}
\end{example}
\begin{definition}[Cycle Notation]
    A permutation can be represented by its cycles with the following procedure.
    Pick any element in $[n]$ and compute its cycle by applying $\sigma$ successively.
    Next, choose another element in $[n]$, that did not show up in any of the previously computed cycles, and compute its cycle.
    Repeat this procedure until all elements appear in exactly one cycle.
    Concatenate all cycles to form the representation,
    \begin{equation}
        \sigma \defeq (c^1_1\ \dots\ c^1_{n_1})(c^2_1\ \dots\ c^2_{n_2}) \dots (c^\ell_1\ \dots\ c^\ell_{n_\ell}),
    \end{equation}
    where $c^i_j$ is the $j$-th element in the $i$-th cycle out of $\ell$.
\end{definition}
\begin{example}[Cycle Notation]\label{example:cycle-notation}
    $\sigma \defeq [3, 1, 2, 5, 4]$ can be represented by any of the following $12$ combinations,
    \begin{align}
        (1\ 3\ 2)(4\ 5) &  & (2\ 1\ 3)(4\ 5) &  & (3\ 2\ 1)(4\ 5) &   \\
        (1\ 3\ 2)(5\ 4) &  & (2\ 1\ 3)(5\ 4) &  & (3\ 2\ 1)(5\ 4) &   \\
        (4\ 5)(1\ 3\ 2) &  & (4\ 5)(2\ 1\ 3) &  & (4\ 5)(3\ 2\ 1) &   \\
        (5\ 4)(1\ 3\ 2) &  & (5\ 4)(2\ 1\ 3) &  & (5\ 4)(3\ 2\ 1) & .
    \end{align}
\end{example}

\begin{theorem}[Foata's Bijection \cite{foata1968netto}]\label{theorem:foata}
    The following sequence of operations defines a bijection between permutations on $n$ elements.
    Write the permutation in cycle notation such that the smallest element of each cycle appears first within the cycle.
    Order the cycles in decreasing order based on the first, i.e., smallest, element in each cycle.
    Remove all parenthesis to form the one-line notation of the output permutation.
    Cycles can be recovered by scanning from left to right and keeping track of the smallest value.
\end{theorem}
Writing a permutation in the form described by Foata's Bijection, before removing parenthesis, is known as ``canonical cycle notation".
\begin{example}[Foata's Bijection - 1]
    $\sigma \defeq [3, 1, 2, 5, 4]$, from \Cref{example:cycle-notation}, in canonical cycle notation is $(4\ 5)(1\ 3\ 2)$.
    The image under Foata's bijection, expressed in one-line notation, is $\pi = [4, 5, 1, 3, 2]$.
    To recover the cycle notation of $\sigma$ from $\pi$ we traverse from left-to-right keeping track of the smallest element.
    A change in the smallest element during the traversal marks the beginning of a new cycle.
\end{example}

\begin{example}[Foata's Bijection - 2]
    Foata's permutation for the identity $\sigma \defeq [1, 2, \dots, n] = (n)(n-1)\dots (1)$ is the reversal $\pi = [n, n-1, \dots, 1]$.
\end{example}

\section{Total Orderings and Ranks}\label{sec:orderings-ranks}
A set does not impose an ordering between its elements.
Order between elements of a set can be introduced through the notion of a \emph{total ordering}.
A total ordering can be specified by writing the elements of the set in a desired order and defining $x \leq y$ if $x = y$ or $x$ comes before $y$ in the sequence.
Since every sequence defines a unique ordering, there are $n!$ possible orderings for elements of a set of size $n$.
\begin{definition}[Total Order]\label{def:total-order}
    A \emph{total order} on an arbitrary, finite, set $\X$ is a binary relation $\leq$ that is reflexive ($x \leq x$), transitive ($x \leq y$ and $y \leq z$ implies $x \leq z$), antisymmetric ($x \leq y$ and $y \leq x$ implies $x = y$), and at least one of $x \leq y$ or $y \leq x$ is true, for any pair $x, y \in X$.
\end{definition}
\begin{example}[Total Order - 1]
    For $\X = \Naturals$, the usual $\leq$ between integers is a total order.
\end{example}
\begin{example}[Total Order - 2]
    Let $\leq$ be a binary relation on $\binstrs$ such that $x \leq y$, if $x$ has a shorter length than $y$ or, if their lengths are equal, if the position of the first $1$ in $x$ is larger than or equal to the position of the first $1$ in $y$.
    Then, $\leq$ is the total order known as the \emph{lexicographic order}.
\end{example}
\begin{definition}[Ordered Set]
    An \emph{ordered set} is a pair $(\X, \leq)$, where $\X$ is a set and $\leq$ a total ordering over $\X$.
    Alternatively, an ordered set is a pair $(\X, x^n)$, where $\abs{\X} = n$, and the sequence $x^n$ is an ordering of all elements of the set; $x_i \in \X$, $x_i \neq x_j$ for $i \neq j$.
\end{definition}
\begin{definition}[Rank]
    Given a set and a total ordering, the \emph{rank} of an element is its position in the sequence defining the total order.
\end{definition}
\begin{example}[Ordered Set]
    The natural ordering on characters $\X \defeq \{a, b, c, \dots, z\}$ is defined by the sequence $(a, b, c, \dots, z)$. The rank of $c$ is $3$.
\end{example}
Permutations can be redefined as operators acting on a set equipped with an ordering.
Applying a permutation yields a new total order defined by acting on the order's sequence.
\begin{definition}[Permutations Act on Ordered Sets]
    A permutation $\sigma$ acts on a totally ordered set $(\X, x^n)$ by mapping it to $(\X, y^n)$, where, $y_i \defeq x_{\sigma(i)}$.
\end{definition}

\section{Permutations on Multisets}
A permutation defines an ordering of a set, a collection of \emph{unique} elements, via a bijection.
This definition can be extended to \emph{multisets}, which are sets that allow repeated elements.
Following \Cref{sec:orderings-ranks} we define permutations on multisets of arbitrary elements, without restricting to contiguous sub-intervals of integers, assuming a total order is given implicitly from context.
\begin{definition}[Multiset]\label{def:multisets}
    A \emph{multiset} $\M$ is a set equipped with multiplicities $\M(x) = n_x \geq 1$, indicating the amount of repeats for each element $x$ in the set.
    The \emph{size} of a multiset, $\abs{\M}$, is equal to the sum of multiplicities $\sum_{x} n_x$.
\end{definition}
\begin{example}[Multiset]
    A multiset with multiplicities $n_a \defeq 2, n_b \defeq 3, n_c \defeq 1$ can be represented in (multi)set notation as
    \begin{equation}
        \M = \{a, a, b, b, b, c \}.
    \end{equation}
\end{example}

\begin{definition}[Permutations on Multisets]\label{def:permutations-on-multisets}
    A \emph{multiset permutation} on some multiset $\M$, with multiplicities $n_x$, is a \emph{partition} of the interval $[\abs{\M}]$.
    Each subset indicates the position in the permutation of the $n_x$ repeats of element $x \in \M$.
    This is defined by,
    \begin{align}
         &  & \sigma\colon \M \mapsto 2^{[\abs{\M}]},
         &  & \sigma(x) \cap \sigma(x^\prime) = \emptyset, \text{ for } x \neq x^\prime,
         &  & \abs{\sigma(x)} = n_x,
         &  & \bigcup_{x \in \M} \sigma(x) = [\abs{\M}].
    \end{align}
    The one-line notation of a multiset permutation is the sequence,
    \begin{align}
        \sigma = [i_1, \dots, i_n],
    \end{align}
    similar to \Cref{def:permutations-on-sets}, where $\sum_{j=1}^n \1\{i_j = x \} = n_x$.
    Similar to sets, a total order equipped with a multiset, $(\M, \leq)$, can be defined by a sequence of elements in $\M$.
    When clear from context, the sequence can also be made up of the ranks.
\end{definition}
\begin{example}[Multiset Permutation]
    For the multiset
    \begin{align}
        \M \defeq \{a, a, a, b, b, c \},
    \end{align}
    $\sigma_1$ and $\sigma_2$ are multiset permutations,
    \begin{align}
        \sigma_1(a) & = \{2, 4, 6\} & \sigma_2(a) & = \{1, 2, 3\} \\
        \sigma_1(b) & = \{1, 3\}    & \sigma_2(b) & = \{4, 5\}    \\
        \sigma_1(c) & = \{5\}       & \sigma_2(c) & = \{6\},
    \end{align}
    with one-line notations,
    \begin{align}
        \sigma_1 = [b, a, b, a, c, a] \quad\quad\quad
        \sigma_2 = [a, a, a, b, b, c].
    \end{align}
    $\sigma_2$ defines the total ordering corresponding to the usual order over alphabetical symbols.
\end{example}
A set permutation is a multiset permutation where all multiplicities are equal to $1$.
Elements are indistinguishable if they are equal, which reduces the count of total possible permutations when comparing to set permutations.
\begin{definition}[Multinomial Coefficient]\label{def:multinomial-coeff}
    The number of possible multiset permutations for a multiset $\M$, with multiplicities $n_x$, is known as the \emph{multinomial coefficient},
    \begin{equation}
        \binom{n}{n_1 \dots n_k} \defeq \frac{n!}{\prod_{i=1}^k n_i!} \leq n!,
    \end{equation}
    with equality when all multiplicities are equal to one.
\end{definition}

\section{Clusters and Partitions}
\emph{Clustering} is the act of partitioning a set of arbitrary elements into multiple pair-wise disjoint groups, called \emph{clusters}.
Clusters are void of labels and are defined only by the elements they contain.
\begin{definition}[Clustering / Partitioning]\label{def:clustering}
    A \emph{clustering} of objects from a set $\X$ is defined by an indicator function,
    \begin{align}
        \pi\colon \X\times\X & \mapsto \{0, 1\},        \\
        \pi(x, x^\prime)     & \defeq \pi(x^\prime, x), \\
        \pi(x, x)            & \defeq 1,
    \end{align}
    where $\pi(x, x^\prime) = 1$ if $x, x^\prime \in \X$ are co-located in the same cluster.
    The relationship must also be transitive, meaning if $\pi(x, z) = \pi(x, y) = k$, then $\pi(z, y) = k$.
    Equivalently, a clustering is a \emph{partitioning} of the set $\X$ into disjoint subsets, corresponding to each cluster.
    Both interpretations are used interchangeably throughout this manuscript.
    The number of different ways to partition a set of size $n$ is equal to the $n$-th Bell Number \cite{aitken1933problem}, which grow exponentially fast in $n$, but is upper bounded by $n!$.
\end{definition}
\begin{example}[Clustering / Partitioning]
    All possible partitionings of $\X \defeq \{a, b, c\}$ are shown below.
    \begin{align}
        \{a\} \{b\} \{c\} \quad \{a\} \{b, c\} \quad  \{b\} \{a, c\} \quad   \{c\} \{a, b\} \quad  \{a, b, c\}
    \end{align}
    Note that $\{a\} \{b, c\}$ and $\{b, c\}\{a\}$ are the same partitioning.
\end{example}

\section{Graphs}\label{sec:combinatorial-objects-graphs}
Graphs are defined by a collection of vertices connected by edges representing a network structure.
The set of vertices can be any arbitrary set and is referred to as the \emph{vertex set}.
Edges are defined as collections of vertices and can be both ordered and unordered depending on the type of graph.
Our definition of graphs builds on the following definitions for \emph{vertex} and \emph{edge} sequences, discussed below.
First, we informally define the graph types used throughout this manuscript.

A \emph{directed graph} on a vertex set $V$ is a collection of vertex pairs $(v, w) \in V^2$ known as edges.
The order between vertices in an edge defines the orientation of the edge in the network, pointing from $v$ to $w$.

An \emph{undirected graph} on a vertex set $V$ can be constructed by removing the order information from the edges of a directed graph.
Each edge is a multiset, $\{v , w\}$.
A \emph{simple graph}, directed or undirected, is a graph with no self-edges, i.e., $v \neq w$ for every edge.

\begin{example}[Directed, Undirected, and Simple Graphs]
    \Cref{fig:combinatorial-objects-graphs} shows an example of directed and undirected graphs.

\end{example}
\begin{figure}[!h]
    \centering
    \begin{tikzpicture}
        \graph [math nodes, nodes={draw,circle}] {
        a -> {b, c} -> d;
        a ->[loop above] a;
        };
        \quad\quad
    \end{tikzpicture}
    \quad\quad
    \begin{tikzpicture}
        \graph [math nodes, nodes={draw,circle}, every loop/.style={}] {
        a -- {b, c} -- d;
        };
    \end{tikzpicture}
    \caption{A non-simple directed graph (left) and simple undirected graph (right).}
    \label{fig:combinatorial-objects-graphs}
\end{figure}
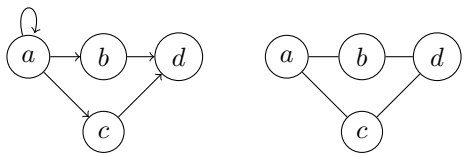


Graphs can be specified by either the vertex sequence or edge sequence,
\begin{align}
    (v_1, v_2, \dots, v_{2m}) \equiv (e_1, e_2, \dots, e_m),
\end{align}
where $e_1 = (v_1, v_2), e_2 = (v_3, v_4)$, and so on, $e_i = (v_{2 i-1}, v_{2 i})$, for directed graphs; and the equivalent set definition for the undirected variant: $e_i = \{v_{2 i-1}, v_{2 i}\}$.

There are multiple sequences that can represent the same graph.
Swapping the position of any two edges, and, for undirected graphs, swapping the position of vertices that are within the same edge, will not change graph.
This leads to the definition of graphs as \emph{equivalence classes} of vertex sequences.
\begin{example}[Graphs as Equivalence Classes]
    The graphs in \Cref{fig:combinatorial-objects-graphs} can be represented by the following sequences.
    The parenthesis within edges has been removed to highlight the edge set for the non-simple, direct graph.
    \begin{center}
        \begin{tabular}{c c}
            Non-simple, directed   & Simple, undirected                         \\
            $(aa, ab, ac, bd, cd)$ & $(\{a, b\}, \{a, c\}, \{b, d\}, \{c, d\})$ \\
            $(ab, aa, ac, bd, cd)$ & $(\{a, c\}, \{a, b\}, \{b, d\}, \{c, d\})$ \\
            $(aa, ac, ab, bd, cd)$ & $(\{a, c\}, \{a, b\}, \{c, d\}, \{b, d\})$ \\
            $(aa, ab, ac, db, cd)$ & $(\{a, b\}, \{c, d\}, \{b, d\}, \{a, c\})$ \\
            $\dots$                & $\dots$
        \end{tabular}
    \end{center}
\end{example}

\begin{definition}[Directed Graph]\label{def:rec-directed-graphs-equivalence-class}
    A \emph{directed graph} of $m$ edges is a set of equivalent sequences, with equivalence denotes by $\sim$, where $v^{2m} \sim w^{2m}$ if their edge sequences are equal up to a permutation,
    \begin{align}
        \left(e^{(v)}_1, \dots, e^{(v)}_m\right) = \left(e^{(w)}_{\sigma(1)}, \dots, e^{(w)}_{\sigma(m)}\right),
    \end{align}
    where $e^{(v)}_i = (v_{2i-1}, v_{2i}), e^{(w)}_i=(w_{2i-1}, w_{2i})$ are edges and $\sigma$ is a permutation over $m$ elements.
\end{definition}
\begin{definition}[Undirected Graph]\label{def:rec-undirected-graphs-equivalence-class}
    An \emph{undirected graph} is a set of equivalent sequences where $v^{2m} \sim w^{2m}$ if they are equal up to a permutation of edges, similar to a directed graph, and permutation of vertices within an edge.
\end{definition}
\chapter{Random Order Coding}\label{chapter:roc}
Lossless compression algorithms typically preserve the ordering of compressed symbols in the input sequence.
However, there are data types where order is not meaningful, such as collections of files, rows in a database, nodes in a graph, and, notably, datasets in machine learning applications.
Formally, these may be expressed as a mathematical object known as a \emph{multiset}: a generalization of a set that allows for repetition of elements.

Compressing a multiset with an entropy coder is possible if we somehow order its elements and communicate the corresponding ordered sequence.
However, unless the order information is somehow removed during the encoding process, this procedure will be sub-optimal, because the order contains information and therefore more bits are used to represent the source than are truly necessary.

The fundamental limits of multiset compression are well understood and were first investigated by \cite{varshney2006}.
The information content of a non-trivial multiset is strictly less than that of a sequence with the same elements, by the number of bits required to represent an ordering, or permutation, of the elements.
However, previous optimal-rate algorithms for multiset compression have computation time which scales linearly with the size of the alphabet from which elements are drawn \cite{Steinruecken2016-oy}.

In this chapter, we show how to compress multisets of statistically exchangeable symbols at an optimal rate, with computational complexity independent of the alphabet size, by converting an existing algorithm for sequences into one for multisets.
This enables us to compress fixed-size multisets of independent and identically distributed (i.i.d.) symbols\footnote{Any i.i.d. sequence is also necessarily exchangeable.} with arbitrarily large alphabets, including multisets of images, where the alphabet size scales exponentially with the number of pixels; and of strings, where it scales exponentially with string length.

The key insight is to avoid encoding the multiset directly, and instead encode a random sequence containing the same elements as the multiset, in $\Omega(n)$ steps, where $n$ is the sequence and multiset size.
This can be done by extending the bits-back coding algorithms \cite{frey1996free, frey1997,townsend2019practical} discussed in \Cref{chapter:bits-back}, where encoding and decoding operations are interleaved during compression and decompression.
During compression, symbols are sampled without replacement from the multiset and encoded sequentially in the order they are sampled.
Sampling is done using an ANS decoder, with the already compressed information used as the random seed for the sample via the ANS state.
This procedure is invertible, because the bits used as the random seed can be losslessly recovered during decompression, using an ANS encoder.
The sampling step consumes bits, reducing the message length by exactly the number of bits needed to represent a permutation.

Our method is attractive in settings where the alphabet size is large, and the symbol distribution assigns most probability mass to a sparse subset of the alphabet.
We anticipate that databases and other unordered structured data, in formats like JavaScript Object Notation (JSON), might be realistic use-cases, while for multisets of large objects, such as images or video files, the savings are marginal.
We investigate these settings experimentally in \Cref{sec:roc-experiments}.

Although our method works for any source of exchangeable symbols, the following sections assumes symbols are i.i.d. to simplify the exposition.

\section{Related Work}\label{sec:roc-related-work}
To the best of our knowledge, there are no previous works that present a method which is both computationally feasible and rate-optimal for compressing multisets of exchangeable symbols with large alphabets.

The fundamental limits on lossless compression of multisets were investigated in \cite{varshney2006}.
For finite alphabets, a sequence can be decomposed into a multiset $\M$ (also known as a \emph{type}) and a permutation conditioned on the multiset.
With this, the authors are able to prove a combinatorial bound on the entropy: $H(\M) \leq \Omega(\abs{\X} \cdot \log\abs{\M})$.
However, a practical algorithm for achieving this bound is not presented.

A rate-optimal algorithm for compressing multisets of i.i.d\ symbols was put forth in \cite{Steinruecken2016-oy}.
There, the information content of the multiset is decomposed recursively over the entire alphabet, allowing it to be interfaced to an entropy coder known as Arithmetic Coding \cite{witten1987arithmetic}.
Alternatively, this can also be achieved through Poissonification, even for arbitrary-sized multisets \cite{steinruecken2014b, yang2014compression, yang2017minimax}.
Although rate-optimal, the computational complexity of these methods scales linearly with alphabet size.

Another approach \cite{gripon2012compressing, Steinruecken2014-zs, reznik2011coding} is to convert multisets of sequences (e.g., cryptographic hashes) to an order-invariant data structure, such as a tree, that is losslessly compressible.
These methods can compress arbitrary multisets of i.i.d.\ symbols by first encoding each symbol with a prefix-free code and constructing the order-invariant data-structure from the multiset of resulting code-words.
If encoding each symbol individually is optimal (e.g., Huffman coding for dyadic sources), then the method presented in \cite{Steinruecken2014-zs} is optimal and has the same complexity as our method.
However, this is not generally the case, and the overhead of compressing each symbol with a prefix-free code may dominate.
Our method can be seen as a generalization of \cite{Steinruecken2014-zs} that allows symbols to be compressed in sequence, effectively removing the overhead.

Compressing a multiset is equivalent to compressing a \emph{sequence with known order} \cite{steinruecken2014b}, as both objects have equal information content.
The transmitter and receiver can agree that sequences will be sorted before encoding, which removes the order information.
An optimal coding scheme could compress sorted symbols sequentially with an adaptive symbol distribution that accounts for sorting.
However, like the method in \cite{Steinruecken2016-oy}, this approach does not scale to multisets with large alphabet due to the high complexity cost of adapting the symbol distribution at each encoding step \cite{steinruecken2014b}.

An alternative to removing the order information is to transform the sequence in a way that aids compression.
A common choice of transformations is to sort and apply run-length encoding (RLE) \cite{robinson1967results}.
The output is a set of symbol-frequency pairs, known as \emph{run-length symbols}, which has the same information content as the multiset and an alphabet of size $\abs{\X} \cdot n$.
Compressing the set of run-length symbols as a sequence with known order is computationally intensive (as previously discussed), but can be done efficiently with our method.
However, the entropy of any transformation is lower bounded by that of the multiset, as we require that the multiset be recoverable.
Therefore, these techniques can not improve the compression rate, although they might outperform our method for carefully chosen examples.

A related but more complex setting is \emph{lossless dataset compression}, where the order between correlated examples in a dataset is irrelevant.
A sub-optimal, two-step algorithm, that first re-orders the dataset based on similarity between examples, followed by predictive coding is proposed in \cite{Barowsky2021-wg}.
A neural network is overfitted on the dataset to predict pixel values autoregressively based on the image to which the pixel belongs, as well as previously seen images.
This setting is equivalent to sub-optimally compressing a multiset of correlated symbols, while our method optimally compresses multisets of i.i.d.\ symbols.

\section{Problem Setting}
Given a sequence of i.i.d.\ discrete random variables $X^n = (X_1, \dots, X_n)$ with symbol alphabet $\X$, the goal is to design a lossless source code for the \emph{multiset} $\M = \{X_1, \dots, X_n\} = \multiset{X^n}$ with code-length, in the large state regime (see \Cref{theorem:optimality-ans}),
\begin{equation}
    \log \frac{1}{P_\M(\M)} = \log\frac{1}{P_{X^n}(x^n)} - \log M, \label{eq:multiset-info-content}
\end{equation}
where the constant $M$ is known as the multinomial coefficient (\Cref{def:multinomial-coeff}) of $\M$.
This coefficient is equal to the number of unique permutations of $X^n$
\begin{align}
    M = \abs{\left\{x^n \in \X^n: \multiset{x^n} = \M \right\}}
    = \frac{n!}{\prod_{x\in\X }\M(x)!} \leq n!,
\end{align}
with equality when all elements in $\M$ are distinct.
The term $\log M$ can be interpreted as the number of bits required to order the elements in $\M$ to create a sequence $X^n$.
For this reason this term is sometimes referred to as the \emph{order information} \cite{varshney2006}.

This problem can equivalently be seen as that of encoding the sequence $X^n$ with complete disregard to the order between symbols $X_i$.

We are concerned with both the rate and computational complexity of encoding and decoding. Although the method we present is applicable to any alphabet $\X$ and multiset size $\abs{\M}=n$, we are mainly interested in sources with large alphabets $\abs{\X} \gg n$. For example, if $X_i$ are images then $\abs{\X}$ will grow exponentially with the number of pixels.

To characterize the information content of $\M$, we can restrict ourselves to exchangeable distributions $P_{X^n}$ without loss of generality.
\begin{theorem}[Sufficiency of Exchangeability]
    For any distribution over sequences $P_{X^n}$, there exists an exchangeable distribution $\bar{P}_{X^n}$ that results in the same distribution over multisets $P_\M$.
\end{theorem}
\begin{proof}
    Let $[x^n] = \left\{ z^n \in \X^n \colon \multiset{z^n} = \multiset{x^n} \right\}$ represent the \emph{type class} to which $x^n$ belongs, i.e., the set of sequences with the same frequency count of symbols as $x^n$.
    Any sequence $z^n \in [x^n]$ creates the same multiset $\M = \multiset{z^n} = \multiset{x^n}$.
    The cardinality of $[x^n]$ is equal to the multinomial coefficient of $\M$. The distribution over multisets induced by $P_{X^n}$ is
    \begin{align}
        P_{\M}(\M) = \sum_{z^n \in [x^n]} P_{X^n}(z^n).
    \end{align}
    The sum is over all elements in $[x^n]$.
    We can therefore exchange probability mass between sequences in $[x^n]$ without changing $P_\M(\M)$.
    In particular, the distribution $\bar{P}_{X^n}$ that assigns equal mass to all sequences in the same set $[x^n]$ is, by definition, exchangeable:
    \begin{align}
        \bar{P}_{X^n}(x^n) = \frac{1}{\abs{[x^n]}} P_\M(\M).
    \end{align}
\end{proof}

Although our method works for any source of exchangeable symbols, the following sections assumes symbols are i.i.d. to simplify the exposition.

\section{Method}
Random Order Coding (ROC) is a lossless compression method for multisets.
Compressing a multiset is equivalent to compressing a sequence $x^n$ with disregard to the order.
In practice, this means the decoded sequence $z^n$ will be a random permutation of $x^n$, with all orderings being equally likely.

Compressing $x^n$ with ANS proceeds by initializing the state $s_0=0$, followed by successively applying $s_i \defeq \encode(s_{n-1}, x_i)$ with $P_X$.
The final ANS state $s_n$ holds the information of the entire sequence, which will be
approximately $\log s_n \approx \sum_{n=1}^n \log 1 / P_X(x_i)$ bits. To optimally compress $\M$, we must somehow recover the bits used to implicitly encode the order. This can be done by using the ANS state as a source of randomness to sample without replacement from $\M$, encoding symbols in whichever order they appear. This idea leads to the process shown in \Cref{alg:roc-encode} which we discuss below.

To encode, we keep track of a multiset of \emph{remaining} symbols $\overline{\M}_i$,
i.e.\ those that have not yet been sampled from $\M$ up to step $i$.
First, $\overline{\M}_1$ is initialized to $\M$.
Then, for each step $i \in \{1, \ldots, n\}$, a symbol $z_i$ is sampled without replacement from $\overline{\M}_i$, and then encoded with $P_Z = P_X$.
We can express this relationship as
\begin{align}
    \M = \overline{\M}_i \bigcup \multiset{z^{i-1}},
\end{align}
where $\multiset{z^{i-1}} = \{z_1, \dots, z_{i-1}\}$ is the multiset of elements that have been sampled up to step $i-1$.
The sampling-without-replacement distribution, which is used to sample $z_i$ via decoding, is
\begin{equation}\label{eq:roc-swor-prob}
    P_{Z_i \g Z^{i-1}, \M}(z_i \g z^{i-1}, \M) = P_{Z_i \g \overline{\M}_i}(z_i \g \overline{\M}_i) = \frac{\overline{\M}_i(z_i)}{n - (i-1)}.
\end{equation}
\begin{algorithm}
    \small
    \caption{Multiset encode}\label{alg:roc-encode}
    Given a multiset $\M = \{x_1, \dots, x_n\}$ with $m$ unique symbols\\
    Assume encoding with $P_X$ has worst-case complexity per-symbol of $\Omega(P^\text{enc}_X)$\\
    Initialize an ANS state $s_0$, and $\overline\M_1 = \M$\\
    \For{$i = 1, \dots, n$}{
        \nl $(s_i', z_i) = \decode(s_{i-1})$ with $P_{Z_i \g \overline\M_i}$
        \hfill \tcp{$\Omega(\log m)$}
        \nl $\overline\M_{i+1} = \overline\M_i\setminus\{z_i\}$
        \hfill \tcp{$\Omega(\log m)$}
        \nl $s_i = \encode(s_i', z_i)$ with $P_X(z_i)$
        \hfill \tcp{$\Omega(P^\text{enc}_X)$}
    }
    \Return $s_n$
\end{algorithm}
\begin{algorithm}
    \small
    \caption{Multiset decode}\label{alg:roc-decode}
    Given an integer ANS state $s_n$\\
    Assume decoding with $P_X$ has worst-case complexity per-symbol of $\Omega(P^\text{dec}_X)$\\
    Initialize $\overline\M_{n+1} = \emptyset$\\
    \For{$i = n, \dots, 1 $}{
        $(s_i', z_i) = \decode(s_i)$ with $P_X$ \hfill
        \tcp{$\Omega(P^\text{dec}_X)$, reverses $\mathbf{3}$}
        $\overline\M_i = \overline\M_{i+1}\cup\{z_i\}$
        \hfill \tcp{$\Omega(\log m)$, reverses $\mathbf{2}$}
        $s_{i-1} = \encode(s_i', z_i)$ with $P_{Z_i \g \overline\M_i}(z_i\g \overline\M_i)$\hfill
        \tcp{$\Omega(\log m)$, reverses $\mathbf{1}$}
    }
    \Return $\overline\M_1 = \M$
\end{algorithm}

Sampling $z_i$ decreases the integer state, while encoding $z_i$ increases it.
Hence, the number of bits required to represent the integer state changes at step $i$ by approximately
\begin{equation}\label{eq:message-length}
    \Delta_i \defeq \log \frac{1}{P_X(z_i)} - \log\frac{1}{P_{Z_i \g \overline\M_i}(z_i \g \overline\M_i)}.
\end{equation}
The total change, in bits, of the state is
\begin{align}
    \sum_{i=1}^{n} \Delta_i
     & = \log \frac{1}{P_{X^n}(z^n)} - \log\frac{1}{P_{Z^n \g \M}(z^n \g \M)}. \label{eq:rate}
\end{align}
Sampling without replacement guarantees all orderings of $z^n$ are equally likely,
implying \eqref{eq:rate} is exactly the information content of the multiset \eqref{eq:multiset-info-content}.
\begin{lemma}[All Orderings are Equally Likely]
    Let $Z^n$ be the sequence of random variables produced from sampling without replacement from a multiset $\M$.
    Then, all sequences are equally likely,
    \begin{align}
        \log\frac{1}{P_{Z^n \g \M}(z^n \g \M)} = \log\frac{n!}{\prod_{z \in \M} \M(z)!}.
    \end{align}
\end{lemma}
\begin{proof}
    The probability of a sampled sequence is given by \Cref{eq:roc-swor-prob},
    \begin{align}
        P_{Z^n \g \M}(z^n \g \M)
         & = \prod_{i=1}^n P_{Z_i \g Z^{i-1}, \M}(z_i \g z^{i-1}, \M)                                  \\
         & = \prod_{i=1}^n \frac{\overline{\M}_i(z_i)}{n - (i-1)}                                      \\
         & = \frac{1}{n!}\cdot\overline{\M}_1(z_1)\cdot\overline{\M}_2(z_2)\cdots \overline{\M}_n(z_n)
    \end{align}
    We can regroup the product terms based on the value of element $z_i \in \M$, by defining
    \begin{align}
        I_z = \{i \in [n] \colon z_i = z \},
    \end{align}
    and writing,
    \begin{align}
        P_{Z^n \g \M}(z^n \g \M)
         & = \frac{1}{n!}\cdot\prod_{z \in \M}\prod_{i \in I_z}\overline{\M}_i(z_i). \\
    \end{align}
    At step $i \in I_z$, the element sampled without replacement is equal to $z$.
    Sampling decrements the multiplicty of $z$ by exactly one.
    Therefore,
    \begin{align}
        \overline{\M}_{i_1}(z), \overline{\M}_{i_2}(z), \dots, \overline{\M}_{i_{\abs{I_z}}}(z),
    \end{align}
    where $i_j \in I_z$ and $i_1 < i_2 < \dots < i_{\abs{I_z}}$, decreases at each step by $1$, allowing us to write
    \begin{align}
        P_{Z^n \g \M}(z^n \g \M)
         & = \frac{1}{n!}\cdot\prod_{z \in \M} \M(z)!
    \end{align}
\end{proof}

\begin{example}[ROC $\{\mathtt{a}, \mathtt{b}, \mathtt{b}\}$]\label{example:roc}
    \Cref{fig:roc-example} shows how ROC is used to compress a multiset $\M \defeq \{\mathtt{a}, \mathtt{b}, \mathtt{b}\}$.
    Each column shows the evolution of the ANS state, ordered from left-to-right, as symbols are encoded.
    The size of the stack $\ell_\M$ is shown in gray.
    Initially, the ANS stack is populated with some initial bits ($\epsilon$) to initiate the sampling procedure.
    The stack is used as a random seed to select an element in the multiset to compress.
    In this example, $\mathtt{b}$ was chosen first with probability $\frac{2}{3}$, resulting in a reduction of the size of the stack.
    The selected element ($\mathtt{b}$) is removed from the multiset and encoded onto the stack with $P_X(\mathtt{b})$, increasing the stack size.
    As before, the stack is used as a random seed to select an element.
    This time, $\mathtt{a}$, with probability $\frac{1}{2}$.
    The bits used to select $\mathtt{a}$ are the initial bits and the encoding of $\mathtt{b}$.
    Finally, the last element is encoded onto the stack.
    The total length change in the ANS state, in the large state regime, is  $-\log P_\M(\{a,b,b\})$.
\end{example}

\begin{figure}[ht]
    \begin{center}
        \includegraphics[width=0.8\textwidth,trim={0 0 1cm 0},clip]{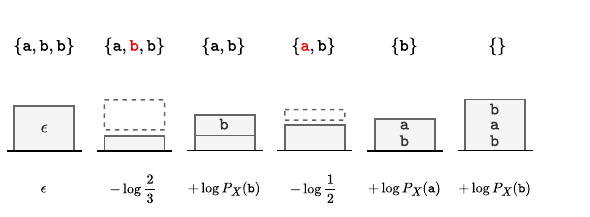}
        \caption{\Cref{example:roc}. Encoding $\M = \{\mathtt{a}, \mathtt{b}, \mathtt{b} \}$ with Random Order Coding (ROC).}
        \label{fig:roc-example}
    \end{center}
\end{figure}

There are two potential sources of overhead, implying that in practice the compressed message length will never be exactly equal to the information content of the multiset.
The first is due to the per-operation and one-time ANS overheads, discussed in \Cref{sec:ans-encoding}.
Second, in order to sample at the very beginning (i.e.\ from $s_0$) there needs to be some pre-existing randomness in the ANS state (i.e., $s_0 \gg 0$).
We can synthesize that randomness (i.e.\ initialize $s_0$ to a larger integer), but this adds a slight overhead because the random bits that we use are (unavoidably) communicated.
This is the initial bits overhead discussed in \Cref{sec:bbans-initial-bits}.
We found in our experiments that the effect of these overheads on rate was small.
There are at least two theoretical results that corroborate these findings, which we discuss in \Cref{sec:roc-initialbits}

The decoder (\Cref{alg:roc-decode}) receives the ANS state $s_n$ containing a compressed representation of all elements of $\M$, and must rebuild the multiset as symbols are retrieved.
Note that, due to the stack-like nature of ANS, this is done in reverse order.

Efficient ANS-based invertible sampling without replacement requires a binary search tree (BST), which we describe in detail in \Cref{sec:roc-swor}, to store the multisets $\overline\M_i$.
The BST is similar to a Fenwick tree \cite{fenwick1994,moffat1999}, and allows lookup, insertion and removal of elements in time proportional to the depth of the tree.

\section{Initial Bits}\label{sec:roc-initialbits}
The following theorem suggests the initial bits are not a significant problem when performing sampling-without-replacement.
The ANS state is likely to be empty only at the very beginning of encoding, i.e.,\ the initial bits are a one-time overhead.

\begin{theorem}[Initial Bits are a One-Time Overhead]
    Let $\Delta_i$ be the net change in message length, in the large state regime, when encoding the $i$-th symbol with ROC.
    For a fixed multiset $\M$, the expected change is always non-negative,
    \begin{align}
        \E{\Delta_i \g \M} \geq 0.
    \end{align}
    Furthermore, the expected change in message length can not decrease,
    \begin{align}
        \E{\Delta_i} \geq  \E{\Delta_{i-1}}
    \end{align}
\end{theorem}
\begin{proof}
    At step $i$, a symbol is sampled via decoding and then encoded using $P_X$, resulting in the following conditional given a fixed multiset $\M$
    \begin{align}
        \E{\Delta_i \g \M}
         & = \E{\log P_{Z_i\g\overline\M_i}(Z_i\g\overline\M_i) - \log P_{Z}(Z_i) \g \M } \\
         & = \KL{P_{Z_i\g\overline\M_i}(\cdot \g \M)}{P_Z}                                \\
         & \geq 0
    \end{align}
    Second, $X^n$ is exchangeable and therefore from
    \begin{align}
        H(Z_i\g\overline\M_i) & = H(Z_i\g\M, Z_{i-1}, \dots, Z_2, Z_1) \\
                              & \leq H(Z_i\g\M, Z_{i-1}, \dots, Z_2)   \\
                              & = H(Z_{i-1}\g\M, Z_{i-2}, \dots, Z_1)  \\
                              & = H(Z_{i-1}\g \M_{i-1}),
    \end{align}
    we have that
    \begin{align}
        \E{\Delta_i - \Delta_{i-1} }= -H(Z_i\g\overline\M_i) + H(Z_{i-1}\g \overline\M_{i-1}) \geq 0.
    \end{align}
\end{proof}

\section{Efficient Sampling Without Replacement}\label{sec:roc-swor}
For our method to be computationally attractive, we must carefully construct an algorithm that can decode $Z_i$ from $\overline\M_i$, as well as adapt to the next symbol $Z_{i+1}$.
Decoding the proxy sequence requires sampling without replacement from the multiset.
In this section we show that both sampling and adapting to the next symbol can be done in sub-linear time for each symbol, or quasi-linear for the entire multiset.
This requires a binary search tree (BST), which we describe in detail next, to store the multisets $\overline\M_i$.
The BST is similar to a Fenwick tree \cite{fenwick1994,moffat1999}, and allows lookup, insertion and removal of elements in time proportional to the depth of the tree. We analyze the overall time complexity of our method in \Cref{sec:roc-complexity}.

Each encoding/decoding step requires a forward/reverse lookup (see \ref{sec:ans-encoding}) with distribution
$P_{Z_i\g\overline\M_i}$. Since we are sampling without replacement, we also
need to perform removals and insertions:
\begin{align}
    \bstremove & \colon (\overline\M_i, z_i) \mapsto \overline\M_i \setminus \{z_i\}     \\
    \bstinsert & \colon (\overline\M_{i+1}, z_i) \mapsto \overline\M_{i+1} \cup \{z_i\}.
\end{align}
All four operations (forward and reverse lookup, insertion and removal) can be
done in $\Omega(\log m)$ time, where \(m\) is the number of unique elements in
\(\M\) as well as the number of nodes in the BST.

We set the ANS precision parameter (see \Cref{sec:ans-encoding}) to \(N = |\overline\M_i| = n - (i -1) \), with
frequency count $p_x = \overline\M_i(x)$ for each $x \in \overline\M_i$.
At each step $i$, a binary search tree (BST) is used to represent the multiset
of remaining symbols $\overline\M_i$.  The BST stores the frequency and cumulative
counts, $p_x$ and $c_x$, and allows efficient lookups.  At each node is a
unique symbol from \(\overline\M_i\) (arranged in the usual BST manner to enable fast
lookup), as well as a count of the total number of symbols in the entire sub-tree
of which the node is root. \Cref{fig:bst-1} shows an example, for the multiset
$\overline\M_i = \{\mathtt{a, b, b, c, c, c, d, e}\}$.

\begin{figure}[ht]
    \centering
    \begin{tikzpicture}[
            edge from parent/.style = {draw,-latex}
        ]
        \node {$\mathtt{b}\colon 8$}
        child {node {$\mathtt{a}\colon 1$}}
        child {node {$\mathtt{d}\colon 5$}
                child {node {$\mathtt{c}\colon 3$}}
                child {node {$\mathtt{e}\colon 1$}}};
    \end{tikzpicture}
    \begin{tikzpicture}
        \draw (0, 0) -- (8, 0);
        \foreach \x in {0,...,8} {
                \draw (\x,.1) -- (\x,-.1) ;
            }
        \foreach \x in {0,...,7} {
                \draw (\x + 0.5, .3) node[fill=white] {\footnotesize\(\x\)};
            }

        \draw (3,   -3)      -- (6, -3);
        \draw (3,   -3 - .1) -- (3, -3 + .1);
        \draw (6,   -3 - .1) -- (6, -3 + .1);
        \draw (4.5, -3) node[fill=white] (c) {\small\(\mathtt{c}\)};

        \draw (7,   -3)      -- (8, -3);
        \draw (7,   -3 - .1) -- (7, -3 + .1);
        \draw (8,   -3 - .1) -- (8, -3 + .1);
        \draw (7.5, -3) node[fill=white] (e) {\small\(\mathtt{e}\)};

        \draw (0,  -2)      -- (1, -2);
        \draw (0,  -2 - .1) -- (0, -2 + .1);
        \draw (1,  -2 - .1) -- (1, -2 + .1);
        \draw (.5, -2) node[fill=white] (a) {\small\(\mathtt{a}\)};

        \draw (3,   -2)      -- (8, -2);
        \draw (3,   -2 - .1) -- (3, -2 + .1);
        \draw (6,   -2 - .1) -- (6, -2 + .1);
        \draw (7,   -2 - .1) -- (7, -2 + .1);
        \draw (8,   -2 - .1) -- (8, -2 + .1);
        \draw (6.5, -2) node[fill=white] (d) {\small\(\mathtt{d}\)};

        \draw (0, -1)      -- (8, -1);
        \draw (0, -1 - .1) -- (0, -1 + .1);
        \draw (1, -1 - .1) -- (1, -1 + .1);
        \draw (3, -1 - .1) -- (3, -1 + .1);
        \draw (8, -1 - .1) -- (8, -1 + .1);
        \draw (2, -1) node[fill=white] (b) {\small\(\mathtt{b}\)};
    \end{tikzpicture}
    \caption{For the multiset $\{\mathtt{a, b, b, c, c, c, d, e}\}$, on the top a schematic representation of the dynamic BST data structure which we use to represent the multiset, and on the bottom the intervals corresponding to each branch of the BST.}\label{fig:bst-1}
\end{figure}
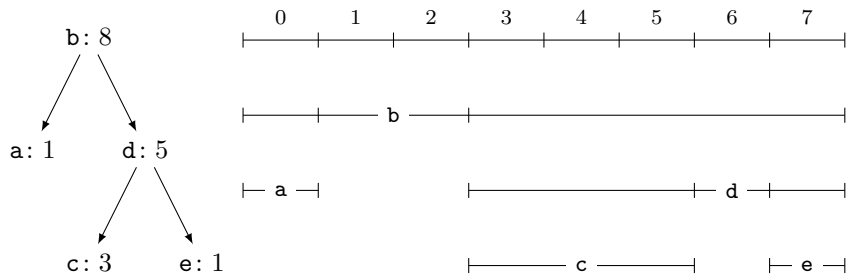
On the left is the actual BST data-structure, i.e.,\ that which is held in memory
and operated on. The interval diagram on the right visually represents the
local information available from inspection of the counts at a node and its
children. For example, the sub-interval to the right of $\mathtt{b}$ has length
$5$, which is the count at it's child node $\mathtt{d}$.

\begin{example}[BST $\rlookup$]
    For illustrative purposes, we step through a $\mathtt{reverse\_lookup}$ of the multiset in \Cref{fig:bst-1}.
    This starts with  an integer $j = s \bmod (|\M| - i + 1)$ taken from the ANS state $s$.
    Searching starts at the root node, and compares $j$ to the frequency and cumulative counts, $p_x$ and $c_x$, to decide between branching left ($j < c_x$), right ($j \geq c_x + p_x$), or returning.
    For example, at the root node, if $s=68$ and $i=1$, then $x=\mathtt{b}$, $c_{\mathtt{b}} = 1, p_{\mathtt{b}} = 8 - 1 - 5 = 2$; \(j = 4 \geq 1 + 2 = c_\mathtt{b} + p_\mathtt{b}\), hence we traverse to the right branch.
    Branching right requires us to re-center to focus on the sub-interval to the right of $\mathtt{b}$.
    This can be done by continuing the traversal with $j - (c_\mathtt{b} + p_\mathtt{b})$.
    Note that, if we had branched left, this re-centering would not be necessary.
    After reaching node $\mathtt{c}$, we return back up the tree, propagating the value of \(p_\mathtt{c}\) and summing the size of all intervals to the left of \(\mathtt{c}\) to compute \(c_\mathtt{c}\).
    Note that, if we decrement (increment) the count of each node we visit by 1, then $\rlookup$ and $\bstremove$ ($\bstinsert$) can both be performed in one pass over the tree. The same also holds for $\flookup$.
\end{example}

Note that all sub-trees correspond to some sub-multiset $S \subseteq \overline\M_i$.
Under this interpretation, the count at each node represents the size $\abs{S}$ of the sub-multiset corresponding to the sub-tree with that particular node as root. The count at the root of the BST ($8$ in \Cref{fig:bst-1}) is equal to the total size $\abs{\overline\M_i}$ of the multiset.

Below we provide pseudo-code for $\bstinsert, \bstremove, \flookup$ and $\rlookup$.
Here the symbol $\B$ is used for the BST that represents $\overline\M_i$.
The BST itself is a 4-tuple $\B = (n, y, \B_L, \B_R)$, where $n=\abs{\overline\M_i}$ is the size of $\overline\M_i$, $y$ is the root symbol of $\B$, and $\B_L,\B_R$ are the BSTs with root symbols equal to the left and right children of the root node $y$ of $\B$.
We also define $\abs{\B} = \abs{\overline\M_i} = n$.

\

\begin{algorithm}[H]
    \caption{$\bstinsert$ and $\bstremove$ of symbol $x$ with BST $\B$.}
    \small
    \begin{minipage}[t]{0.45\textwidth}
        \DontPrintSemicolon
        \SetKwFunction{InsertFunc}{$\bstinsert$}
        \SetKwProg{Fn}{}{:}{}
        \Fn{\InsertFunc{$\B$, $x$}}{
            $(n, y, \B_L, \B_R) = \B$\;
            \uIf{$x < y$}{
                $\B_L = \bstinsert(\B_L, x)$\;
            }
            \uElseIf{$x > y$}{
                $\B_R = \bstinsert(\B_R, x)$\;
            }
            \KwRet $(n+1, y, \B_L, \B_R)$\;
        }
    \end{minipage}
    \label{alg:insert-remove}
    \hfill
    \begin{minipage}[t]{0.45\textwidth}
        \DontPrintSemicolon
        \SetKwFunction{RemoveFunc}{$\bstremove$}
        \SetKwProg{Fn}{}{:}{}
        \Fn{\RemoveFunc{$\B$, $x$}}{
            $(n, y, \B_L, \B_R) = \B$\;
            \uIf{$x < y$}{
                $\B_L = \bstremove(\B_L, x)$\;
            }
            \uElseIf{$x > y$}{
                $\B_R = \bstremove(\B_R, x)$\;
            }
            \KwRet $(n-1, y, \B_L, \B_R)$\;
        }
    \end{minipage}
\end{algorithm}

\begin{algorithm}[H]
    \caption{$\flookup$ of symbol $x$, and $\rlookup$ of index $i$, in $\B$.}
    \small
    \begin{minipage}[t]{0.45\textwidth}
        \DontPrintSemicolon
        \SetKwFunction{FLookupFunc}{$\flookup$}
        \SetKwProg{Fn}{}{:}{}
        \Fn{\FLookupFunc{$\B$, $x$}}{
            $(n, y, \B_L, \B_R) = \B$\;
            $(n_R, n_L) = (\abs{\B_R}, \abs{\B_L})$\;
            $(c_y, p_y) = (n_L, n - (n_L + n_R))$\;
            \uIf{$x = y$}{
                $(c_x, p_x) = (c_y, p_y)$\;
            }
            \uElseIf{$x < y$}{
                $(c_x, p_x) = \flookup(\B_L, x)$\;
            }
            \uElseIf{$x > y$}{
                $(c, p_x) = \flookup(\B_R, x)$\;
                $c_x = c - (n - n_R)$\;
            }
            \KwRet $(c_x, p_x)$
        }
    \end{minipage}
    \hfill
    \begin{minipage}[t]{0.45\textwidth}
        \DontPrintSemicolon
        \SetKwFunction{RLookupFunc}{$\rlookup$}
        \SetKwProg{Fn}{}{:}{}
        \Fn{\RLookupFunc{$\B$, $i$}}{
            $(n, y, \B_L, \B_R) = \B$\;
            $(n_R, n_L) = (\abs{\B_R}, \abs{\B_L})$\;
            $(c_y, p_y) = (n_L, n - (n_L + n_R))$\;
            \uIf{$c_y \leq i < c_y + p_y$}{
                $(x, c_x, p_x) = (y, c_y, p_y) $\;
            }
            \uElseIf{$i < c_y $}{
                $(c_x, p_x) = \rlookup(\B_L, i)$
            }
            \uElseIf{$i \geq c_y + p_y$}{
                $i' = i - (n - n_R)$\;
                $(c, p_x) = \rlookup(\B_R, i')$\;
                $c_x = c + (n - n_R)$\;
            }
            \KwRet $(x, c_x, p_x)$
        }
    \end{minipage}
\end{algorithm}

As shown in \Cref{alg:insert-remove}, all necessary operations are straightforward to implement using depth-first traversal of the BST with time complexity which scales linearly with the depth of the tree.
As long as the tree is balanced enough (which is guaranteed during encoding and highly likely during decoding), we have \(\Omega(\log m)\) complexity for each operation.

\section{Time Complexity}\label{sec:roc-complexity}
In practice, the dominant factors affecting the runtime of encoding and decoding with our method are the total ($n$) and unique ($m$) number of symbols in the multiset, as well as the complexity of encoding and decoding with $P_X$.

The BST is used to store the elements of the multisets $\overline\M_i$.
The BST allows lookup, insertion and removal of elements in time proportional to the depth of the tree.
Traversing the BST also requires comparing symbols under a fixed (usually lexicographic) ordering.
In the extreme case where symbols in the alphabet $\X$ are represented by binary strings of length $\log\abs{\X}$, a single comparison between symbols would require $\log\abs{\X}$ bit-wise operations.
However, ``short-circuit" evaluations, where the next bits are compared only if all previous comparisons result in equality, make this exponentially unlikely.  In practice, average time may have no dependence on $\abs{\X}$.

The expected and worst-case time-complexity of the multiset encoder and decoder which we propose are
\begin{align}
    \Omega(nP_X + n\log m),
\end{align}
where $P_X$ represents the complexity of encoding (at the encoder) and decoding (at the decoder) with $P_X$.
The complexities are independent of the alphabet size, while current methods require at least $\Omega(\abs{\X})$ iterations.

We now detail the time dependence of sampling without replacement, and its inverse, on $n$ and $m$, assuming a fixed cost of comparing two symbols.
At the encoder, a balanced BST representing $\overline\M_1 = \M$ must first be constructed from the sequence $x^n$.
This requires sorting followed by performing the $\bstinsert$ operation $n$ times, starting with an empty BST.
Sampling $z^n$ from the multiset is done via $\rlookup$ and $\bstremove$.
All operations have worst-case and average complexity equivalent to that of a search on a balanced BST, implying the overall complexity is $\Omega(n\log m)$, including the initial sorting step.
At the decoder, the BST is constructed by inserting $z_n, \dots, z_1$, which are ordered randomly, implying the BST will not always be balanced.
However, the expected depth of any node in a randomly initialized BST, which is proportional to the expected time of the \(\flookup\) and \(\bstinsert\) operations, is \(\Omega(\log m)\) \cite{knuth1998}.
The worst-case is a tree with one long, thin branch, where those operations would take \(\Omega(m)\) time.
A self-balancing tree, such as an AVL tree \cite{adelson-velsky1962} or red-black tree \cite{bayer1972}, can be used to achieve the same worst-case complexity as encoding.

\changed{
    \section{ROC and the Method of Types}
    A \emph{type} $\hat{P}[x^n]$ of a sequence of length $n$, in the method of types \cite{cover1999elements}, is the empirical distribution with probabilities $\hat{P}[x^n](x) \defeq \frac{1}{n} \sum_{i=1}^n \1\{x_i = x\}$, constructed from counting the occurrence of symbols in the sequence.
    The \emph{type class} $T(\hat{P}[x^n])$ of a type $\hat{P}[x^n]$ is the set of all sequences of that type.
    The number of sequences in a type class is equal to the multinomial coefficient associated to that type,
    \begin{align}
        \abs{T(\hat{P}[x^n])} = \frac{n!}{\prod_{x \in \X}(n\cdot\hat{P}[x^n](x))!}.
    \end{align}

    Multisets have a one-to-one mapping with types.
    If $\M = \multiset{x^n}$, then $\M(x) = n \cdot \hat{P}[x^n](x)$.
    The type class is exactly the set of all permutations of $x^n$, the equivalence class of sequences.
    As far as assigning probabilities to these objects, we can reason about them as equivalent to one another.

    It is known that, as $n \rightarrow \infty$, the probability of the multiset converges to a delta function \cite{cover1999elements} over a single multiset.
    This is intuitive from a statistical standpoint: if the sequence length is large, and the symbols are i.i.d., then the empirical distribution, must approach the true data distribution.
    To see how, note that the negative log-probability of a sequence under an arbitrary i.i.d.\ distribution, $Q_{X^n} = \prod_i Q_X$, can be written as a function of its type,
    \begin{align}\label{eq:method-of-types-prob}
        -\frac{1}{n}\cdot\log Q_{X^n}(x^n)
         & = -\frac{1}{n}\cdot\log \prod_{i=1}^n Q_X(x_i)                            \\
         & = -\frac{1}{n}\cdot\log \prod_{x \in \X} Q_X(x)^{n \cdot \hat{P}[x^n](x)} \\
         & = - \sum_{x \in \X} \hat{P}[x^n](x) \cdot \log Q_X(x)                     \\
         & =  H(\hat{P}[x^n]) + \KL{\hat{P}[x^n]}{Q_X}.
    \end{align}
    Second, the size of the equivalence class can be upper bounded by considering the probability of the equivalence class under the empirical distribution,
    \begin{align}
        \sum_{x^n \in [x^n]} \hat{P}[x^n](x^n)
         & = \abs{T(\hat{P}[x^n])} \cdot \hat{P}[x^n](z^n)             \\
         & = \abs{T(\hat{P}[x^n])} \cdot 2^{-n \cdot H(\hat{P}[x^n])},
    \end{align}
    for any $z^n \in [x^n]$,
    where the last step follows from \Cref{eq:method-of-types-prob} with $Q_X =\hat{P}[x^n]$.
    Since $\hat{P}[x^n](x^n)$ is a probability mass function, then the sum above must be upper-bounded by $1$, implying
    \begin{align}
        \abs{T(\hat{P}[x^n])} \leq 2^{n \cdot H(\hat{P}[x^n])}.
    \end{align}
    This allows us to write, for any $z^n \in [x^n]$,
    \begin{align}
        P_\M(\M)
         & = \sum_{x^n \in [x^n]} P_{X^n}(x^n)        \\
         & = \abs{T(\hat{P}[x^n])} \cdot P_{X^n}(z^n) \\
         & \leq 2^{-n \cdot \KL{\hat{P}[x^n]}{P_X}},
    \end{align}
    which is non-zero for $n \rightarrow \infty$ only when $\hat{P}[x^n] = P_X$.
    This suggests that, as $n$ grows, the information present in the type is less relevant to the overall cost of sending the sequence.
    The information present in the permutation, i.e., the ordering between elements, dominates.

    ROC can be seen as an algorithm to communicate the type, without the ordering information.
    This is somewhat at odds with the result discussed above, as the amount of bits present in the type goes to zero when $n$ grows to infinity.
    However, for practical applications, where $n$ is always finite, the amount of information in the type is still relevant.
    Of practical relevance is the example of when $P_X$ is uniform over $\X$ (e.g., sets of identifiers in database applications, as in \Cref{sec:rcc-experiments}).
    In this case, the probability of the multiset is a constant times the multinomial coefficient, implying the information content of two multisets will be equal if the set of their counts is equal.
}

\section{Experiments}\label{sec:roc-experiments}
In this section we present experiments on synthetically generated multisets with known source distribution, multisets of grayscale images with lossy codecs, and collections of JSON maps represented as a multiset of multisets.
We used the ANS implementation in the Craystack library \cite{townsend2020a} for all experiments.

\subsection{Synthetic multisets}
Here, synthetically generated multisets are compressed to provide evidence of the computational complexity, and optimal compression rate of the method.
We grow the alphabet size $|\X|$ while sampling from the source in a way that guarantees a fixed number of unique symbols $m = 512$. The alphabet $\X$ is always a subset of $\mathbb{N}$.

For each run, we generate a multiset with $m=512$ unique symbols, and use a skewed distribution, sampled from a Dirichlet prior with coefficients $\alpha_k = k$ for $k=1,\ldots,|\X|$, as the distribution $P_X$.

The final compressed size of the multiset and the information content (i.e.  Shannon lower bound) assuming the distribution \(P_X\), are shown in \Cref{fig:roc-toy-multisets} for different settings of \(\abs{\X}\), alongside the total encode plus decode time.
Results are averaged over 20 runs, with shaded regions representing the 99\% to 1\% confidence intervals.
In general, the new codec compresses a multiset close to it's information content for varying alphabet sizes, as can be seen in the left plot.

The total encode plus decode time is unaffected by the alphabet size $|\X|$.
As discussed previously, the overall complexity depends on that of coding under \(P_X\).
Here, the \(P_X\) codec does include a logarithmic time binary search over $|\M|$, but this is implemented efficiently and the alphabet size can be seen to have little effect on overall time.
The total time scales linearly with the multiset size $|\M|$ (right plot), as expected.

\begin{figure}[ht]
    \begin{center}
        \begin{tabular}{cc}
            \includegraphics[width=0.45\textwidth]{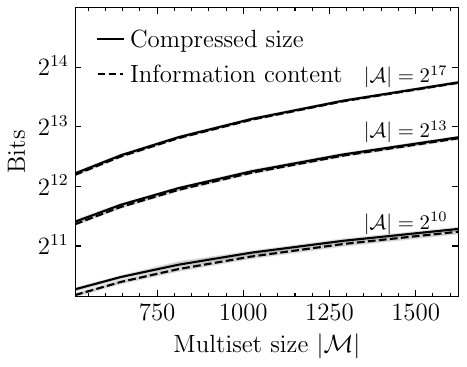} &
            \includegraphics[width=0.45\textwidth]{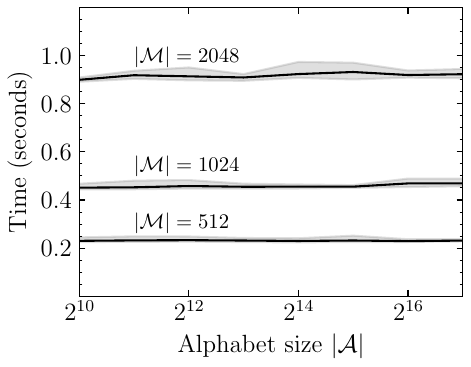}
        \end{tabular}
        \caption{Left: Final compressed length is close to the information content for varying alphabet and multiset sizes.
            Right: Computational complexity does not scale with alphabet size $|\A|=|\X|$, and is linear in $|\M|.$}
        \label{fig:roc-toy-multisets}
    \end{center}
\end{figure}

\subsubsection{MNIST with lossy WebP}
We implemented compression of multisets of grayscale images using the lossy WebP codec.
We tested on the MNIST test set, which is composed of $10,000$ distinct grayscale images of handwritten digits, each $28\times 28$ in size.
To encode the multiset, we perform the sampling procedure to select an image to compress, as usual.
The output of WebP is a prefix-free, variable-length sequence of bytes, which we encoded into the ANS state via a sequence of $\encode$ steps with a uniform distribution.

We compared the final compressed length with and without the sampling step. In other words, treating the dataset as a multiset and treating it as an ordered sequence. The savings achieved by using our method are shown in \Cref{fig:roc-mnist-jsonmaps}.
The theoretical limit shown in the left plot is $\log |\M|!$, while in the right plot this quantity is divided by the number of bits needed to compress the data sequentially.

\begin{figure}[ht]
    \begin{center}
        \begin{tabular}{cc}
            \includegraphics[width=0.485\textwidth]{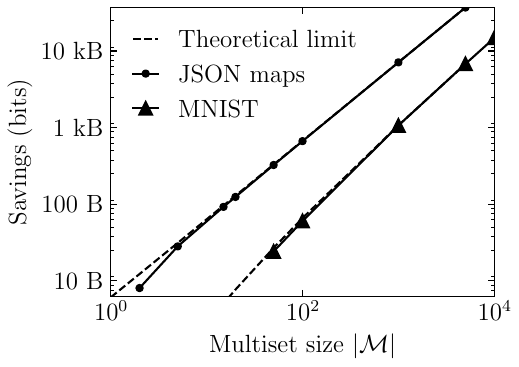} &
            \includegraphics[width=0.46\textwidth]{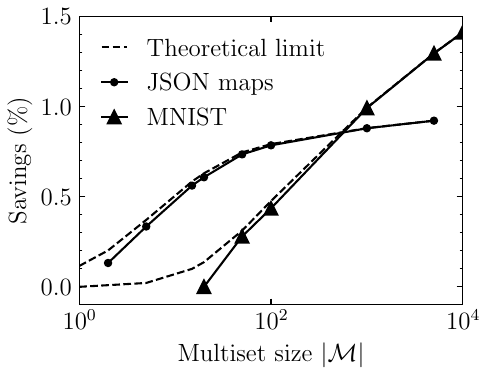}
        \end{tabular}
    \end{center}
    \caption{
        Rate savings due to using our method to compress a multiset instead of treating it as an ordered sequence.
        Savings are close to the theoretical limit in both cases.
        The symbols are bytes outputted by lossy WebP for MNIST, and UTF-8 encoding for JSON maps.
        A uniform distribution over bytes is used to encode with ANS.
        Left: Savings in raw bits.
        Right: Percentage savings.}
    \label{fig:roc-mnist-jsonmaps}
\end{figure}

Note that the maximum savings per symbol $\frac{1}{|\M|}\log_2 |\M|! \approx \log_2 |\M|$ depends only on the size of multiset.
Therefore, when the representation of the symbol requires a large number of bits, the percentage
savings are marginal (roughly 1.5\% for 10,000 images, in our case).
To improve percentage savings (right plot), one could use a better symbol codec or an adaptive codec which doesn't treat the symbols as independent.
However, as mentioned, the savings in raw bits (left plot) would remain the same, as it depends only on the multiset size $|\M|$.

\subsection{Collection of JSON maps as nested multisets}
The method can be nested to compress a multiset of multisets, by performing an additional sampling step that first chooses which inner multiset to compress.
In this section we show results for a collection of JSON maps $\M \defeq \{\J_1, \dots, \J_{|\M|}\}$, where each map $\J_i \defeq \{(k_1, v_1), \dots, (k_{|\J_i|}, v_{|\J_i|})\}$ is itself a multiset of key-value pairs.

To compress, a depth-first approach is taken.
First, some $\J \in \M$ is sampled without replacement using an ANS decode.
Key-value pairs are then sampled from $\J$, also without replacement, using ANS decode, and compressed to the ANS state until $\J$ is depleted.
This procedure repeats, until the outer multiset $\M$ is empty. Assuming all maps are unique, the maximum number of savable bits is
\begin{equation}\label{eq:nested-savings}
    \log |\M|! + \sum_{i=1}^{|\M|} \log |\J_i|!.
\end{equation}
The collection of JSON maps is composed of public GitHub user data taken from a release of the Zstandard project\footnote{\url{https://github.com/facebook/zstd/releases/tag/v1.1.3}}.
All key-value elements were cast to strings, for simplicity, and are encoded as UTF-8 bytes using a uniform distribtuion.
\Cref{fig:roc-mnist-jsonmaps} shows the number and percentage of saved bits.
The theoretical limit curve shows the maximum savable bits with nesting, i.e.\ \Cref{eq:nested-savings}.
Note that, without nesting, the theoretical limit would be that of MNIST (i.e. $\log |\M|!$).
The method gets very close to the maximum possible savings, for various numbers of JSON maps.
The rate savings were small, but these could be improved by using a better technique to encode the UTF-8 strings.

Assuming $\J$ represents the JSON map in $\M$ with the largest number of key-value pairs, and that the time complexity of comparing two JSON maps is \(\Omega(|\J|)\), the complexity of the four BST operations for the \emph{outer} multiset is \(\Omega(|\M| \cdot |\J| \cdot \log|\M|)\).
The overall expected time complexity for both encoding and decoding is therefore
\begin{align}
    \Omega(|\M| \cdot |\J| \cdot (\log|\M| + \log|\J|))
\end{align}
It is possible to reduce this by performing the inner sampling steps in parallel, or by speeding up the JSON map comparisons.

\subsection{Lossless Neural Compression on Binarized MNIST}
Originally, BB-ANS was introduced as an entropy coder for latent variable models (LVM), such as Gaussian mixtures, that are learned from source samples (see \cite{townsend2020a} and \Cref{chapter:bits-back-lvms} for a detailed discussion).
In this section, we use the original implementation of BB-ANS in-place of the codec for $P_X$.
The average bit-length achieved by BB-ANS is an upper-bound on the cross-entropy between the LVM and the true source distribution, and is equal to the Negative Evidence Lower Bound \cite{townsend2020a}.
We used the pre-trained model and code made publicly available by the author\footnote{\url{https://github.com/bits-back/bits-back}}.

First, an ANS decode is performed to select a binarized MNIST image for compression.
Then, BB-ANS is applied as described in the experimental section of \cite{townsend2020a}.
This process is repeated, until all $10,000$ images are compressed.
We compared the average bit-length with and without the invertible sampling step.
Since the images are all unique, the maximum theoretical savings is $\log(10,000!) \text{ bits} \approx 14\ \text{kB}$.
This represents a potential savings of $7.6\%$, which is achieved by our method, at the cost of only $10\%$ extra computation time on average.

\section{Discussion}
As discussed in \cite{varshney2006}, a sequence can be viewed as a multiset (frequency count of symbols) paired with a permutation (the order of symbols in the multiset).
In this sense, our method implicitly pairs a multiset with a permutation that is incrementally built through sampling without replacement.
Through this pairing, we convert a computationally intractable problem (compressing a multiset) into a tractable one (compressing a sequence), while still achieving the optimal rate.
This general technique of making a compression problem easier, by augmenting the data in some way, is known as bits-back coding (see Related Work section). Other fruitful applications of bits-back coding may exist which have yet to be discovered.

The stack-like nature of ANS precludes random access to symbols in the compressed multiset, as well as streaming multiset communication, because the full multiset must be known before encoding can begin.
Methods may exist which allow streaming and/or random access, and this may be another interesting research direction.

We see a number of potential generalizations of the method we have presented.
Firstly, as suggested in the introduction, the independence assumption can be relaxed: the method can also be used with an adaptive sequence compressor, as long as the implied model over symbols is \emph{exchangeable} (that is, the probability mass function is invariant to permutations of its inputs).
This condition is equivalent to the adaptive decoder not depending on the order of previously observed symbols, i.e.,\ at each decoding step the symbol codec only has access to the \emph{multiset} of symbols observed so far.

Two examples of distributions over multisets in which symbols are exchangeable but not i.i.d.\ are given in \cite{Steinruecken2016-oy}.
The first is multisets of symbols drawn i.i.d.\ from an \emph{unknown} distribution, where the distribution itself is drawn from a Dirichlet prior.
The unknown distribution is effectively learned during decoding.
The second is uniformly distributed K-combinations (or submultisets) of some fixed ambient multiset.
One way to generate a uniformly distributed K-combination is to sample elements without replacement from the ambient multiset.
It is possible to efficiently compress K-combinations using an extension of the method in this paper, taking advantage of the fact that the ambient multiset is static and avoiding materializing it explicitly.
We leave more detailed discussion of this to future work.

As well as using the method with adaptive codecs, it would also be interesting to explore applications to more elaborate multiset-like structures.
For example, multigraphs, hypergraphs, and more sophisticated tree structured files, such as the JSON example in the experiments section, which may be represented as nested multisets all the way down.
\chapter{Random Cycle Coding}\label{chapter:rcc}
A \emph{clustering} is a collection of pairwise disjoint sets, called \emph{clusters}, used throughout science and engineering to group data under context-specific criteria.
A clustering can be decomposed conceptually into two parts of differing nature.
The \emph{data set}, created by the set union of all clusters, and the \emph{assignments}, indicating which elements belong to which cluster.
This chapter is concerned with the \emph{lossless} communication and storage of the assignment information, for arbitrary data sets, from an information theoretic and algorithmic viewpoint.

Communicating clusters appears as a fundamental problem in modern vector similarity databases such as FAISS \cite{johnson2019billion}.
FAISS is a database designed to store vectors of large dimensionality, usually representing pre-trained embeddings, for similarity search.
Given a query vector, FAISS returns a set of the $k$-nearest neighbors \cite{lloyd1982least} available in the database under some pre-defined distance metric (usually the L2 distance).
Returning the exact set requires an exhaustive search over the entire database for each query vector which quickly becomes intractable in practice.
FAISS can instead return an approximate solution by performing a two-stage search on a coarse and fine grained set of database vectors.
The database undergoes a training phase where vectors are clustered into sets and assigned a representative (i.e., a centroid).
FAISS first selects the $k^\prime$-nearest clusters, $k^\prime < k$, based on the distance of the query to the centroids, and then performs an exhaustive search within them to return the approximate $k$-nearest neighbors.

The cluster assignments must be stored to enable cluster-based approximate searching.
In contrast to a class, a cluster is distinguishable only by the elements it contains, and is void of any labelling.
However, cluster assignments are often stored alongside the data set in the form of artificially generated labels.
Current numbers \cite{chen2010approximate, martinez2016revisiting, babenko2014additive, jegou2010product, huijben2024residual} suggest that labelling, for the sake of clustering, can represent the majority of bits spent for communication and storage in a typical use case, and will become the dominating factor as the performance of lossy compression algorithms, used to store the high-dimensional vectors, improves.


In this chapter we show how to communicate and store cluster assignments without creating artificial labels, providing substantial storage savings for vector similarity search applications.
Assignments are implicitly represented by a cycle of a permutation defined by the order between encoded elements.
Our method, \emph{Random Cycle Coding} (RCC), uses bits-back coding \cite{townsend2019practical} to pick the order in which data points are encoded.
The choice of orderings is restricted to the set of permutations having disjoint cycles with elements equal to some cluster.
RCC is optimal for distributions over assignments that assign probability proportional to the product of cluster sizes.
In the sense given by \cite{severo2023compressing}, it achieves the Shannon bound \cite{cover1999elements} in bit savings for the class of bits-back algorithms.
The worst-case computational complexity of RCC is quasi-linear in the largest cluster size and requires no training or machine learning techniques.

\section{Problem Setting}
Let $X^n = (X_1, \dots, X_n)$ be a sequence of independent random variables $X_i$ with common, but arbitrary, alphabet $\X$.
Throughout we assume that a \emph{total ordering} can be defined for $\X$, i.e., elements of the set can be compared and ranked/sorted according to some predefined criteria (e.g., lexicographical ordering).
We assume no repeats happen in the sequence.
This is motivated by applications where elements are high-dimensional vectors such as embeddings or images where repeats are unlikely to happen.

We are interested in the setting where the elements of $X^n$ are grouped into a random clustering $\Pi$  (\Cref{def:clustering}).
The order between elements in a cluster is irrelevant and clusters are void of labels.
Conditioned on the sequence $X^n=x^n$ the clustering $\Pi$ defines a partition of the data set $\SetDataSet = \{x_1, \dots, x_n\}$.
The objective is to design a lossless source code, with implied probability model $Q_{\Pi \g \SetDataSet}$, for the assignments $\Pi$ that can be used alongside any codec for the data set $\SetDataSet$.
In what follows, we first describe the coding procedure defining the model $Q_{\Pi \g \SetDataSet}$.
We then show the model resulting from our procedure assigns probability proportional to the product of cluster sizes.

\section{Method}
\begin{figure}[t]
    \centering
    \includegraphics[width=0.8\textwidth]{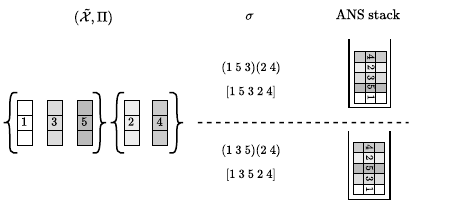}
    \caption
    {
        High-level description of our method, Random Cycle Coding (RCC).
        RCC encodes the clustering $(\tilde{\X}, \Pi)$ as cycles in the permutation $\sigma$ induced by the relative ordering of objects.
        \textbf{Left:} Indices represent the rankings of objects in $\tilde{\X}$ according to the total ordering of $\X$.
        \textbf{Middle:} One of two permutations, in Foata's space, will be randomly selected with bits-back coding to represent the clustering.
        \textbf{Right:} The objects of $\tilde{\X}$ are encoded according to the selected permutation, implicitly encoding the clusters $\Pi$ in the permutation's disjoint cycles.}
    \label{fig:method-rcc}
\end{figure}
Our strategy will be to send the elements of $\SetDataSet$ in a particular ordering such that it implicitly encodes the clustering information.
To achieve this, we view cluster assignments as partitions of the integer interval $[n]$.
Each ordering of the data set will induce a cluster assignment via the cycles of an associated permutation.

We associate a permutation $\sigma_{x^n}$ to each of the possible $n!$ orderings of $\SetDataSet$ based on sorting.
Define $s^n$ to be a reference sequence created by sorting the elements in $\SetDataSet$ according to the total ordering $\leq$ of $\X$,
\begin{align}
    s_i \in \SetDataSet \quad\text{such that}\quad s_1 < s_2 < \dots < s_n.
\end{align}
Let $\SetOrders$ be the set of all possible orderings of $\SetDataSet$ created from the action of permutations in $\Sym_n$ on the reference sequence,
\begin{align}\label{eq:set-orders}
    \SetOrders = \left\{(s_{\sigma(1)}, \dots, s_{\sigma(n)}) \colon \sigma \in \Sym_n\right\}.
\end{align}
For any $x^n \in \SetOrders$, the induced permutation $\sigma_{x^n}$ is defined as that which permutes the elements of the reference $s^n$ such that $x^n$ is obtained,
\begin{align}
    \left(s_{\sigma_{x^n}(1)}, \dots, s_{\sigma_{x^n}(n)}\right) = (x_1, \dots, x_n).
\end{align}
Under this definition the induced permutation can also be constructed by directly substituting $x_i$ for its \emph{ranking} in $\SetDataSet$.
The following equivalence relation helps us establish the relationship between permutations and cluster assignments.
\begin{definition}[Cycle Equivalence]\label{def:cycle-equivalence}
    Let $\SetDataSet$ be an arbitrary set with orderings $\SetOrders$ as in \Cref{eq:set-orders}.
    Two sequences $x^n, y^n \in \SetOrders$ are equivalent, written $x^n \sim y^n$, if the disjoint cycles of their induced permutations contain the same elements.
\end{definition}
\begin{example}[Cycle-Equivalence]\label{example:cycle-equivalence}
    Let $\SetDataSet = \{2, 4, 6, 8\}$ under the usual ordering for natural numbers.
    Sequences,
    \begin{align}
        x^n=(4,6,2,8) \quad\text{and}\quad z^n=(6,2,4,8),
    \end{align}
    induce permutations
    \begin{align}
        \sigma_{x^n} = [2,3,1,4] = (4)(1\ 2\ 3) \quad\text{and}\quad \sigma_{z^n} = [3,1,2,4] = (4)(1\ 3\ 2).
    \end{align}
    The sequences are equivalent as the disjoint cycles of the induced permutations contain the same elements.
    For both sequences, elements $2, 4$, and $6$ are in the same cycle, while $8$ is in a cycle of its own.
\end{example}

\begin{definition}[Partitions from Permutations]
    To each $x^n \in \SetOrders$, we associate an induced partitioning of $[n]$ created from the cycles of $\sigma_{x^n}$,
    \begin{align}
        \{c^1_1\ \dots\ c^1_{n_1}\} \{c^2_1\ \dots\ c^2_{n_2}\} \dots \{c^\ell_1\ \dots\ c^\ell_{n_\ell}\}.
    \end{align}
\end{definition}
\begin{example}[Partitions as Permutations]\label{example:permutations-and-partitions}
    The following grid shows 4 partitions (top) and the set of all induced permutations that create them.
    \begin{center}
        \begin{tabular}{c c c c}
            $\{4\}\{2, 6\}\{1, 5, 3\}$ & $\{4, 1\}\{2, 6\}\{3, 5\}$ & $\{1, 2, 3, 4\}\{5\}$ & $\{1\}\{2\}\{3\}\{4\}\{5\}$ \\
            \hline
            $(4)(2\ 6)(1\ 5\ 3)$       & $(4\ 1)(2\ 6)(3\ 5)$       & $(1\ 2\ 3\ 4)(5)$     & $(1)(2)(3)(4)(5)$           \\
            $(4)(2\ 6)(1\ 3\ 5)$       &                            & $(1\ 2\ 4\ 3)(5)$     &                             \\
                                       &                            & $(1\ 3\ 2\ 4)(5)$     &                             \\
                                       &                            & $(1\ 3\ 4\ 2)(5)$     &                             \\
                                       &                            & $(1\ 4\ 2\ 3)(5)$     &                             \\
                                       &                            & $(1\ 4\ 3\ 2)(5)$     &
        \end{tabular}
    \end{center}
\end{example}
With this definition we can group orderings into equivalent sets, where all orderings in the same set induce the same partitioning of $[n]$.
This is known as the \emph{quotient set}, with elements called \emph{equivalence classes},
\begin{align}
    \SetOrders/{\sim} \defeq \left\{\mathcal{E} \colon  x^n, y^n \in \mathcal{E} \text{ if, and only if, } \sigma_{x^n} \sim \sigma_{y^n}, \text{ for all } x^n, y^n \in \SetOrders \right\}.
\end{align}
From \Cref{def:clustering} we know each cluster assignment is equivalent to a partition on $[n]$.
We can therefore interpret each equivalence class as a cluster assignment and vice-versa.
\begin{definition}[Cluster Assignments As Equivalence Classes]
    Let $(\SetDataSet, \Pi)$ be an arbitrary finite set with cluster assignments, and $\sim$ the cycle-equivalence relation (\Cref{def:cycle-equivalence}).
    Then $\Pi$ can be redefined, without loss of generality, as a random variable taking on values in the quotient set $\quotient{\SetOrders}{\sim}$, where all orderings of $\SetDataSet$ that induce the same partitioning/assignments belong to the same equivalence class.
\end{definition}
\begin{example}[Cluster Assignments as Equivalence Classes]
    For \Cref{example:cycle-equivalence}, the assignments, viewed as an equivalence class, is equal to $\Pi$ = $\{x^n, z^n\}$, as there are no other permutations over $\SetDataSet$ that are equivalent to the two shown.
\end{example}
\begin{example}[Cluster Assignments as Equivalence Classes - 2]
    For \Cref{example:permutations-and-partitions}, if $\SetDataSet = [5]$, then the set of permutations underneath each partition is the equivalence class interpretation of the assignments.
\end{example}

The size of an equivalence class $\Pi \in \quotient{\SetOrders}{\sim}$ can be computed by counting permutations with the same elements in each cycle.
Given a permutation, shifting one cycle left or right, with wrap-around, does not change the cycle.
To generate a new permutation, we can fix one element of a cycle and permute the rest.
For a cluster assignment with $\ell$ clusters and $n_i$ elements in each cluster, the logarithm of the size of the equivalence class is therefore,
\begin{align}\label{eq:logpi}
    \log\abs{\Pi} = \sum_{i=1}^\ell \log\left((n_i-1)!\right).
\end{align}
Similar to Random Order Coding, Random Cycle Coding uses bits-back to select a representative of the equivalence class to be encoded in place of $\Pi$, which is composed of all elements of the data set $\SetDataSet$.
To achieve the wanted optimality, RCC must remove exactly $\log\abs{\Pi}$ bits from the ANS state while encoding the representative sequence.
At each step, an element from $\SetDataSet$ is selected, using an ANS decode operation, and then encoded using an arbitrary codec (which we will call the \emph{symbol codec}).
Interleaving decoding/sampling and encoding avoids the \emph{initial bits} issue \cite{townsend2019practical} resulting from the initially empty ANS state as discussed in \Cref{sec:bbans-initial-bits}.

Random Order Coding \cite{severo2023compressing} (ROC) performs a similar procedure for multiset compression where elements are also sampled without replacement from $\SetDataSet$.
However, there the equivalence classes consist of all permutations over $\SetDataSet$, and therefore sampling can be done by picking any element from $\SetDataSet$ uniformly at random.
RCC requires sampling without replacement from $\SetDataSet$ non-uniformly such that the resulting permutation has a desired cycle structure (i.e., such that the resulting sequence is in the correct equivalence class).
This can be done with the following extension of Foata's Bijection \cite{foata1968netto}.

\begin{definition}[Foata's Canonicalization]\label{def:foatas-canon}
    The following steps map all sequences in the same equivalence class, $x^n \in \Pi$, to the same \emph{canonical} sequence $c^n \in \Pi$.
    First, write the permutation in disjoint cycle notation and sort the elements within each cycle, in ascending order, yielding a new permutation.
    Next, sort the cycles, based on the first (i.e., smallest) element, in descending order.
\end{definition}

\begin{example}[Foata's Canonicalization]
    The set composed of permutations $\sigma = (3\ 1)(5\ 2\ 4)$ and $\pi = (3\ 1)(2\ 5\ 4)$ is an equivalence class. Applying Foata's Canonicalization to either $\sigma$ or $\pi$ yields $(2\ 4\ 5)(1\ 3)$, which is equal to $\sigma$.
\end{example}

\begin{algorithm}
    \textbf{Inputs:}
    \begin{itemize}
        \item Assignment as a list of lists $\Pi = [[x^1_1, \dots, x^1_{n_1}], [x^2_1, \dots, x^2_{n_2}], \dots, [x^\ell_1, \dots, x^\ell_{n_\ell}]]$ sorted according to Foata's Theorem
        \item Initial ANS state
        \item Symbol codec
    \end{itemize}
    \For{$c = \ell, \dots, 1$}{
    \nl Encode $\left\{x^c_2, \dots, x^c_{n_c}\right\}$ with ROC using the given symbol codec\\
    \nl Encode $x^c_1$ with symbol codec
    }
    \Return Final ANS state
    \caption{Pseudo-code for encoding with RCC.}
    \label{alg:rcc-encode}
\end{algorithm}

\begin{algorithm}
    \textbf{Inputs:}
    \begin{itemize}
        \item Total number of elements $n$
        \item Final ANS state, constructed from \Cref{alg:rcc-encode}
        \item Symbol Codec
    \end{itemize}
    Initialize $\Pi = [\ ], c=0$\\
    \While{total number of elements in $\Pi$ is less than $n$}{
    Decode $x^c_1$ with the symbol codec\\
    Decode elements $x^c_i$ with ROC until an element smaller than $x^c_1$ is seen\\
    Add all decoded elements to $\Pi$ as a list $[x^c_1, \dots, x^c_{n_c}]$\\
    Increment $c$
    }
    \Return $\Pi$, Initial ANS state
    \caption{Pseudo-code for decoding with RCC.}
    \label{alg:rcc-decode}
\end{algorithm}
\paragraph{Algorithm}
RCC encodes a permutation in Foata's space using the procedure outlined in \Cref{alg:rcc-encode}.
The elements of $\mathcal{\tilde{X}}$ are inserted into lists according to their clusterings.
The clustering is canonicalized according to \Cref{def:foatas-canon}.
The encoder starts from the last, i.e., right-most, list.
The list is encoded as a set using ROC, with the exception of the smallest element, which is held-out and encoded last.
This procedure repeats until all lists are encoded.
During decoding the first element is known to be the smallest in its cycle.
The decoder then decodes the remaining cycle elements using ROC, and stops when it sees an element smaller than the current smallest element.
This marks the start of a new cycle and repeats until all elements are recovered.

\paragraph{Savings}
Encoding the smallest value last guarantees that the cycle structure is maintained.
Permuting the remaining elements in the cycle spans all permutations in $\Pi$.
For the $i$-th cluster with $n_i$ elements the savings from encoding $n_i - 1$ elements with ROC is equal to $\log((n_i - 1)!)$.
The total savings is equal to \eqref{eq:logpi}, implying RCC saves $\log\abs{\Pi}$.

\paragraph{Implied Probability Model}
The probability model $Q_{\Pi \g \SetDataSet}$ is indirectly defined by the savings achieved by RCC.
The set of elements and clustering assignments $(\SetDataSet, \Pi)$ are encoded via a sequence $x^n \in \Pi$.
We can assume some lossless source code is used for the data points, requiring $-\log Q_{X^n}(x^n)$ bits to encode the sequence in the large ANS state regime.
The cost of encoding the dataset and cluster assignments equals the cost of encoding a sequence minus the discount given by bits-back,
\begin{align}
    - \log Q_{\SetDataSet, \Pi}(\SetDataSet, \Pi) = - \log Q_{X^n}(x^n) - \log \abs{\Pi}.
\end{align}
From \Cref{eq:multiset-info-content} we know the cost of encoding the set $\SetDataSet$ is that of the sequence minus the cost of communicating an ordering,
\begin{align}
    - \log Q_{\SetDataSet}(\SetDataSet) = -\log Q_{X^n}(x^n) - \log(n!).
\end{align}
From this, we can write,
\begin{align}
    - \log Q_{\Pi \g \SetDataSet}(\Pi \g \SetDataSet)
     & = \log Q_{\SetDataSet}(\SetDataSet) - \log Q_{X^n}(x^n) - \log \abs{\Pi} \\
     & = \log(n!) - \log \abs{\Pi}.
\end{align}
The implied probability model only depends on the cluster sizes, and assigns higher probability when there are few clusters with many elements,
\begin{align}\label{eq:rcc-probability-model}
    Q_{\Pi \g \SetDataSet}(\Pi \g \SetDataSet) = \frac{\prod_{i=1}^k (n_i -1)!}{n!},
\end{align}
where $n_i$ is the size of the $i$-th cluster, and $k$ the total number of clusters.

\paragraph{Complexity}
The complexity of RCC will vary significantly according to the number of clusters and elements.
Initializing RCC requires sorting elements within each cluster, which can be done in parallel, followed by a sort across clusters.
ROC is used as a sub-routine and has both worst- and average-case complexities equal to $\Omega(n_i \log n_i)$ for encoding and decoding the $i$-th cluster.
The total worst- and average-case computational complexities of RCC adapts to the size of the equivalence class:
\begin{align}
    \Omega\left(\sum_i n_i \log n_i\right) = \Omega(\log\abs{\Pi}).
\end{align}
When only one permutation can represent the cluster assignments, i.e., $n = k$, implying $\log\abs{\Pi} = 0$, then RCC has the same complexity as compressing a sequence: $\Omega(n)$.

\section{Related Work}\label{sec:related-work}
To the best of our knowledge there is no other method which can perform lossless compression of clustered high-dimensional data.

ROC is a method to compress multisets of elements drawn from arbitrary sets.
ROC can compress clusterings by viewing them as sets of clusters, but requires encoding the cluster sizes, which can become a significant share of the compression rate.
We compare RCC against the following two variants of ROC next and provide experiments in \Cref{sec:rcc-experiments}.

\paragraph{ROC-1}
The cluster sizes are communicated with a uniform distribution of varying precision and clusters are then encoded into a common ANS state.
Each cluster contributes $\log(n_i!)$ to the bits-back savings, resulting in a reduction in bit-rate of
\begin{align}
    \Gamma_{\text{ROC-1}}
     & = \sum_{i=1}^{k}\left(\log(n_i!) - \log(n-N_i)\right)                 \\
     & = \sum_{i=1}^{k} \log \left(\frac{n_i}{n-N_i} \right) + \log\abs{\Pi} \\
     & \leq \log\abs{\Pi},
\end{align}
where $k$ is the number of clusters, $N_i = \sum_{j=1}^{i-1} n_j$ counts the number of encoded elements before step $i$, and $\log(n - N_i)$ is the cost of encoding the size of the $i$-th cluster.
The gap to optimality increases with the number of clusters, while RCC is always optimal as it achieves $\log\abs{\Pi}$ for any configuration of elements and clusters.

\paragraph{ROC-2}
This variant views the clusterings as a set of sets.
The cluster sizes are communicated as in ROC-1.
However, an extra bits-back step is done to randomly select the ordering in which the $k$ clusters are compressed, resulting in further savings.
The complexity of this step scales quasi-linearly with the number of clusters, $\Omega(k \log k)$, and requires sending the number of clusters ($\log(n)$ bits), which is also the size of the outer set.
The total reduction in bit-rate is
\begin{align}
    \Gamma_{\text{ROC-2}}
    = \Gamma_{\text{ROC-1}} + \log(k!) - \log(n) \leq \log\abs{\Pi}.
\end{align}
This method achieves a better rate than ROC-1, but can require significantly more compute and memory resources due to the extra bits-back step to select clusters compared to both ROC-1 and RCC.
Conditioned on knowing the cluster sizes, ROC-2 compresses each cluster independently.
Intuitively, the method does not take into account that clusters are pairwise disjoint and their union equals the interval of integers from $1$ to $n$, which explains why it achieves a sub-optimal rate savings.

\section{Experiments}\label{sec:rcc-experiments}
\subsection{Minimum and maximum achievable savings}\label{sec:experiments-minmax-savings}
\begin{figure}[t]
    \centering
    \includegraphics[width=\textwidth]{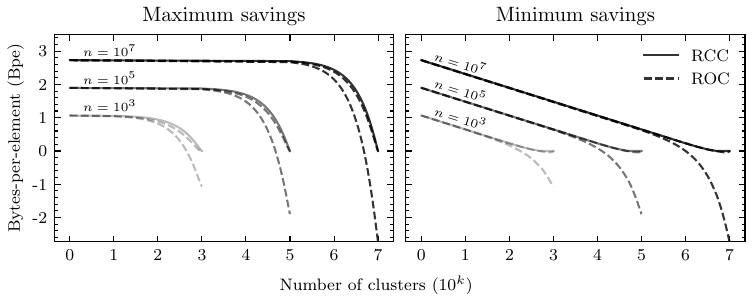}
    \caption{
        Maximum (left) and minimum (right) byte savings per element as a function of the number of clusters and elements.
        Savings are maximized when one cluster contains most elements and all others are singletons.
        The minimum is achieved if clusters have roughly the same number of elements.
        Two variants of Random Order Coding (ROC) \cite{severo2023compressing} are shown (see \Cref{sec:related-work}) with dashed lines.
        Random Cycle Coding (RCC) achieves higher savings than both variants of ROC while requiring less memory and computational resources.
    }
    \label{fig:max-and-min-savings}
\end{figure}
In applications targeted by RCC method (see \Cref{sec:experiments-faiss}) the cluster size is set according to some budget and elements are allocated into clusters via a training procedure.
For a fixed number of elements ($n$) and clusters ($k$) the savings for ROC-1, ROC-2, and RCC will depend only on the cluster sizes $(n_1, \dots, n_k)$.
We empirically analyzed the minimum and maximum possible savings as a function of these quantities.
Results are shown in \Cref{fig:max-and-min-savings}.

The term dominating the bit savings of all algorithms is the product of the factorials of cluster sizes, $\prod_i n_i!$, constrained to $\sum_i n_i = n$ and $n_i \geq 1$.
The maximum is achieved when $n-k$ elements fall into one cluster, $n_j = n - k + 1$, and all others are singletons: $n_i = 1$ for $i \neq j$.
Savings are minimized when all clusters have roughly the same size: $n_i = (n \div k) + \1\{i \leq n \bmod k\}$%
\footnote{$\div$ represents integer division, $n \bmod k$ is the remainder, and $\1\{\}$ is the indicator function that evaluates to 1 if the expression is true.}.

All methods provide similar savings when $k \ll n$.
RCC has better maximal and minimal savings than both ROC-1 and ROC-2 in all settings considered.
The need to encode cluster sizes, without exploiting the randomness of cluster orders as in ROC-2, results in ROC-1 achieving \emph{negative} savings when the number of clusters $k$ is large.
RCC savings converge to $0$ bits as the number of clusters approaches the number of elements, as expected.
As $k$ approaches $n$, ROC-2 also suffers from negative savings, but the values are negligible compared to those of ROC-1.

\subsection{Encoding and decoding times}\label{sec:encoding-decoding-times}
\begin{figure}[t]
    \centering
    \includegraphics[width=\textwidth]{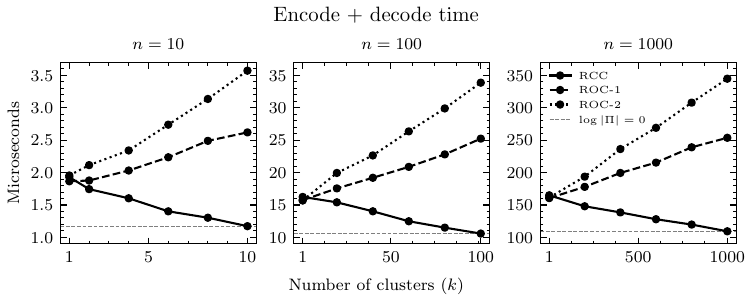}
    \caption{Median encoding plus decoding times, across $100$ runs, for Random Order Coding (ROC) \cite{severo2023compressing} and our method Random Cycle Coding (RCC).
        Clusters are fixed to have roughly the same size, $n/k$, mirroring vector database applications discussed in \Cref{sec:experiments-faiss}.
        Reported times are representative of the amount of compute needed to sample the permutation in the bits-back step as data vectors are encoded with ANS using a uniform distribution.
    }
    \label{fig:times}
\end{figure}
\Cref{fig:times} shows the total encoding plus decoding time as a function of number of elements and clusters.
RCC outperforms both variants of ROC in terms of wall-time by a wide margin, while achieving the optimal savings.
RCC is compute adaptive and requires the same amount of time to encode a sequence when $\log\abs{\Pi}=0$.
The compute required for ROC variants increases with the number of clusters, correlating negatively with $\log\abs{\Pi}$.
ROC-2 is slower than ROC-1 due to the extra bits-back step needed to select clusters with ANS.

\subsection{Inverted-lists of Vector Databases (FAISS)}\label{sec:experiments-faiss}
\begin{table}[t]
    \centering
    \begin{tabular}{rrccccccc}
        \toprule
                                &                                               &                  & \multicolumn{6}{c}{Savings}                                                                                      \\ \cmidrule{4-9}
        Dataset                 & \# Elements                                   & \# Clusters      & Max.                        & Min.   & $\frac{1}{8n}\log\abs{\Pi}$ & RCC             & ROC-2           & ROC-1   \\
        \midrule
        \multirow{6}{*}{SIFT1M} & \multirow{6}{*}{$1{\small,}000{\small,}000$}  & $500$            & $2.31$                      & $1.19$ & $1.20$                      & $\mathbf{0.00}$ & $\mathbf{0.00}$ & $0.04$  \\
                                &                                               & $1000$           & $2.31$                      & $1.06$ & $1.08$                      & $\mathbf{0.00}$ & $\mathbf{0.00}$ & $0.10$  \\
                                &                                               & $4{\small,}973$  & $2.30$                      & $0.77$ & $0.81$                      & $\mathbf{0.00}$ & $0.04$          & $0.86$  \\
                                &                                               & $9{\small,}821$  & $2.29$                      & $0.65$ & $0.71$                      & $\mathbf{0.00}$ & $0.12$          & $2.17$  \\
                                &                                               & $46{\small,}293$ & $2.20$                      & $0.37$ & $0.48$                      & $\mathbf{0.00}$ & $1.28$          & $18.28$ \\
                                &                                               & $95{\small,}284$ & $2.07$                      & $0.24$ & $0.30$                      & $\mathbf{0.00}$ & $2.07$          & $61.43$ \\
        \midrule
        \multirow{6}{*}{BigANN} & \multirow{3}{*}{$1{\small,}000{\small,}000$}  & $1{\small,}000$  & $2.31$                      & $1.06$ & $1.07$                      & $\mathbf{0.00}$ & $\mathbf{0.00}$ & $0.10$  \\
                                &                                               & $10{\small,}000$ & $2.29$                      & $0.65$ & $0.66$                      & $\mathbf{0.00}$ & $0.01$          & $2.27$  \\
                                &                                               & $99{\small,}946$ & $2.06$                      & $0.23$ & $0.25$                      & $\mathbf{0.00}$ & $0.79$          & $76.28$ \\
        \cmidrule{3-9}
                                & \multirow{3}{*}{$10{\small,}000{\small,}000$} & $1{\small,}000$  & $2.73$                      & $1.48$ & $1.49$                      & $\mathbf{0.00}$ & $\mathbf{0.00}$ & $0.01$  \\
                                &                                               & $10{\small,}000$ & $2.72$                      & $1.06$ & $1.07$                      & $\mathbf{0.00}$ & $\mathbf{0.00}$ & $0.14$  \\
                                &                                               & $99{\small,}998$ & $2.70$                      & $0.65$ & $0.66$                      & $\mathbf{0.00}$ & $0.01$          & $2.89$  \\
        \bottomrule
    \end{tabular}
    \vspace{1em}
    \caption{Byte savings, per element, from compressing SIFT1M \cite{jegou2010product} and BigANN \cite{jegou2011searching} as a function of number of elements and clusters.
        Values in columns RCC, ROC-2, and ROC-1, indicate the gap, in percentage (lower is better), to the optimal savings in bytes-per-element, in column $\frac{1}{8n}\log\abs{\Pi}$.
        A value of $0.00$ indicates the method achieves the maximum bit savings shown in column $\frac{1}{8n}\log\abs{\Pi}$.
        Columns Max. and Min. show the theoretical maximum and minimum savings as discussed in \Cref{sec:experiments-minmax-savings}.
    }
    \label{table:faiss-savings-bytes}
\end{table}
\begin{table}[t]
\centering
\begin{tabular}{cccccccc}
\toprule
\multicolumn{1}{l}{} & \multicolumn{1}{l}{}        & \multicolumn{6}{c}{\% Savings}                                 \\
\multicolumn{1}{l}{} & \multicolumn{1}{l}{}        & \multicolumn{3}{c}{Sequential ids} & \multicolumn{3}{c}{External ids} \\
\cmidrule(lr){3-5}
\cmidrule(lr){6-8}
$n$                  & $\frac{1}{8n}\log\abs{\Pi}$ & $4$        & $8$       & $16$      & $4$       & $8$       & $16$     \\
\midrule
$1$M                 & $1.06$                      & $54.8$     & $33.9$    & $19.2$    & $8.9$     & $6.7$     & $4.4$    \\
$10$M                & $1.27$                      & $60.5$     & $38.3$    & $22.1$    & $10.6$    & $8.0$     & $5.3$    \\
$100$M               & $1.48$                      & $65.6$     & $42.4$    & $24.9$    & $12.3$    & $9.3$     & $6.2$    \\
$1$B                 & $1.69$                      & $70.1$     & $46.2$    & $27.5$    & $14.1$    & $10.6$    & $7.0$    \\
\bottomrule
\end{tabular}
\vspace{0.5em}
\caption{
Columns under ``\% Savings" show the savings, in percentage, for the setting of \cite{johnson2019billion} where the number of clusters is held fixed to approximately $\sqrt{n}$.
Savings in bytes-per-element are shown in the second column, where $\log\abs{\Pi} = \sqrt{n}\log((\sqrt{n}-1)!)$, and agree with \Cref{table:faiss-savings-bytes}.
For external ids, $8$ bytes are added to $\frac{1}{8n}\log\abs{\Pi}$ to compute the total size per element, as well as to the cost under RCC.
Meanwhile, $\log(n)$ bits are added to $\frac{1}{8n}\log\abs{\Pi}$ for sequential ids, but not to the cost under RCC, as RCC does not require ids to represent clustering information.
See \Cref{sec:experiments-faiss} for a full discussion.
}
\label{table:faiss-savings-percentage}
\end{table}

We experimented applying ROC and RCC to FAISS \cite{johnson2019billion} databases of varying size.
Results are shown in \Cref{table:faiss-savings-bytes}.
Scalar quantization \cite{cover1999elements, lloyd1982least} was used to partition the set of vectors into disjoint clusters.
This results in clusters of approximately the same number of elements, which is the worst-case savings for both ROC and RCC.
RCC achieves optimal savings for all combinations of datasets, number of elements, and clusters.
ROC-2 has similar performance to RCC but requires significantly more compute as shown in \Cref{fig:times}.

The total savings will depend on the cluster sizes, the number of bytes used to encode each element (i.e., FAISS vector/embedding), as well as the size of id numbers in the database.
Cluster sizes are often set to $\sqrt{n}$ resulting in $\log\abs{\Pi} = \sqrt{n}\log((\sqrt{n} - 1)!)$ \cite{johnson2019billion}.
A vast literature exists on encoding methods for vectors \cite{chen2010approximate, martinez2016revisiting, babenko2014additive, jegou2010product, huijben2024residual}
with typical values ranging from $8$ to $16$ bytes for BigANN and $4$ to $8$ for SIFT1M.
Typically $8$ bytes are used to store database ids when they come from an external source and have semantics beyond the vector database itself.
Alternatively, ids are assigned sequentially taking up $\log(n)$ bits each when their only purpose is to be stored as sets to represent clustering information.
These ids can be removed if vectors are stored with RCC as the clustering information is represented by the relative orderings between objects without the need for ids.
\Cref{table:faiss-savings-percentage} shows savings for RCC for the setting of \cite{johnson2019billion} with $k \approx \sqrt{n}$ clusters.

\section{Discussion}
This chapter provides an efficient lossless coding algorithm for storing random clusters of objects from arbitrary sets.
Our method, Random Cycle Coding (RCC), stores the clustering information in the ordering between encoded objects and does away with the need to assign meaningless labels for storage purposes.
RCC implicitly encodes clustering information in the disjoint cycles of the permutation defined by the relative ordering between data points.
RCC randomly picks the order in which data points are encoded into the ANS state, while simultaneously guaranteeing that each disjoint cycle, of the permutation defined by the ordering, has the same elements as some cluster.

We show that a random clustering can be decomposed into 2 distinct mathematical quantities, the data set of objects present in the clusters (i.e., the union of all clusters), and an equivalence class representing the cluster assignments.
For a given clustering, the equivalence class contains all possible orderings of the data that have cycles with the same elements as some cluster.
The logarithm of the equivalence class size is exactly the number of bits needed to communicate an ordering of the data points, with the wanted permutation cycles, which we refer to as the \emph{order information}.
This quantity was previously defined by \cite{varshney2006toward} as the amount of bits required to communicate an ordering of a sequence if the multiset of symbols was given, and equaled $\log n!$ when there are no repeated symbols.
In the cluster case, the order information $\log\abs{\Pi}$ is strictly less than $\log n!$ as the clustering carries partial information regarding the ordering between symbols in the following way: given the clustering, only orderings with the corresponding cycle structure will be communicated.

The savings achieved by RCC equals exactly the \emph{assignment information} of the data, implying RCC is optimal in terms of compression rate for the probability model shown in \Cref{eq:rcc-probability-model}.
The computational complexity of RCC scales with the number of bits recovered by bits-back, and reverts back to that of compressing a sequence when all clusters are atomic.

The savings for RCC scales quasi-linearly with the cluster sizes, and is independent of the representation size of the data.
The experiments on real-world datasets from vector similarity search databases showcases where we think our method is most attractive: clusters of data requiring few bytes per element to communicate, where the bits-back savings can represent a significant share of the total representation size.
\chapter{Random Edge Coding}\label{chapter:rec}
\emph{Network data}, such as social networks and web graphs \cite{newman2018networks}, can be represented as large graphs, or multigraphs, with millions of nodes and edges corresponding to members of a population and their interactions.
To preserve the underlying community structure, as well as other relational properties, algorithms for compressing and storing network graphs must be \emph{lossless}.

The entropy coding of a network can, in general, be done in quadratic time with respect to the number of vertices by storing binary variables indicating the absence or presence of each possible edge in the graph.
Most real-world networks are \emph{sparse} in the sense that the number of edges $m$ is significantly smaller than the maximum number of possible edges $\binom{n}{2}$ in a simple graph \cite{newman2018networks}.
A graph that is not sparse is known as a \emph{dense} graph.
Algorithms with sub-quadratic complexity in the number of observed edges are more attractive for compressing network data than those that scale quadratically with the number of nodes.
Some networks exhibit \emph{small-world} characteristics where most nodes are not connected by an edge but the degree of separation of any 2 nodes is small \cite{newman2018networks}.
The degree distributions are heavy-tailed due to the presence of hubs, i.e., vertices with a high degree of connectivity.
Random graph models that assign high probability to graphs with small-world characteristics are thus preferred to model and compress these network types.

This chapter presents a lossless source code for large labeled graphs called \emph{Random Edge Coding} (REC).
REC is optimal under a broad class of distributions referred to as \emph{edge-permutation invariant} (\Cref{def:rec-edge-permutation-invariance}) and can achieve competitive performance on real-world networks as we show in \Cref{sec:rec-experiments}.
When paired with Pólya's Urn \cite{mahmoud2008polya}, a parameter-free model described in \Cref{sec:rec-method}, REC requires only integer arithmetic and the worst-case computational and memory complexities scale quasi-linearly and linearly with the number of observed edges in the graph.
REC is applicable to both simple and non-simple labeled graphs, with directed or undirected edges, as well as hyper-graphs.

REC uses \emph{bits-back coding} \cite{frey1996free,townsend2019practical} to sample edges and vertices without replacement from the graph's edge-list, similar to Random Order Coding (ROC) and Random Cycle Coding (RCC).
Sampling is done by decoding from a shared random state, which also stores the final message (i.e. the bits of the graph).

Recent methods for lossless compression of graphs are reviewed in \Cref{sec:rec-related-work}.
In \Cref{sec:rec-method} we discuss parameter-efficient models that yield good probability values on large sparse networks for which REC is optimal.
Together with Pólya's Urn, REC is shown to achieve compression results competitive with the state-of-the-art on real-world network datasets in \Cref{sec:rec-experiments}.

\section{Related Work}\label{sec:rec-related-work}
To the best of our knowledge there is no previous coding method that can scale to large graphs and is optimal for the broad class of edge-permutation invariant (EPI) graphs of \Cref{def:rec-edge-permutation-invariance}.

A number of previous works have presented adhoc methods for lossless compression of large graphs including Pool Compression \cite{yousuf2022pool}, SlashBurn \cite{lim2014slashburn}, List Merging \cite{grabowski2014tight-LM}, BackLinks \cite{chierichetti2009compressing-BL}, and Zuckerli \cite{versari2020zuckerli}.
In sum, these methods attempt to exploit local statistics of the graph edge-list by defining an ordering of the vertex sequence that is amenable to compression.
See \cite{yousuf2022pool} for an overview of the methods.
Re-ordering techniques would yield no effect for EPI models as all permutations of the edge sequence have the same probability, a consequence of \Cref{def:rec-edge-permutation-invariance}.
In \Cref{sec:rec-experiments}, \Cref{table:rec-results} we compare the performance of these methods with that of entropy coding under Pólya's Urn model using our method and show that it performs competitively and can even outperform previous methods on sparser datasets.

Another machine learning method for lossless graph compression is Partition and Code \cite{bouritsas2021partition}.
The method decomposes the graph into a collection of subgraphs and performs gradient descent to learn a dictionary code over subgraphs.
While achieving good compression performance on small graph datasets, it is unclear if these methods can scale to networks with millions of nodes and edges.

Modeling random graphs is a well studied field dating back to the early work of Erdős, Rényi, and Gilbert \cite{erdHos1960evolution} where either graphs with the same number of edges are equally likely or an edge is present in the graph with a fixed probability $p$.
The field has since evolved to include the stochastic block model \cite{holland1983stochastic}, where the edge probability is allowed to depend on its endpoints, as well as its mixed-membership variant \cite{airoldi2008mixed}.
More recently, \cite{caron2017sparse, cai2016edge, crane2018edge} have found some success in modeling real-world network graphs.
These models have been used in a number of applications including clustering \cite{sewell2020model}, anomaly detection \cite{luo2021anomalous}, link-prediction \cite{JMLR:v17:16-032}, community detection \cite{zhang2022node}, and have been extended to model hierarchical networks \cite{dempsey2021hierarchical}.

Our work draws upon a large body of work on the statistical modeling of networks \cite{newman2018networks, bloem2017random, crane2018edge}.
The model used in this work assigns probability to a vertex- or edge-sequence autoregressively.
Neural network models have also been used for autoregressive graph modeling such as \cite{you2018graphrnn, bacciu2020edge, goyal2020graphgen}. See \cite{zhu2022survey} for a survey.
The probability assigned by these models usually depend on the order in which vertices or edges were added to the graph, in contrast to the Pólya's Urn based-model used in this work which is order-invariant.

\section{Problem Setting}
For any $m \in \Naturals$, let $(V_1, \dots, V_{2m}) \sim P_{V^{2m}}$ be a sequence of random variables, with alphabet $\V^{2m}$, representing elements of the vertex set.
We're interested in compressing the graph defined by the random equivalence class with alphabet equal to the quotient set under some equivalence relation between sequences in $V^{2m}$,
\begin{align}
    G = [V^{2m}], \text{ with alphabet } {\V^{2m}}/{\sim}.
\end{align}
The definition of $\sim$ will depend on the type of graph (e.g., simple, directed, undirected) being compressed, as discussed in \Cref{sec:combinatorial-objects-graphs}.
For a fixed instance $G = g$, we want to find an algorithm that approaches the optimal code-length, in the large ANS state regime, for graphs with a large number of edges,
\begin{align}
    \log 1/P_{V^{2m}}(v^{2m}) - \log\abs{\eclass{v^{2m}}}.
\end{align}

\section{Method}\label{sec:rec-method}
\emph{Random Edge Coding} is an optimal lossless code for vertex sequences drawn from a PMF that is invariant to permutations of the edges and of vertices within an edge.
This is characterized formally by the following definition.

\begin{definition}[Edge-Permutation Invariance (EPI)]\label{def:rec-edge-permutation-invariance}
    Let $v^{2m}$ be a vertex sequence with edges defined as $e_i = (v_{2i-1}, v_{2i})$ and $\sigma$ an arbitrary permutation over $m$ elements.
    Given a collection $(\pi_k)_{k=1}^m$ of permutation functions over 2 elements, each over integers $(2 \cdot j-1, 2 \cdot j)$, we say that a PMF $P_{V^{2m}}$ is \emph{edge-permutation invariant} if the following holds
    \begin{align}
        P_{V^{2m}}(e_1, \dots, e_m) = P_{V^{2m}}(\tilde{e}_{\sigma(1)}, \dots, \tilde{e}_{\sigma(m)}),
    \end{align}
    where
    \begin{align}
        \tilde{e}_j = (v_{\pi_j(2\cdot j-1)}, v_{\pi_k(2 \cdot j)}).
    \end{align}
\end{definition}

A stronger property that implies EPI is that of \emph{vertex-permutation invariance}, which coincides with the common definition of finite exchangeability of sequences.

\begin{definition}[Vertex-Permutation Invariance (VPI)]\label{def:rec-vertex-permutation-invariance}
    Let $v^{2m}$ be a vertex sequence and $\pi$ an arbitrary permutation function over $2m$ elements.
    We say that a PMF $P_{V^{2m}}$ is \emph{vertex-permutation invariant} if the following holds
    \begin{align}
        P_{V^{2m}}(v_1, \dots, v_{2m}) = P_{V^{2m}}(v_{\pi(1)}, \dots, v_{\pi(2m)}),
    \end{align}
    for all permutations $\pi$.
\end{definition}

For an arbitrary $k \in \Naturals$, Pólya's Urn (PU) model \cite{mahmoud2008polya} defines a joint probability distribution over a vertex sequence $v^k$ that is VPI.
The generative process of PU is as follows.
An urn is initialized with $\beta_v$ copies of each vertex labeled from $v=1$ to $n$.
At step $i$, a vertex is sampled from the urn, assigned to $v_i$, and then returned to the urn together with an extra copy of the same vertex.
The joint PMF is defined via a sequence of conditional distributions
\begin{align}\label{eq:fhm-conditional}
    P_{V_{i+1} \g V^i}(v_{i+1} \g v^i) \propto d_{v^i}(v_{i+1}) + \beta,
\end{align}
where $d_{v^i}(v) = \sum_{j=1}^i \1\{v = v_j\}$ is the degree of vertex $v$ in $v^i$.

\begin{lemma}[Polya's Urn is VPI]
    The joint PMF of Pólya's Urn is VPI and, therefore, EPI,
    \begin{align}
        P_{V^k}(v^k)
         & = \frac{1}{(n\beta)^{\uparrow k}}\prod_{v \in \V} \beta^{\uparrow d_{v^k}(v)}.
    \end{align}
\end{lemma}
\begin{proof}
    The proof follows from directly computing the joint $P_{V^k}(v^k)$ from the conditionals and showing that it depends on factors that are invariant to permutations of the vertices in $v^k$.
    The normalizing constant for each conditional is
    \begin{align}
        \sum_{v_{i+1} \in \V} \left(d_{v^i}(v_{i+1}) + \beta\right) = i + n\beta.
    \end{align}
    The joint is defined as the product of the conditionals,
    \begin{align}
        \Pr(v^k) = \prod_{i=0}^{k-1} \frac{d_{v^i}(v_{i+1}) + \beta}{i + n\beta}.
    \end{align}
    At step $i$, the generative process appends a vertex to the existing sequence resulting in the product of non-decreasing degrees (plus the bias $\beta$) in the numerator.
    We can regroup the degree terms and rewrite it as a function of the final degree $d_{v^k}$.
    For example, consider the following sequence and its joint distribution
    \begin{align}
        v^k & = 12\ 23\ 21                                    \\
        \Pr(v^k)
            & = \frac{\overbrace{\beta}^{v_1=1}}{n\beta}\cdot
        \frac{\overbrace{\beta}^{v_2=2}}{1 + n\beta}\cdot
        \frac{\overbrace{1 + \beta}^{v_3=2}}{2 + n\beta}\cdot
        \frac{\overbrace{\beta}^{v_4=3}}{3 + n\beta}\cdot
        \frac{\overbrace{2 + \beta}^{v_5=2}}{4 + n\beta}\cdot
        \frac{\overbrace{1 + \beta}^{v_6=1}}{5 + n\beta}      \\
            & = \frac{1}{\prod_{i=0}^{k-1}(i + n\beta)}\cdot
        \underbrace{\beta\cdot(1 + \beta)}_{v_1=v_6=1}\cdot
        \underbrace{\beta\cdot(1 + \beta)\cdot(2 + \beta)}_{v_2=v_3=v_5=2}\cdot
        \underbrace{\beta}_{v_4=3}.
    \end{align}

    In general, the joint takes on the form below,
    where equation \eqref{eq:polya-joint} depends only on the degrees of the vertices in the final sequence $v^k$, guaranteeing VPI.
    \begin{align}\label{eq:polya-joint}
        P_{V^k}(v^k)
         & = \frac{1}{\prod_{i=0}^{k-1}(i + n\beta)}\prod_{v \in \V} (\beta)(1 + \beta)\dots(d_{v^k}(v) - 1 + \beta) \\
         & = \frac{1}{(n\beta)^{\uparrow k}}\prod_{v \in \V} \beta^{\uparrow d_{v^k}(v)}.
    \end{align}
\end{proof}

Although the generative process is fairly simple, the resulting joint distribution can achieve competitive probability values for real-world networks with small-world characteristics as indicated in \Cref{sec:rec-experiments}.
To understand why, note that equation \eqref{eq:polya-joint} assigns higher probability to graphs with non-uniform degree distributions.
The product of factorial terms is the denominator of the multinomial coefficient, and hence is largest when some vertex dominates all others in degree.

PU is parameter-free and therefore requires $0$ bits to be stored and transmitted.
The PMF and CDF of PU, needed for coding with ANS, can be computed with integer arithmetic.
It therefore does not suffer from floating-point rounding errors, which is common in compression algorithms relying on probability estimates \cite{yangneural2023, balle2019integer}.

Compressing the graph with entropy coding requires computing the PMF and CDF of the graph from the PMF of the vertex sequence.
Although the PMF of the graph for PU has a closed form expression, the alphabet size grows exponentially with the number of nodes which makes direct entropy coding infeasible.

To compress, we employ the same strategies as in ROC and RCC, where we map the graph to a sequence of equivalence classes containing vertex sequences.
The vertex equivalence class $\eclass{v^{2m}}$ of a graph $G$ with $m$ edges is the set of all vertex sequences that map to $G$ (see \Cref{sec:combinatorial-objects-graphs}).
PMFs satisfying \Cref{def:rec-edge-permutation-invariance} assign equal probability to sequences in the same equivalence class.
Therefore, the negative log-probabilities of $G, \eclass{v^{2m}}$, and $v^{2m}$ are related by
\begin{align}\label{eq:info-content}
    \log1/P_G(G)
     & = \log 1/P_{V^{2m}}(v^{2m}) - \log\abs{\eclass{v^{2m}}},
\end{align}
for any $v^{2m}$ that maps to graph $G$.
The size of the equivalence class can be computed by counting the number of edge-permutations and vertex-permutations within an edge, which for undirected graphs add up to
\begin{align}\label{eq:savings}
    \log\abs{\eclass{w^{2m}}} = m + \log m!.
\end{align}
The relationship between the probabilities implies we can create a bits-back code for the graph by compressing one of its vertex sequences if we can somehow get a number of bits back equal to $\log\abs{\eclass{v^{2m}}}$.
This leads to the naive \Cref{alg:naive}, which we describe below, that suffers from the initial bits issue (see \Cref{sec:bbans-initial-bits}).

\begin{algorithm}[h]
    \caption{Naive Random Edge Encoder}
    \label{alg:naive}
    \begin{algorithmic}
        \STATE {\bfseries Input:} Vertex sequence $v^{2m}$ and ANS state $s$.
        \STATE 1) Edge-sort the vertex sequence $v^{2m}$
        \STATE 2) Decode a permutation $\sigma$ uniformly w/ prob. $1/\abs{\eclass{v^{2m}}}$
        \STATE 3) Apply the permutation to the vertex sequence
        \STATE 4) Encode the permuted vertex sequence
        \STATE 5) Encode $m$ using $\log m$ bits
    \end{algorithmic}
\end{algorithm}
At step 1) we sort the vertex sequence without destroying the edge information by first sorting the vertices within an edge and then sorting the edges lexicographically.
For example, edge-sorting sequence $(34\ 12\ 32)$ yields $(12\ 23\ 34)$.
In step 2) an index is decoded that corresponds to a permutation function agreed upon by the encoder and decoder, which is applied to the sequence in step 3).
Note these permutations do not destroy the edge information by design.
Finally, in step 4), the permuted sequence is encoded using $P_{V_{i+1} \g V^i}(v_{i+1} \g v^i)$ followed by the number of edges.
Decoding a permutation reduces the number of bits in the ANS state by exactly $\log\abs{\eclass{v^{2m}}}$, while encoding the vertices increases it by $\log1/P_{V^{2m}}(v^{2m})$.
From \eqref{eq:info-content}, the net change, in the large state regime, is exactly the information content of the graph: $\log1/P_G(G)$.

The decoder acts in reverse order and perfectly inverts the encoding procedure, restoring the ANS state to its initial value.
First, $m$ is decoded.
Then the sequence is decoded and the permutation is inferred by comparing it to its sorted version.
Finally, the permutation is encoded to restore the ANS state.

Unfortunately, this method suffers from the initial bits problem, as the decode step happens before encoding, implying there needs to be existing information in the ANS state for the bit savings to occur.
It is possible to circumvent this issue by incrementally sampling a permutation, similarly to Random Order Coding and Random Cycle Coding.
This yields \Cref{alg:bb-exg}, which we describe below.

\begin{algorithm}[h]
    \caption{Random Edge Encoder}
    \label{alg:bb-exg}
    \begin{algorithmic}
        \STATE {\bfseries Input:} Vertex sequence $v^{2m}$ and ANS state $s$.
        \STATE 1) Edge-sort the vertex sequence $v^{2m}$\\
        \REPEAT
        \STATE 2) Decode an edge $e_k$ uniformly from the sequence
        \STATE 3) Remove $e_k$ from the vertex sequence
        \STATE 4) Decode a binary vertex-index $b$ uniformly in $\{0, 1\}$
        \STATE 5) Encode $e_k[b]$
        \STATE 6) Encode $e_k[1-b]$
        \UNTIL{The vertex sequence is empty}
        \STATE 7) Encode $m$ using $\log m$ bits.
    \end{algorithmic}
\end{algorithm}

REC progressively encodes the sequence by removing edges in a random order until the sequence is depleted.
As before, we edge-sort the vertex sequence in step 1) without destroying the edge information.
Then, in steps 2) and 3), an edge is sampled without replacement from the sequence by decoding an integer $k$ between $1$ and the size of the remaining sequence.
Since the graph is undirected, we must destroy the information containing the order of the vertices in the edge.
To do so, in step 4), we decode a binary index $b$ and then encode vertices $e_k[b], e_k[1-b]$ in steps 5 and 6.
Finally, $m$ is encoded.

The initial bits overhead is amortized as the number of edges grows.
This makes REC an optimal entropy coder for large EPI graphs.

\begin{theorem}[REC Optimality]
    Let $V^{2m}$ be a random vertex sequence representing a random graph $G$ on $m$ edges.
    If $P_{V^{2m}}$ is edge-permutation invariant, for any $m$, then the code-length of $G = g$ under Random Edge Coding approaches optimality as $m \rightarrow \infty$,
    \begin{align}
        \lim_{m \rightarrow \infty} \abs{\ell_g - \log 1/P_G(g)} = 0
    \end{align}
\end{theorem}
\begin{proof}
    Encoding steps add $\log 1/P_{V_{i+1} \g V^i}(v_{i+1} \g v^i)$ bits to the ANS state resulting in a total increase of $\log1/P_{V^{2m}}(v^{2m})$ bits.
    Each decoding operation removes bits from the ANS state and together save $\sum_{i=1}^m (1 + \log i) = \log\abs{\eclass{v^{2m}}}$.
    From \eqref{eq:info-content}, the net change is exactly the information content of the graph: $\log1/P_G(g)$.
    The initial and $\log m$ bits (needed to encode $m$) are amortized as $m \rightarrow \infty$.
    Therefore, the number of bits in the ANS state approaches \eqref{eq:info-content}, which concludes the proof.
\end{proof}
In \Cref{sec:rec-experiments} and \Cref{table:rec-optimality} we show empirical evidence for the optimality of REC by compressing networks with millions of nodes and edges down to their information content under Pólya's Urn model.

\section{Extension to Non-Simple Graphs}\label{sec:method-non-simple}
If the graph is non-simple then the size of the equivalence class $\eclass{w^{2m}}$ will be smaller.
The savings in \eqref{eq:savings} must be recalculated by counting the number of valid permutations of edges, and vertices within an edge, that can be performed on the sequence.
Furthermore, \Cref{alg:bb-exg} must be modified to yield the correct savings.
Handling repeated loops and repeated edges requires different modifications which we discuss below.

Each non-loop edge doubles the size of the equivalence class, while loops do not as the vertices are indistinguishable and thus permuting them will not yield a different sequence.
This can be handled by skipping step 4) and setting $b=0$ if $e_k[0] = e_k[1]$.

In general, the number of possible edge-permutations in a non-simple graph $G$ with $m$ undirected edges is equal to the multinomial coefficient
\begin{align}
    \binom{m}{c_1, c_2, \dots} = \frac{m!}{\prod_{e \in G} c_e!}\leq m!,
\end{align}
where $c_e$ is the number of copies of edge $e$ and $e \in G$ iterates over the \emph{unique} edges in $G_m$.
Equality is reached when there are no repeated edges ($c_e=1$ for all edges).

To achieve this saving, \Cref{alg:bb-exg} must be modified to sample edges uniformly from the graph \emph{without replacement}.
In other words, step 2) is generalized to sample $e_k=e$ with probability $c^k_e/k$, where $c^k_e$ are the number of remaining copies of edge $e$ at step $k$.
The count $c^k_{e}$ is non-increasing for all $e$ due to step 3) which, together with step 2), implements sampling without replacement.
Furthermore, since all edges will eventually be decoded, the product of counts $\prod_{k=1}^m c^k_{e_k}$ contains all terms appearing in the factorial $c_e!$ for all edges $e$.
The saving at each step is $-\log c^k_{e_k}/k$ and together will equal the log of the multinomial coefficient
\begin{align}
    -\log\prod_{k=1}^{m} \frac{c^k_{e_k}}{k} = \log\frac{m!}{\prod_{e \in G} c_e!}.
\end{align}
These modifications together guarantee that REC is optimal for non-simple graphs.

\section{Non-EPI Models}
While REC is only optimal for EPI models, it can still be paired with any probability model over vertex sequences that have well defined conditional distributions.
For models that are not EPI the order of the vertices will affect the probability assigned to the graph.
REC in its current form will discount at most $m + \log m!$ bits (with equality when the graph is simple) and all vertex sequences will have equal probability of appearing.
The selected sequence $v^{2m}$ will be determined by the initial bits present in the ANS stack (see \Cref{sec:bbans-initial-bits}).
The number of bits needed to store the graph (i.e. information content) will therefore be,
\begin{align}
    \log 1/P_{V^{2m}}(v^{2m}) - \left(\tilde{m} + \log \binom{m}{c_1, c_2, \dots}\right),
\end{align}
where $\tilde{m}$ is the number of non-loop edges and $v^{2m}$ the random sequence selected via the sampling-without-replacement mechanism of REC.

\section{Time Complexity}
For a graph with $m$ edges, the worst-case computational complexity of encoding and decoding with REC under Pólya's Urn model is quasi-linear in the number of edges, $\Omega(m\log m)$, while the memory is linear: $\Omega(m)$.

We discuss only encoding with REC as decoding is analogous.
The first step during encoding is to sort the edge-sequence which has worst-case complexity $\Omega(m\log m)$.
Then, the edge-list is traversed and the frequency count of all vertices are stored in a binary search tree (BST) with at most $2m$ elements ($\Omega(m)$ memory).
The BST allows for worst-case look-ups, insertions, and deletions in $\Omega(\log m)$, which are all necessary operations to compute the probability, as well as cumulative probability, used during ANS coding.
Traversing the edge-list, together with the updates to the BST, require $\Omega(m \log m)$ computational complexity in the worst-case.

\section{Experiments}\label{sec:rec-experiments}
\begin{table}
    \centering
    \caption{Optimality of Random Edge Coding (REC) with Polya's Urn.}
    \label{table:rec-optimality}
    \begin{tabular}{rcccccc}
        \toprule
        Network    & Seq. NLL & REC (Ours) & Graph NLL (Optimal) & Gap (\%) \\
        \midrule
        YouTube    & 37.91    & 15.19      & 15.19               & 0.0      \\
        FourSquare & 31.14    & 9.96       & 9.96                & 0.0      \\
        Digg       & 32.67    & 10.62      & 10.62               & 0.0      \\
        Gowalla    & 32.11    & 11.69      & 11.69               & 0.0      \\
        Skitter    & 37.22    & 14.26      & 14.26               & 0.0      \\
        DBLP       & 35.48    & 15.92      & 15.92               & 0.0      \\
        \bottomrule
    \end{tabular}

\end{table}

\begin{table}
    \centering
    \label{table:rec-results}
    \caption{Lossless compression with Random Edge Coding of real-world networks}
    \begin{tabular}{rcccccc}
        \toprule
                            & \multicolumn{4}{c}{Social Networks} & \multicolumn{2}{c}{Others}                                                                  \\
                            & YouTube                             & FourSq.                    & Digg          & Gowalla        & Skitter       & DBLP          \\
        \midrule
        \# Nodes            & 3,223,585                           & 639,014                    & 770,799       & 196,591        & 1,696,415     & 317,080       \\
        \# Edges            & 9,375,374                           & 3,214,986                  & 5,907,132     & 950,327        & 11,095,298    & 1,049,866     \\
        $10^6\times$Density & 1.8                                 & 15.8                       & 19.8          & 50.2           & 7.7           & 20.9          \\
        \midrule
        PU w/ REC           & \textbf{15.19}                      & 9.96                       & 10.62         & 12.19          & 14.26         & 15.92         \\
        Pool Comp.          & 15.38                               & \textbf{9.23}              & 11.59         & \textbf{11.73} & 7.45          & \textbf{8.78} \\
        Slashburn           & 17.03                               & 10.67                      & \textbf{9.82} & 11.83          & 12.75         & 12.62         \\
        Backlinks           & 17.98                               & 11.69                      & 12.56         & 15.56          & 11.49         & 10.79         \\
        List Merging        & 15.80                               & 9.95                       & 11.92         & 14.88          & \textbf{8.87} & 14.13         \\
        \bottomrule
    \end{tabular}

\end{table}
In this section, we showcase the optimality of REC on large graphs representing real-world networks.
We entropy code with REC using Pólya's Urn (PU) model and compare the performance to state-of-the-art compression algorithms tailored to network compression.
We report the average number of bits required to represent an edge in the graph (i.e., bits-per-edge) as is common in the literature.
We use the ANS implementation available in Craystack \cite{craystack}.
Code implementing Random Edge Coding, Pólya's Urn model, and experiments are available at \url{https://github.com/dsevero/Random-Edge-Coding}.

We used datasets containing simple network graphs with small-world statistics such as YouTube, FourSquare, Gowalla, and Digg \cite{rossi2015network} which are expected to have high probability under PU.
As negative examples, we compress Skitter and DBLP networks \cite{snapnets}, where we expect the results to be significantly worse than the state of the art, as these networks lack small-world statistics.
The smallest network (Gowalla) has roughly $200$ thousand nodes and $1$ million edges, while the largest (YouTube) has more than $3$ million nodes and almost $10$ million edges.
The cost of sending the number of edges $m$ is negligible but is accounted for in the calculation of the bits-per-edge by adding $32$ bits.

To compress a graph, the edges are loaded into memory as a list where each element is an edge represented by a tuple containing two vertex elements (integers).
At each step, an edge is sampled without replacement using an ANS decode operation as described in \Cref{alg:bb-exg}.
Encoding is performed in a depth-first fashion, where an edge is encoded to completion before moving on to another.
Then, a vertex is sampled without replacement from the edge and entropy-encoded using \eqref{eq:fhm-conditional}.
The process repeats until the edge is depleted, and then starts again by sampling another edge without replacement.
The process terminates once the edge-list is empty, concluding the encoding of the graph.
Decoding is performed in reverse order and yields a vertex sequence that is equivalent (i.e., maps to the same graph as) the original graph.

We showcase the optimality of REC by compressing real-world graphs to the negative log-probability under the Pólya's Urn (PU) model.
\Cref{table:rec-optimality} shows the negative log-probability (NLOGP) of the vertex sequence and graph under PU.
As discussed in \Cref{subsec:rates-below-entropy-are-achievable} the graph's NLOGP is the value an optimal entropy coder should achieve to minimize the average number of bits with respect to the model for large sequence lengths.
REC can compress the graph to its NLOGP (as indicated by the last column of \Cref{table:rec-optimality}) for all datasets.
Compressing the graph as a sequence of vertices (i.e., without REC) would require a number of bits-per-edge equal to the sequence's NLOGP, which is significantly higher than the NLOGP of the graph as can be seen by the first column of \Cref{table:rec-optimality}.
As these methods evolve to achieve better probability values the compression performance is expected to improve automatically due to the optimality of REC.

In \Cref{table:rec-results} we compare the bits-per-edge achieved by PU using REC with current state-of-the-art algorithms for network data.
PU performs competitively on all social networks and can even outperform previous works on networks such as YouTube \cite{rossi2015network}.
The probability assigned by PU for non-social networks is expected to be low, resulting in poor compression performance, as indicated by the last 2 columns of \Cref{table:rec-results}.
The performance of PU deteriorates as the edge density increases and is visible from \Cref{table:rec-results}.

While the compression with REC is optimal for PU, the final results depend on the probability assigned to the graph under PU, which is why ad hoc methods can achieve better performance.
Nonetheless, the bits-per-edge of PU with REC is close to that of current methods.

\section{Discussion}
In this chapter we developed an algorithm capable of performing lossless coding with ANS of large edge-permutation invariant graphs: \emph{Random Edge Coding} (REC).
We provide an example use case with the self-reinforcing Pólya's Urn model \cite{mahmoud2008polya} which performs competitively with state of the art methods despite having $0$ parameters and being void of floating-point arithmetic.
This optimality implies that the efforts from the graph modeling community in improving the probability of network data under edge-permutation invariant models can be directly translated into an increased performance in lossless compression tasks, as Random Edge Coding can compress graphs to the theoretical optimum (negative log-probability) for any EPI data distribution.

Any model satisfying EPI can be used for optimal compression with REC, including neural network based models.
Learning exchangeable models has been explored in the literature \cite{niepert2014exchangeable, bloem2020probabilistic} but, to the best of our knowledge, using them for compression of graphs and other structured data is an under-explored field.

Pólya's Urn satisfies edge-permutation invariance through the invariance of the PMF to permutations of the vertices, which is a sufficient, but not necessary, condition.
An interesting direction to investigate is if there are similar models that are strictly edge-permutation invariant, that is, the PMF is invariant to permutations of edges and vertices within an edge, but not to permutation of vertices from different edges.
One way to achieve this is through the general framework of \cite{zaheer2017deep, hartford2018deep, meng2019hats}.
Let $\psi\colon [n]^2 \mapsto \mathbb{R}$ and $\Psi\colon \bigcup_{k \in \mathbb{N}} \mathbb{R}^k \mapsto \mathbb{R}^+$ be functions invariant to permutation of their arguments, possibly parameterized by some neural network.
The following joint distribution is clearly EPI,
\begin{align}
    Q_{V^{2m}}(v^{2m}) \propto \Psi(\psi(v_1, v_2), \psi(v_3, v_4), \dots, \psi(v_{2m-1}, v_{2m})).
\end{align}
However, the joint distribution may not be invariant to permutations of vertices between different edges (as intended, making it strictly EPI).
As a concrete example, take
\begin{align}
    \psi(v, w)                     & = \langle\theta_v, \theta_w\rangle \\
    \Psi(\phi_1, \dots, \phi_{2m}) & = \sum_{i \in [2m]} \exp(\phi_i),
\end{align}
where $\theta_v, \theta_w \in \mathbb{R}^\ell$ are embeddings that can be learned and $\langle\cdot,\cdot\rangle$ is the inner-product.

To apply Random Edge Coding, we need to define the conditional distributions
\begin{align}
    Q_{V_{2i}, V_{2i-1} \g V^{2(i-1)}}(v_{2i}, v_{2i-1} \g v^{2(i-1)}) = \frac{Q_{V^{2i}}(v^{2i})}{\sum_{v_{2i}, v_{2i-1}}Q_{V^{2i}}(v^{2i})}.
\end{align}
This model can also be learned via stochastic gradient descent but quickly becomes intractable in the form presented for graphs with millions of edges.

In general, a trade-off exists between the model performance and the complexity required to compute the conditional distributions.
Pólya's Urn model lies on an attractive point of this trade-off curve, but there might exist other methods that perform better without increasing complexity significantly.
We think this is a promising line of work that can yield better probability models for network data and can provide a principled approach to lossless compression of these data types.
For example, the extended PU model assigns a unique bias $\beta_v$ to each vertex $v$, which will not break the edge-permutation invariance,
\begin{align}
    Q_{V_{i+1} \g V^i}(v_{i+1} \g v^i) =\frac{d_{v^i}(v_{i+1}) + \beta_{v_{i+1}}}{i + \sum_{v \in [n]}\beta_v},
\end{align}
with corresponding joint distribution,
\begin{align}
    Q_{V^k}(v^k)
     & = \frac{1}{(\sum_{v \in [n]}\beta_v)^{\uparrow k}}\prod_{v \in [n]} (\beta_v)^{\uparrow d_{v^k}(v)}.
\end{align}
The parameters $\{\beta_v\}_{v \in [n]}$ can be learned via gradient descent methods as the gradient of the joint is easily computable.
However, this model has $n$ parameters, one for each vertex, which would need to be transmitted together with the model depending on how the model generalizes as the network grows.
\chapter{Coding Combinatorial Random Variables with Random Permutation Codes}\label{chapter:crv-rpc}
\chaptermark{Coding CRVs with Random Permutation Codes}
\begin{figure}[t]
    \centering
    \includegraphics[width=0.8\textwidth]{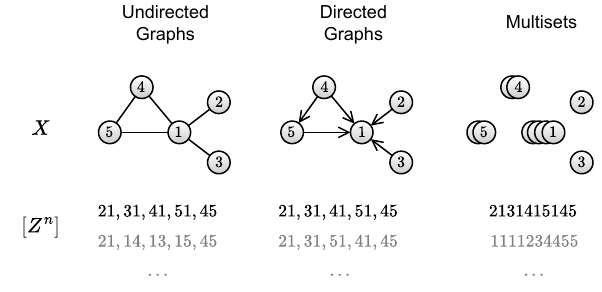}
    \caption{Varying the definition of the equivalence relation $\sim$ yields a different combinatorial random variable $X = [Z^n]$ in each column.}
    \label{fig:crv-examples}
\end{figure}
This chapter concludes the thesis by unifying previous algorithms under a common framework.
The previous chapters described algorithms to compress non-sequential data by encoding data points, in one-of-many possible orderings, with codes from the bits-back family \cite{townsend2019practical,frey1996free}.
The common element between all algorithms is the view of the non-sequential object as an equivalence class, with the equivalence relation changing depending on the type of object being compressed.

In the case of ROC, the equivalence relation groups together sequences that are equal up to a permutation of their elements; the equivalence classes are therefore interpreted as multisets (\Cref{def:multisets}).
Similarly, clusters/partitions for RCC, and all graph types discussed for REC, are defined via equivalence relations which are \emph{finer} (\Cref{def:finer-equivalence-relation}) than the equivalence relation defining multisets, i.e., permutation-equivalence.

This chapter presents a unifying view through \emph{Combinatorial Random Variables} (CRV, \Cref{def:combinatorial-random-variables}); a random variable with alphabet equal to the quotient set under some equivalence relation finer than permutation-equivalence.
Restricting the class of equivalence classes allows for the construction of efficient lossless source codes for CRVs, and allows us to characterize the achievable rates in the large state ANS regime.

\emph{Random Permutation Codes} (RPCs) are presented as a family of algorithms, which include ROC, RCC, and REC, capable of achieving the optimal rates, in the large ANS state regime, for general CRVs.
Finally, the rate achieved by RPCs is shown to be equal to the Negative Evidence Lower Bound \cite{jordan1999introduction}, a well studied quantity in the literature of variational inference.
This establishes the view that RPCs implicitly make use of a latent variable model, achieving the rates given in \Cref{lemma:rate-for-lvm-with-ans}.

\section{Combinatorial Random Variables (CRVs)}
\begin{definition}[Permutation-Equivalence]\label{def:permutation-equivalence}
    Let $\Z^n$ be an arbitrary set of sequences.
    Then, the following binary relation $\sim$ is an equivalence relation (\Cref{def:equivalence-relation}).
    Let $z^n \sim w^n$ if, and only if, there exists a multiset-permutation $\sigma$ on $n$ elements such that,
    \begin{align}
        \left( z_{\sigma(1)}, \dots, z_{\sigma(n)} \right) = (w_1, \dots, w_n).
    \end{align}
    Sequences obeying this relationship are said to be permutation-equivalent.
\end{definition}

\begin{definition}[Combinatorial Random Variables]\label{def:combinatorial-random-variables}
    A combinatorial random variable is defined with respect to a 3-tuple $(\Z, n, \sim)$,
    where $\Z$ is an arbitrary set called the \emph{symbol} alphabet,
    $n \in \Naturals$,
    and $\sim$ an equivalence relation on $\Z^n$ \emph{finer} than permutation-equivalence (\Cref{def:permutation-equivalence}).
    A CRV is any random variable with alphabet equal to the quotient set $\quotient{\Z^n}{\sim}$.
\end{definition}

This definition states the instances of CRVs are \emph{equivalence classes} under the given equivalence relation $\sim$.
The alphabet $\Z$ can be any countable set such as the set of images, all positive integers, or even a quotient set of some other CRV.
Every equivalence class under $\sim$ is a subset of an equivalence class under permutation-equivalence.

A wide range of data types can be recovered by varying the definition of $\sim$, $n$, and $\Z$, including multisets, partitions, graphs, hypergraphs, trees and many others.
We give examples next assuming $\Z$ is an arbitrary set (unless stated otherwise) and $z^n, w^n \in \Z^n$.
More example are shown in \Cref{fig:crv-examples}.

\begin{crv}[Multisets]\label{crv:multisets}
    $\sim$ is the permutation-equivalence relation of \Cref{def:permutation-equivalence}
\end{crv}
\begin{example}[Multiset CRV]
    Let $\Z = \Naturals, n=3$ and $\sim$ the permutation-equivalence relation.
    Then, the set $\{1, 1, 4\}$, expressed in usual set notation, takes on the form of the following equivalence class: $\{114, 141, 411\} \in \quotient{\Z^3}{\sim}$.
\end{example}
\begin{example}[Multiset CRV - 2]
    Let $\Z = \{\blacktriangle, \bigstar, \square\}, n=3$ and $\sim$ the permutation-equivalence relation.
    Then, the alphabet $\quotient{\Z^3}{\sim}$ is shown below alongside the corresponding multiset.
    \begin{center}
        \begin{tabular}{c c c c c c c c c c}
            Multiset                                                &  & Equivalence class in $\quotient{\Z^3}{\sim}$                                                                                                                                                                                       \\
            \cmidrule{1-1} \cmidrule{3-3}
            $\{ \blacktriangle,  \blacktriangle,  \blacktriangle\}$ &  & $\{ \blacktriangle \blacktriangle \blacktriangle\}$                                                                                                                                                                                \\
            $\{ \bigstar ,  \bigstar ,  \bigstar \}$                &  & $\{ \bigstar  \bigstar  \bigstar \}$                                                                                                                                                                                               \\
            $\{ \square ,  \square ,  \square \}$                   &  & $\{ \square  \square  \square \}$                                                                                                                                                                                                  \\
            $\{ \blacktriangle,  \blacktriangle,  \bigstar \}$      &  & $\{ \blacktriangle \blacktriangle \bigstar ,  \blacktriangle \bigstar  \blacktriangle,   \bigstar  \blacktriangle \blacktriangle\}$                                                                                                \\
            $\{ \blacktriangle,  \blacktriangle,  \square \}$       &  & $\{ \blacktriangle \blacktriangle \square ,  \blacktriangle \square  \blacktriangle,   \square  \blacktriangle \blacktriangle\} $                                                                                                  \\
            $\{ \blacktriangle,  \bigstar ,  \bigstar \} $          &  & $\{ \blacktriangle  \bigstar  \bigstar ,  \bigstar  \blacktriangle  \bigstar ,   \bigstar  \bigstar  \blacktriangle\}$                                                                                                             \\
            $\{ \bigstar ,  \bigstar ,  \square \}$                 &  & $\{ \bigstar  \bigstar  \square ,  \bigstar  \square  \bigstar ,   \square  \bigstar  \bigstar \}$                                                                                                                                 \\
            $\{ \bigstar ,  \square ,  \square \}$                  &  & $\{ \bigstar  \square  \square ,  \square  \bigstar  \square ,   \square  \square  \bigstar \}$                                                                                                                                    \\
            $\{ \blacktriangle,  \square ,  \square \}$             &  & $\{ \blacktriangle \square  \square ,  \square  \blacktriangle \square ,   \square  \square  \blacktriangle\}$                                                                                                                     \\
            $\{ \blacktriangle,  \bigstar ,  \square \}$            &  & $\{ \blacktriangle  \bigstar  \square ,  \blacktriangle \square  \bigstar ,   \bigstar  \blacktriangle \square ,   \bigstar  \square  \blacktriangle,   \square  \blacktriangle  \bigstar ,   \square  \bigstar  \blacktriangle\}$
        \end{tabular}
    \end{center}
\end{example}

\begin{crv}[Clusters/Partitions]\label{crv:clusters-partitions}
    $\sim$ is the equivalence relation defined in \Cref{def:cycle-equivalence}, where two sequences are equivalent if the cycles of their induced permutations contain the same elements.
\end{crv}
\begin{example}[Cycle CRV]
    Let $Z = \Naturals$.
    Then, the cycle $(1\ 2\ 1\ 3)$ is associated to the equivalence class $\{1213, 3121, 1312, 2131\} \in \quotient{\Z^4}{\sim}$.
\end{example}

\begin{crv}[Undirected Graphs]\label{crv:undirected-graphs}
    $\sim$ is the equivalence relation defined in \Cref{def:rec-undirected-graphs-equivalence-class}, where two vertex-sequences are equivalent if they are equal up to a permutation of edges or vertices within an edge.
\end{crv}

\begin{crv}[Directed Graphs]\label{crv:directed-graphs}
    $\sim$ is the equivalence relation defined in \Cref{def:rec-directed-graphs-equivalence-class}, where two vertex-sequences are equivalent if they are equal up to a permutation of edges.
\end{crv}

\begin{crv}[Cycles]\label{crv:cycles}
    $z^n \sim w^n$ if, and only if, the elements in one sequence can be shifted by some $\delta \in \Naturals$, with wrap-around, such that the resulting sequence equals the other,
    \begin{align}
        z_i = w_{i + \delta \bmod n}, \text{ for all } i \in [n].
    \end{align}
\end{crv}

\section{Distributions over CRVs}
Specifying a distribution over sequences in $\Z^n$ defines a distribution for the CRV.
It is sufficient to consider distributions over sequences that assign equal mass to equivalent sequences, as the following lemma characterizes.
\begin{lemma}[Sufficiency of Exchangeability]\label{lemma:sufficiency-of-exchangeability}
    Let $X \sim P_X$ be a CRV defined from $(\Z, n, \sim)$.
    For any $P_X$ there exists a distribution $Z^n \sim P_{Z^n}$ that assigns equal mass to equivalent sequences,
    \begin{align}\label{eq:crv-exchangeable}
        P_{Z^n}(z^n) = P_{Z^n}(w^n) \text{ if } z^n \sim w^n.
    \end{align}
\end{lemma}
\begin{proof}
    Given $P_X$, construct the following distribution assigning equal mass to all sequences in the equivalence class,
    \begin{align}
        P_{Z^n}(z^n) = \abs{[z^n]}^{-1} P_X([z^n]).
    \end{align}
\end{proof}
A distribution satisfying \Cref{eq:crv-exchangeable} is said to be \emph{exchangeable} under $\sim$.
Despite this sufficiency result, it can be useful in practice to use non-exchangeable distributions if they offer a computational advantage in computing  and range values required for coding with ANS.

The equivalence class is a deterministic function of the sequence.
Their probabilities are related by,
\begin{align}
    -\log P_{X, Z^n}(x, z^n) & = -\log P_{X}(x) - \log P_{Z^n \g X}(z^n \g x)     \\
                             & = -\log P_{Z^n}(z^n) - \log P_{X \g Z^n}(x \g z^n) \\
                             & = -\log P_{Z^n}(z^n) - \log \delta_{x}([z^n]),
\end{align}
where $\delta_{x}$ is a PMF that places all mass on $x$.
For any $z^n \in x$, this gives,
\begin{align}
    -\log P_{X}(x) = -\log P_{Z^n}(z^n) - \left(-\log P_{Z^n \g X}(z^n \g x)\right).
\end{align}
The information content of the equivalence class is equal to that of the sequence \emph{minus} the \emph{order information}, shown in parenthesis.
The order information was defined in \cite{varshney2006} as the number of bits required to communicate a sequence if the frequency count of symbols, i.e., multiset, was known.
The order information defined here is a generalization to arbitrary combinatorial objects, and falls back to the definition of \cite{varshney2006} when $\sim$ is permutation-equivalence (\Cref{def:permutation-equivalence}).
%
When the sequence distribution $P_{Z^n}$ is exchangeable, then $P_{Z^n \g X}$ is uniform over the equivalence class,
\begin{align}\label{eq:order-information}
    - \log P_{Z^n \g X}(z^n \g x) = \log\abs{x} \leq \log(n!),
\end{align}
with equality only when $\sim$ is permutation-equivalence.

\section{Random Permutation Codes (RPCs)}
The algorithms presented in this thesis, ROC (\Cref{chapter:roc}), RCC (\Cref{chapter:rcc}), and REC (\Cref{chapter:rec}) encode CRVs through the generation of random permutations with bits-back coding.
This section shows a unifying view of these algorithms.
We refer to this family of algorithms as \emph{Random Permutation Codes}.

In practice, a sequence in the equivalence class is compressed to represent the combinatorial object in digital mediums.
For example, in the case of graphs, an edge-list is used to store and communicate the graph structure.
Any choice of lossless code for the elements of the sequence, referred to as the \emph{symbol codec}, defines a joint distribution $Q_{Z^n}$.
The distribution over $Z^n$ defines a distribution for the CRV,
\begin{align}\label{eq:crv-implied-model}
    Q_X(x)
     & \defeq \sum_{z^n \in x} Q_{Z^n}(z^n).
\end{align}
Similar to \Cref{lemma:sufficiency-of-exchangeability}, any rate, measured by the cross-entropy $\E*{-\log Q_X}$, can be achieved by a $Q_X$ induced from a symbol codec $\bar{Q}_{Z^n}$ that is exchangeable under the equivalence relation defining the CRV,
\begin{align}
    \bar{Q}_{Z^n}(z^n) \defeq \frac{1}{\abs{[z^n]}} \cdot \sum_{z^n \in x} Q_{Z^n}(z^n),
\end{align}
\begin{align}
    \E*{-\log Q_X}
     & = - \sum_{x \in \X} P_X(x) \cdot \log Q_X(x)                                         \\
     & = - \sum_{x \in \X} P_X(x) \cdot \log \left(\sum_{z^n \in x} Q_{Z^n}(z^n)\right)     \\
     & = - \sum_{x \in \X} P_X(x) \cdot \log \left(\abs{x} \cdot \bar{Q}_{Z^n}(z^n)\right).
\end{align}

The assumption made in this thesis is that $Q_X$ is not directly available to perform coding with ANS, in the sense that computing the probability and cumulative probability values is not possible.
This can happen, for example, if the size of the equivalence class is large, implying the sum in \Cref{eq:crv-implied-model} is intractable.

To solve this, the family of algorithms presented in this thesis employ a sampling mechanism, defined by a conditional distribution $P_{Z^n \g X}$, which generates a representative sequence of the non-sequential object, $X$, using asymmetric numeral systems (ANS) \cite{duda2009asymmetric}.
The sequence is selected via ANS decodes (\Cref{sec:ans-decode}), where at each step an element of the sequence is sampled and encoded into a common ANS state using a chosen symbol codec.
Tallying up all decode operations decreases the ANS state by the information content $-\log P_{Z^n \g X}(z^n \g x)$ of the sampled sequence $z^n$.
Interleaving decoding and encoding avoids the initial bits issue (\Cref{sec:bbans-initial-bits}) common to this family of coding algorithms.

Random Permutation Codes can be unified under a common mathematical framework, described next.
Given an instance of a CRV, $X = x$, with alphabet $\X = \quotient{\Z^n}{\sim}$, a random sequence $z^n \in x$ is encoded with ANS.
We assume the encoder observes the CRV instance in the form of a \emph{reference sequence}, defined by $s\colon \X \mapsto \Z^n$, where,
\begin{align}
    s(x) \in x \text{ such that } s(x)_1 \leq s(x)_2 \leq \dots \leq s(x)_n,
\end{align}
with $\leq$ an arbitrary total order on $\Z$.

\begin{example}[Reference Sequence - Multisets]
    Let $X$ be a multiset CRV (\Cref{crv:multisets}) with $\Z = [3]$ and $n=3$, and $\leq$ the usual total order on natural numbers.
    Each column below shows an equivalence class in $\quotient{\Z^n}{\sim}$.
    The reference sequence of each equivalence class is in the first row.
    \begin{center}
        \begin{tabular}{cccccccccc}
            $111$ & $222$ & $333$ & $112$ & $113$ & $122$ & $133$ & $223$ & $233$ & $123$ \\
                  &       &       & $121$ & $131$ & $212$ & $313$ & $232$ & $323$ & $132$ \\
                  &       &       & $211$ & $311$ & $221$ & $331$ & $322$ & $332$ & $213$ \\
                  &       &       &       &       &       &       &       &       & $231$ \\
                  &       &       &       &       &       &       &       &       & $312$ \\
                  &       &       &       &       &       &       &       &       & $321$
        \end{tabular}
    \end{center}
\end{example}

At step $i$, the $j_i$-th element, $s(x)_{j_i}$, is selected via an ANS decode and then encoded using the symbol codec into the same ANS state used for decoding.

After sampling an index, $Z_i$ is encoded with the symbol codec adding $-\log Q_{Z^n}(Z^n)$  bits to the ANS state (after $n$ steps) in the large state regime (see \Cref{theorem:source-coding-ans}).
The rate of any code designed for the CRV $X$ will be lower bounded by the cross-entropy between the true data distribution, $P_X$,  and the implied model, $Q_X$.
Decoding to sample the latents decreases the ANS state by exactly the order information $\eqref{eq:order-information}$, which in expectation is equal to $H(P_{Z^n \g X})$.
Due to this, the average change in the ANS state, in the large state regime, is equal to the cross-entropy, establishing the definition of optimality for this class of algorithms.
\begin{theorem}[Optimality of CRV Codes]
    For a CRV $X \sim P_X$, sampling distribution $P_{Z^n \g X}$, and symbol codec $Q_{Z^n}$,
    \begin{align}
        \E*{-\log Q_{Z^n}(Z^n)} - H(P_{Z^n \g X}) = \E{-\log Q_X(X)}.
    \end{align}
    where $Q_{X}(x) \defeq  \abs{x} \cdot Q_{Z^n}(s(x))$.
\end{theorem}
\begin{proof}
    The equivalence class is deterministic conditioned on the sequence,
    \begin{align}
        H(P_{X, Z^n})
         & = H(P_X) + H(P_{Z^n \g X})                         \\
         & = H(P_{Z^n}) + \underbrace{H(P_{X \g Z^n})}_{= 0}.
    \end{align}
    Furthermore,
    \begin{align}
        \KL{P_{Z^n}}{Q_{Z^n}}
         & = \sum_{z^n \in \Z^n} P_{Z^n}(z^n)\cdot\log\left(\frac{P_{Z^n}(z^n)}{Q_{Z^n}(z^n)}\right)                                           \\
         & = \sum_{x \in \X} \sum_{z^n \in x} P_{Z^n}(z^n)\cdot\log\left(\frac{P_{Z^n}(z^n)}{Q_{Z^n}(z^n)}\right)                              \\
         & = \sum_{x \in \X} \sum_{z^n \in x} P_{Z^n}(s(x))\cdot\log\left(\frac{P_{Z^n}(s(x))}{Q_{Z^n}(s(x))}\right)                           \\
         & = \sum_{x \in \X} \abs{x} \cdot P_{Z^n}(s(x))\cdot\log\left(\frac{P_{Z^n}(s(x))}{Q_{Z^n}(s(x))}\right)                              \\
         & = \sum_{x \in \X} \abs{x} \cdot P_{Z^n}(s(x))\cdot\log\left(\frac{\abs{x} \cdot P_{Z^n}(s(x))}{ \abs{x} \cdot Q_{Z^n}(s(x))}\right) \\
         & = \sum_{x \in \X} P_X(x)\cdot\log\left(\frac{P_X(x)}{ Q_X(x)}\right)                                                                \\
         & = \KL{P_X}{Q_X}.
    \end{align}
    Finally,
    \begin{align}
        \E*{-\log Q_{Z^n}(Z^n)}
         & = H(P_{Z^n}) + \KL{P_{Z^n}}{Q_{Z^n}}       \\
         & = H(P_X) + H(P_{Z^n \g X}) + \KL{P_X}{Q_X} \\
         & = \E{-\log Q_X(X)} + H(P_{Z^n \g X})
    \end{align}
\end{proof}

\section{Constructive and Destructive RPCs}
This section discuss an alternative view for the family of Random Permutation Codes.
In this view, the encoder observes an instance of $X=x$ to compress and progressively transforms the sample into a reference $x_0$ known by the decoder.
The steps taken by the encoder are communicated using bits-back coding such that the decoder can reconstruct the sample starting from the reference.
The rate achieved by RPC codes can be shown to equal a negated Evidence Lower Bound \cite{jordan1999introduction} with an exact approximate posterior.

To understand the alternative view, first note that an RPC algorithm necessarily generates a sequence of CRVs during encoding.

\begin{remark}[RPC Generates a Sequence of CRVs]
    Given a CRV $X$, defined by $(\Z, n, \sim)$, let $X_i$ be a random variable defined as the equivalence class of the \emph{remaining} symbols at the $i$-th step of an RPC,
    \begin{align}
        X_i = [(Z_{i+1}, \dots, Z_n)].
    \end{align}
    Let $\sim_i$ be an equivalence relation on $\Z^{n-i}$ such that
    \begin{align}
        (z_{i+1}, \dots, z_n) \sim_i (w_{i+1}, \dots, w_n),
    \end{align}
    if, and only if, $z^n \sim w^n$ for some $z_1, \dots, z_i$ and $w_1, \dots, w_i$.
    Then, $X_i$ is a CRV defined by $(\Z, n-i, \sim_i)$.
\end{remark}

\begin{example}[Multiset CRV Sequence]\label{example:multiset-crv-sequence}
    For a multiset-valued CRV, the following figure shows the possible values for $X_i$, conditioned on $X_4 = \{a, a, b, c\}$.
    The path indicated by the directed edges is that which encodes the sequence $z^4 = abca$.

    \begin{figure}[!h]
        \center
        \begin{tikzpicture}
            \node (root) at (0,0) {$\{a,a,b,c\}$};
            \node (abc) at (-2,-1) {$\{a,b,c\}$};
            \node (aac) at (0,-1) {$\{a,a,c\}$};
            \node (aab) at (2,-1) {$\{a,a,b\}$};
            \node (ab) at (-3,-2) {$\{a,b\}$};
            \node (ac) at (-1,-2) {$\{a,c\}$};
            \node (bc) at (1,-2) {$\{b,c\}$};
            \node (aa) at (3,-2) {$\{a,a\}$};
            \node (a) at (-2,-3) {$\{a\}$};
            \node (b) at (0,-3) {$\{b\}$};
            \node (c) at (2,-3) {$\{c\}$};
            \node (empty) at (0,-4) {$\{\}$};
            \draw[->, line width=1.1pt] (root) -- (abc);
            \draw[->, line width=1.1pt] (abc) -- (ac);
            \draw[->, line width=1.1pt] (ac) -- (a);
            \draw[->, line width=1.1pt] (a) -- (empty);

            \node (X4) at (-4, 0) {$X_4$};
            \node (X3) at (-4,-1) {$X_3$};
            \node (X2) at (-4,-2) {$X_2$};
            \node (X1) at (-4,-3) {$X_1$};
            \node (X0) at (-4.45,-4) {$X_0 = x_0$};

            \foreach \y in {0,...,4}{
                    \draw[dashed, rounded corners] (-3.5,-\y-0.4) rectangle (3.5,-\y+0.4);
                }
        \end{tikzpicture}
    \end{figure}

\end{example}

\begin{theorem}[Markov Property of CRV Sequences]
    Let $(X_1, \dots, X_n) \sim P_{X^n}$ be a CRV sequence generated via an RPC, where $X = X_n$. Then, $X_{i-1}$ is independent of $X_{i+1}, \dots, X_n$ conditioned on $X_i$,
    \begin{align}
        P_{X_{i-1} \g X_i, \dots, X_n} = P_{X_{i-1} \g X_i}.
    \end{align}
\end{theorem}
\begin{proof}
    This proof is a direct consequence of the sufficiency of exchangeability of $P_{Z^n}$ (\Cref{lemma:sufficiency-of-exchangeability}).

    As exemplified in \Cref{example:multiset-crv-sequence}, conditioned on $X=X_n$, there is a bijection between the instances of $Z^{i-1}$ and $(X_i, \dots, X_{n - (i-1)})$.
    This implies the following equality,
    \begin{align}
        P_{X_{i-1} \g X_i, \dots, X_n}(x_{i-1} \g x_i, \dots, x_n)
        = P_{Z_i \g Z^{i-1}}(z_i \g z^{i-1}).
    \end{align}
    However, it is sufficient to know the equivalence class $[Z^{i-1}]$ to compute the probability under $P_{Z_i \g Z^{i-1}}$ due to sufficiency of exchangeability,
    \begin{align}
        P_{Z_i \g Z^{i-1}}(z_i \g z^{i-1}) = P_{Z_i \g Z^{i-1}}(z_i \g w^{i-1}),
    \end{align}
    for any $w^{i-1} \in [z^{i-1}]$.
    Since the equivalence class is exactly $X_i$,  we can write,
    \begin{align}
        P_{X_{i-1} \g X_i}(x_{i-1} \g x_i)
         & =  P_{Z_i \g Z^{i-1}}(z_i \g z^{i-1})                          \\
         & = P_{X_{i-1} \g X_i, \dots, X_n}(x_{i-1} \g  x_i, \dots, x_n).
    \end{align}
\end{proof}

In this alternative view, the symbol codec can be seen as defining a \emph{constructive} conditional distribution $Q_{X_i \g X_{i-1}}$, used for encoding.
An encoder implementing an RPC code samples $X_{i-1}$ from $X_i$ via ANS decoding by selecting a value for $Z_{n-(i-1)}$.
Then, $X_i$ is encoded conditioned on $X_{i-1}$, using $Q_{X_i \g X_{i-1}}$.
The encoder repeats these steps until it reaches the reference.
Since the reference is known to the decoder,
\begin{align}
    -\log P_{X_0 \g X_1}(x_0 \g x_1) = -\log \delta_{x_0}(x_0) = 0,
\end{align}
for all values of $x_1$, implying $0$ bits are required to communicate the reference.

\begin{example}[Constructive/Destructive Multiset RPC]
    Let $X_n=X$ be a multiset-valued CRV with RPC sequence $X^n$ and symbol codec $Q_{Z_i} = Q_Z$ for all $i \in [n]$.
    Let $z_i$ be the only element in the multiset difference $x_i \setminus x_{i-1}$.
    Then, the constructive and destructive conditional distributions of Random Order Coding (\Cref{chapter:roc}) are,
    \begin{align}
        P_{X_{i-1} \g X_i}(x_{i-1} \g x_i) & = \frac{X_i(z_i)}{i} \cdot \1\{ z_i \in X_i \}, \\
        Q_{X_i \g X_{i-1}}(x_i \g x_{i-1}) & = Q_Z(z_i),
    \end{align}
    where $X_i(z)$ is the frequency count of $z$ in $X_i$.
\end{example}

The change in ANS state, in the large state regime, of an RPC code in this alternative view is exactly,
\begin{align}\label{eq:rpc-alternative-rate}
    - \log Q_{X_0} - \sum_{i=1}^n \log\left( \frac{Q_{X_i \g X_{i-1}}}{P_{X_{i-1} \g X_i}} \right),
\end{align}
where $- \log Q_{X_0} = 0$.
This quantity equals the entropy of the source when $Q_{X_i \g X_{i-1}} = P_{X_i \g X_{i-1}}$, as implied by the following lemma.

\begin{lemma}[Constructive-Destructive Decomposition]\label{lemma:constructive-destructive-decomposition}
    The data distribution of a CRV, $P_{X_n} = P_X$, can be written as a ratio of constructive and destructive conditional distributions,
    \begin{align}
        P_X = P_{X_0} \cdot \prod_{i=1}^n \frac{P_{X_i \g X_{i-1}}}{P_{X_{i-1} \g X_i}}.
    \end{align}
    where $X^n$ is a CRV sequence and $X_0 = x_0, P_{X_0} = \delta_{x_0}$, the reference known to both the encoder and decoder.
\end{lemma}
\begin{proof}
    \begin{align}
        P_{X_0, X^n}
         & = P_{X_n} \cdot P_{X_{n-1} \g X_n} \cdots P_{X_0 \g X_1}  \\
         & = P_{X_0} \cdot P_{X_1 \g X_0} \cdots P_{X_n \g X_{n-1}}.
    \end{align}

    \begin{align}
        P_{X_n}
         & = P_{X_0} \cdot \left( \frac{P_{X_1 \g X_0} \cdots P_{X_n \g X_{n-1}}}{P_{X_{n-1} \g X_n} \cdots P_{X_0 \g X_1}}\right)               \\
         & =  P_{X_0} \cdot \left(\frac{P_{X_1 \g X_0}}{P_{X_0 \g X_1}}\right) \cdots \left(\frac{P_{X_n \g X_{n-1}}}{P_{X_{n-1} \g X_n}}\right) \\
         & = P_{X_0} \cdot \prod_{i=1}^n \frac{P_{X_i \g X_{i-1}}}{P_{X_{i-1} \g X_i}}
    \end{align}
\end{proof}

\begin{example}[Constructive/Destructive Graph CRV - Erdős–Rényi]
    Consider an undirected-graph CRV $G$ with the following constructive and destructive process from some RPC, known as the Erdős–Rényi model \cite{erdHos1960evolution}.
    At each step, the encoder removes an edge from the graph and encodes it into the ANS state.
    Let $(E_1, \dots, E_n)$ be the random edge sequence (see \Cref{sec:combinatorial-objects-graphs}) where $G_i$ is the graph on $m$ nodes with edge-set $E(G) \defeq \{E_1, \dots, E_i\}$.
    The conditional distributions of the process are,
    \begin{equation}\label{eq:edge-oriented-edge-prob}
        P_{G_i \g G_{i-1}} = \frac{1}{\binom{n}{2}-(i-1)}\cdot\1\{E(G_i) = E(G_{i-1}) \cup \{ E_i \} \},
    \end{equation}
    which induces a uniform distribution over $E_i$, with alphabet the set of all possible edges $e \notin E(G_{i-1})$ that can be added to the graph.
    The destructive process, where edges are removed, is defined as
    \begin{equation}
        P_{G_{i-1} \g G_i} = \frac{1}{i},
    \end{equation}
    where we have dropped the term $\1\{.\}$ to ease notation.
    Similarly, this defines a uniform distribution over all edges that can be removed from the graph.
    For a fixed integer $m$, the marginal distribution of the first $m$ graphs $G^m = (G_1, \dots, G_m)$ is
    \begin{align}
        \prod_{i=1}^{m} P_{G_i \g G_{i-1}}
         & = \prod_{i=1}^{m} \frac{1}{\binom{n}{2} - (i-1)}          \\
         & = \frac{1}{\binom{n}{2}!/\left(\binom{n}{2} - m\right)!}.
    \end{align}
    The term in the denominator counts the number of ways in which we can permute $m$ out of $\binom{n}{2}$ edges, which makes intuitive sense considering that $G^m$ defines an ordering on selected edges.
    That is, $G^m$ implicitly defines a sequence of edges $e^m$.
    The marginal for a single graph $G_m$ is
    \begin{align}
        P_{G_m}
         & = \frac{P_{G^m}}{\prod_{i=1}^m P_{G_{i-1} \g G_i}}                       \\
         & = \frac{\left(\binom{n}{2} - m\right)!/\binom{n}{2}!}{\prod_{i=1}^m 1/i} \\
         & = \frac{m!\left(\binom{n}{2} - m\right)!}{\binom{n}{2}!}                 \\
         & = \binom{\binom{n}{2}}{m}^{-1},
    \end{align}
    which is a uniform distribution over all edge sets of size $m$.
\end{example}

\begin{remark}[NELBO Rate with Latent Trajectories]
    The generated RPC sub-sequence can be interpreted as a latent trajectory, $T = (X_{n-1}, \dots, X_0)$, generated by the sampling mechanism $P_{X_{i-1} \g X_i}$.
    \Cref{eq:rpc-alternative-rate} can be rewritten to reveal an upper bound on the entropy of the source in the form of a NELBO \cite{jordan1999introduction} (see \Cref{theorem:vae-elbo}),
    \begin{align}
         & - \log Q_{X_0} - \sum_{i=1}^n \log\left( \frac{Q_{X_i \g X_{i-1}}}{P_{X_{i-1} \g X_i}} \right)                                 \\
         & = - \log\left( \frac{Q_{X_n \g X_{n-1}} \cdots Q_{X_1 \g X_0} \cdot Q_{X_0}}{P_{X_0 \g X_1} \cdots P_{X_{n-1} \g X_n}} \right) \\
         & = - \log\left( \frac{Q_{T, X}}{P_{T \g X}} \right),
    \end{align}
    where $X_n = X$%
    \footnote{In the variational inference \cite{jordan1999introduction} literature, $Q_{T \g X}$ would be the default notation to denote the posterior from which latents are generated.
        However, we chose to maintain consistency with earlier chapter of the thesis where $Q$ is reserved for models used for encoding with ANS.}.
\end{remark}

\changed{
    \section{Discussion}
    The equivalence relation in the definition of a CRV is restricted to be finer, or as fine as, permutation-equivalence.
    This allows us to make strong statements regarding the achievable rates of CRVs, as well as develop the family of RPCs that can efficiently achieve these rates.

    A suggestion for future work is to relax this requirement to allow arbitrary equivalence relations.
    This would permit modeling other, more complex, objects such as unlabelled graph, where graphs are equivalent up to a permutation of their vertex labels that preserves the edges in the graph.

    Similarly, another direction is to change the restriction to be an equivalence relation \emph{coarser}, or as coarse as, permutation-equivalence.
    This would mean that the equivalence classes of the CRV can be written as the disjoint union of equivalence classes under permutation-equivalence.
    The advantage of this direction is the potential to re-use RPCs as sub-routines in the to be developed compression algorithms.
    The coarseness, discussed previously, need not be with respect to permutation-equivalence, but instead any of the previously discussed equivalence relations such as those underlying REC or RCC.
    However, it is unclear what type of real-world non-sequential objects could be modeled with this construction.

    Finally, as a last remark on potential future work, this manuscript deals exclusively with lossless source coding, and does not investigate the limits of \emph{lossy} source coding: such as when some distortion can be tolerated in the sequences $Z^n$ or equivalence classes $X$.

}

\appendix
\backmatter
\printbibliography[heading=bibintoc]
\end{document}